\documentclass[aps,prx,reprint,10pt,twocolumn,citeautoscript,superscriptaddress, nofootinbib]{revtex4-2}
\pdfoutput=1

\usepackage{natbib}
\usepackage[english]{babel}
\usepackage{letltxmacro}
\usepackage{latexsym}
\LetLtxMacro{\ORIGselectlanguage}{\selectlanguage}
\makeatletter
\DeclareRobustCommand{\selectlanguage}[1]{%
  \@ifundefined{alias@\string#1}
    {\ORIGselectlanguage{#1}}
    {\begingroup\edef\x{\endgroup
       \noexpand\ORIGselectlanguage{\@nameuse{alias@#1}}}\x}%
}
\newcommand{\definelanguagealias}[2]{%
  \@namedef{alias@#1}{#2}%
}
\renewcommand\onecolumngrid{% <<<<<<
\do@columngrid{one}{\@ne}%
\def\set@footnotewidth{\onecolumngrid}% <<<<<<<<<<<<<<<<
\def\footnoterule{\kern-6pt\hrule width 1.5in\kern6pt}%
}

\renewcommand\twocolumngrid{% <<<<<<
        \def\footnoterule{% restore rule
        \dimen@\skip\footins\divide\dimen@\thr@@
        \kern-\dimen@\hrule width.5in\kern\dimen@}
        \do@columngrid{mlt}{\tw@}
}%
\makeatother
\definelanguagealias{en}{english}
\definelanguagealias{English}{english}
\usepackage{graphicx}
\usepackage{amsmath}
\usepackage{amsthm}
\usepackage{amsfonts}
\usepackage{amssymb}
\usepackage{bm}
\usepackage{color}
\usepackage[percent]{overpic}
\usepackage{soul} 
\usepackage{amssymb}
\usepackage{wasysym}
\usepackage{dsfont}
\usepackage{float}
\usepackage[mathscr]{euscript}
\usepackage{lipsum}
\usepackage{hyperref}
\hypersetup{
    pdfstartview={FitH},    % fits the width of the page to the window
    colorlinks=true,       % false: boxed links; true: colored links
    linkcolor=blue,          % color of internal links (change box color with linkbordercolor)
    citecolor=blue,        % color of links to bibliography
    filecolor=magenta,      % color of file links
    urlcolor=blue           % color of external links
}
\usepackage{xparse} % for differential \differential

\newcommand{\be}{\begin{equation}}
\newcommand{\ee}{\end{equation}}
\newcommand{\bea}{\begin{eqnarray}}
\newcommand{\eea}{\end{eqnarray}}

\newcommand{\pdv}[3][]{\frac{\partial^{#1} #2}{\partial^{#1} #3}}

\NewDocumentCommand\differential{ o g d() }{ % Differential 'd'
	\IfNoValueTF{#2}{
		\IfNoValueTF{#3}
			{\mathrm{d}\IfNoValueTF{#1}{}{^{#1}}}
			{\mathinner{\mathrm{d}\IfNoValueTF{#1}{}{^{#1}}\argopen(#3\argclose)}}
		}
		{\mathinner{\mathrm{d}\IfNoValueTF{#1}{}{^{#1}}#2} \IfNoValueTF{#3}{}{(#3)}}
	}
\NewDocumentCommand\dd{}{\differential} % Shorthand for \differential

\newcommand{\rr}{\mathbb{R}}
\newcommand{\cc}{\mathbb{C}}
\newcommand{\pp}{\mathbb{P}}
\newcommand{\Abb}{\mathbb{A}}
\newcommand{\est}[1]{\hat{#1}^{(N)}}
\newcommand{\estn}[1]{{#1}^{(N)}}
\newcommand{\rmT}{\mathrm{T}}
\newcommand{\rmD}{\mathrm{D}}
\newcommand{\projH}{\mathsf{P}(\mathcal{H})}
\newcommand{\tg}{\tilde{g}}
\newcommand{\tM}{\tilde{M}}

\newtheorem{theorem}{Theorem}

\newtheorem{proposition}{Proposition}

\newtheorem{lemma}{Lemma}
\newtheorem{corollary}{Corollary}
\newtheorem{definition}{Definition}
\newtheorem{fact}{Fact}
\newtheorem{assumption}{Assumption}

\newcommand{\tr}{\mathop{\mathrm{Tr}}}

\newcommand{\re}{\mathop{\mathrm{Re}}}

\newcommand{\sgn}{\mathop{\mathrm{sgn}}}

\usepackage{graphicx}
\usepackage[colorinlistoftodos]{todonotes}
\usepackage{verbatim}
\usepackage[normalem]{ulem}
\usepackage{enumitem}
\usepackage{multirow}

\newcommand{\abs}[1]{\lvert#1\rvert}
\newcommand{\norm}[1]{\lVert#1\rVert}

\newcommand{\Tr}{\mathrm{Tr}}

\DeclareMathOperator{\GL}{GL}
\DeclareMathOperator{\Expect}{\mathbb{E}}

\newcommand{\braket}[2]{\mbox{$ \langle #1 | #2 \rangle $}}

\newcommand{\ket}[1]{\mbox{$ | #1 \rangle $}}
\newcommand{\bra}[1]{\mbox{$ \langle #1 | $}}
\newcommand{\ev}[1]{\langle #1 \rangle}
\newcommand{\dyad}[1]{\rvert #1 \rangle \! \langle #1 \lvert}
\newcommand{\ketbra}[2]{\lvert #1 \rangle \! \langle #2 \rvert}
\newcommand{\expval}[2]{\langle #2 \lvert #1 \rvert  #2 \rangle}
\newcommand{\matrixel}[3]{\langle #1 \lvert #2 \rvert  #3 \rangle}

\usepackage{listings} % python code
\definecolor{codegreen}{rgb}{0,0.6,0}
\definecolor{codegray}{rgb}{0.5,0.5,0.5}
\definecolor{codepurple}{rgb}{0.58,0,0.82}
\definecolor{backcolour}{rgb}{0.95,0.95,0.92}

\lstdefinestyle{mystyle}{
    backgroundcolor=\color{backcolour},   
    commentstyle=\color{codegreen},
    keywordstyle=\color{magenta},
    numberstyle=\tiny\color{codegray},
    stringstyle=\color{codepurple},
    basicstyle=\ttfamily\footnotesize,
    breakatwhitespace=false,         
    breaklines=true,                 
    captionpos=b,                    
    keepspaces=true,                 
    numbers=left,                    
    numbersep=5pt,                  
    showspaces=false,                
    showstringspaces=false,
    showtabs=false,                  
    tabsize=2
}

\newcommand{\equref}[1]{Eq.~\eqref{#1}}

\newcommand{\figref}[1]{Fig.~\ref{#1}}
\newcommand{\supfigref}[1]{Supplementary Fig.~\ref{#1}}
\newcommand{\secref}[1]{Sec.~\ref{#1}}
\newcommand{\refcite}[1]{Ref.~\cite{#1}}

\newcommand{\appref}[1]{Appendix~\ref{#1}}

\newcommand{\wwc}[1]{#1}

\begin{document}

\title{Learning topological states from randomized measurements using \\
variational tensor network tomography}
\date{\today}
\author{Yanting Teng}
\affiliation{Department of Physics, Harvard University, Cambridge, MA 02138, USA}
\affiliation{
\wwc{Institute of Physics, Ecole Polytechnique Fédéderale de Lausanne (EPFL), CH-1015, Lausanne, Switzerland}}
\author{Rhine Samajdar}
\affiliation{Department of Physics, Princeton University, Princeton, NJ 08544, USA}
\affiliation{Princeton Center for Theoretical Science, Princeton University, Princeton, NJ 08544, USA}
\author{Katherine Van Kirk}
\affiliation{Department of Physics, Harvard University, Cambridge, MA 02138, USA}
\author{Frederik Wilde}
\affiliation{Dahlem Center for Complex Quantum Systems, Freie Universit{\"a}t Berlin, 14195 Berlin, Germany}
\author{Subir Sachdev}
\affiliation{Department of Physics, Harvard University, Cambridge, MA 02138, USA}
\author{Jens Eisert}
\affiliation{Dahlem Center for Complex Quantum Systems, Freie Universit{\"a}t Berlin, 14195 Berlin, Germany}
\author{Ryan Sweke}
\affiliation{IBM Quantum, Almaden Research Center, San Jose, CA 95120, USA}
\author{Khadijeh Najafi}
\affiliation{IBM Quantum, IBM T.J. Watson Research Center, Yorktown Heights, NY 10598 USA}
\affiliation{MIT-IBM Watson AI Lab, Cambridge, MA 02142, USA
}

% \email{yanting.teng@epfl.ch}

 \begin{abstract}

Learning faithful representations of quantum states is crucial to fully characterizing the variety of many-body states created on quantum processors. 
While various tomographic methods such as classical shadow and MPS tomography have shown promise in characterizing a wide class of quantum states, they face unique limitations in detecting topologically ordered two-dimensional states. 
To address this problem, we implement and study a heuristic tomographic method that combines variational optimization on tensor networks with randomized measurement techniques. 
Using this approach, we demonstrate its ability to learn the ground state of the surface code Hamiltonian as well as an experimentally realizable quantum spin liquid state. 
In particular, we perform numerical experiments using MPS ans\"atze and systematically investigate the sample complexity required to achieve high fidelities for systems with sizes of up to $48$ qubits. 
In addition, we provide theoretical insights into the scaling of our learning algorithm by analyzing the statistical properties of maximum likelihood estimation. 
Notably, our method is sample-efficient and experimentally friendly, only requiring snapshots of the quantum state measured randomly in the $X$ or $Z$ bases. Using this subset of measurements, our approach can effectively learn any real pure states represented by tensor networks, and we rigorously prove that random-$XZ$ measurements are tomographically complete for such states.
 \end{abstract}

\maketitle

\section{Introduction}
In recent years, quantum simulators have proven instrumental to studying a diverse range of correlated quantum many-body systems, facilitating the investigation of phenomena that are often difficult to probe in conventional solid-state settings~\cite{CiracZollerSimulation}. Today, there exist a host of candidate quantum simulation architectures, including ultracold atoms in optical lattices~\cite{BlochSimulation}, neutral atom arrays~\cite{ebadiQuantumPhasesMatter2021}, trapped ions~\cite{BlattSimulator}, photonic devices~\cite{madsen2022quantum}, and superconducting quantum processors~\cite{KochSimulator}; additionally, many of these platforms can be operated in either analog or digital modes. A common use case for such simulators is preparing the ground (or low-energy) states of some desired Hamiltonian. Importantly, going beyond simple symmetry-breaking ground states, quantum simulation has also enabled the experimental realization of highly entangled states with topological order~\cite{semeghiniProbingTopologicalSpin2021,satzingerRealizingTopologicallyOrdered2021}, characterized by fractionalized excitations and emergent gauge fields~\cite{Wen91,Sachdev92}.  

Tomographic and other recovery methods for extracting information from these experiments are crucial, in particular, when they provide 
detailed diagnostic information about the prepared quantum state~\cite{BenchmarkingReview}.
Indeed, recent approaches using randomized measurements have shown promise for measuring specific properties, such as the R\'enyi entropy~\cite{satzingerRealizingTopologicallyOrdered2021, bluvsteinQuantumProcessorBased2022, brydgesProbingRenyiEntanglement2019, hu2024demonstration}. However, extracting the von Neumann entanglement entropy across arbitrary subsystem sizes using these methods might require multiple copies, substantially more experimental control, or more samples than what is currently available. 
\emph{Quantum state tomography} (QST) is designed to reconstruct a quantum state from informationally complete measurements and is useful for systems where the desired observables are not known \textit{a priori}~\cite{vankirkHardwareefficientLearningQuantum2022} or cannot be measured due to the constraints of the experimental setup. Although QST is generally not scalable~\cite{grossQuantumStateTomography2010, haahSampleOptimalTomographyQuantum2017}, it becomes significantly more effective when applied with priors for particular classes of states. This makes QST a valuable alternative for extracting many properties from such systems, particularly when traditional methods are limited.
Established state tomographic methods like \emph{matrix product state} (MPS) tomography~\cite{cramerEfficientQuantumState2010,  baumgratzScalableReconstructionDensity2013,Wick_MPS} provide a natural description with robust guarantees for one-dimensional systems but fall short in accurately capturing two-dimensional quantum states, especially when they exhibit noninjective properties~\cite{perez-garciaMatrixProductState2007} and complex entanglement patterns~\cite{AreaReview}. While \emph{neural network quantum state tomography}  (NNQST)~\cite{torlaiNeuralnetworkQuantumState2018, carrasquillaReconstructingQuantumStates2019} offers yet another option, its efficacy for topological states also remains unclear due to their intrinsic long-range entanglement  and intricate phase structures across different components of the wavefunction~\cite{huangPredictingManyProperties2020}.

\begin{figure*}[t!]
    \centering
    \includegraphics[scale=0.3]{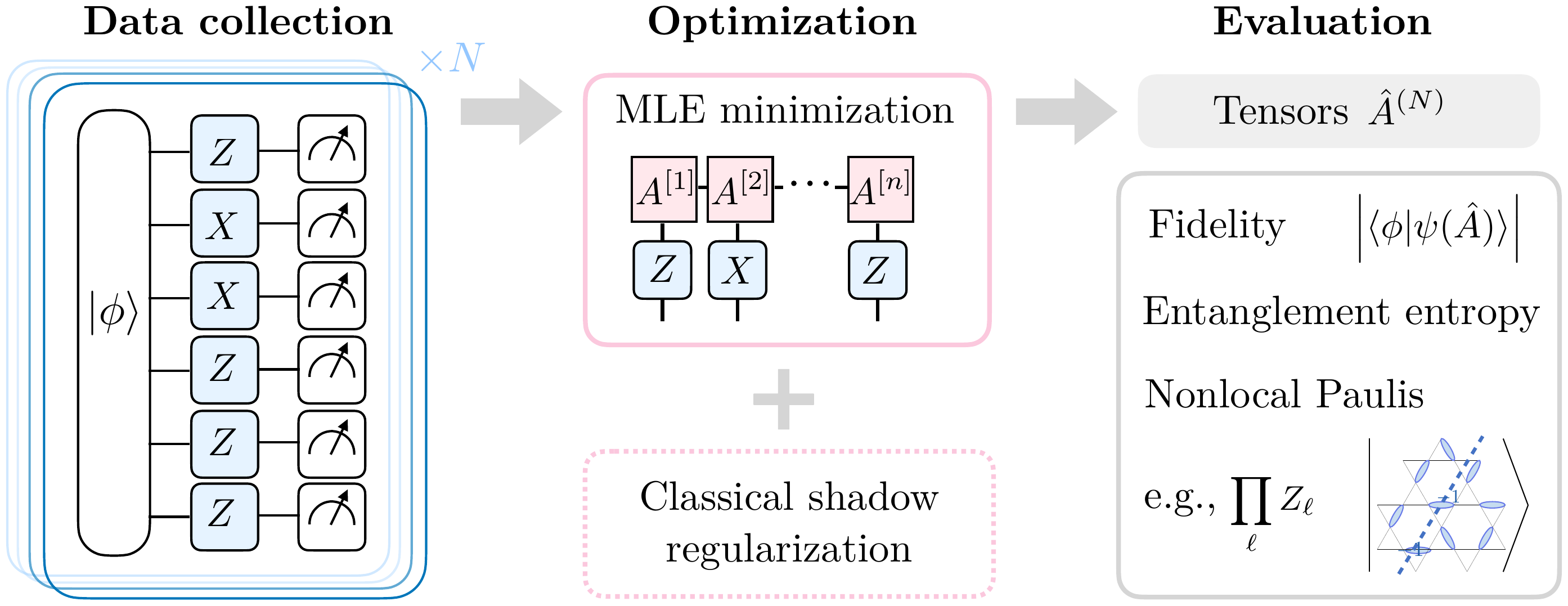}
    \caption{\textbf{Schematic illustration of variational tensor network tomography.} \textbf{Data collection}: A target state vector $\ket{\phi}$ is measured in multiple bases for a collection of $N$ samples defined in ~\equref{eq:data}. For instance, the first panel shows measurements in a random $X$ or $Z$ basis for each qubit. \textbf{Optimization}: A tensor-network model is variationally optimized from $N$ samples using gradient descent, which outputs tensors $\est{A}$ representing the 
    state vector $\ket{\psi(\est{A})}$. The middle panel illustrates an MPS model prior, as numerically investigated in this work. The loss function given by Eq.~\eqref{eq:loss_function_mps} includes the \emph{negative log-likelihood} (NLL), typically used in \emph{maximum likelihood estimation} (MLE). It is optionally regularized by physical observables estimated from the dataset (e.g., using classical shadow tomography represented by the pink dashed box). \textbf{Evaluation}: 
    The algorithm outputs a tensor network or, in the current work, an MPS representation $\est{A}$ as the learned state. We sometimes omit the number of samples and abbreviate the output tensor as $\hat{A}$. In the right panel, we evaluate the performance of our learning protocol by computing the fidelity $F=\abs{\braket{\phi}{\psi(\hat{A})}}$. From the reconstructed state, we can extract arbitrary physical observables, including the von Neumann entanglement entropy and highly nonlocal Pauli strings (the inset shows an example of the nonlocal $Z$ string operators considered in Fig.~\ref{fig:4_spin_liquid}).}
    \label{fig:1_schematics}
\end{figure*}

To address the need for viable tomographic techniques for exotic two-dimensional systems~\cite{semeghiniProbingTopologicalSpin2021}, here, we demonstrate a heuristic state tomography method combining a tensor-network prior with \emph{maximum likelihood estimation}  (MLE)~\cite{wangScalableQuantumTomography2020, gomezReconstructingQuantumStates2022, kurmapuReconstructingComplexStates2022, torlai2023quantum} and randomized measurements~\cite{huangPredictingManyProperties2020, huangProvablyEfficientMachine2021}. 
Previous works have considered learning one-dimensional quantum states from uniformly random single-qubit measurements~\cite{wangScalableQuantumTomography2020} or two-dimensional quantum channels with random \wwc{Pauli} measurements~\cite{torlai2023quantum}. In this work, we focus on two-dimensional real pure states and show that they are completely characterized by only randomized $XZ$ measurements, which randomly measure each qubit in either the $X$ or the $Z$ basis. The states considered here are not only potentially more complex than in previous studies but also easy to implement with high fidelity in experimental systems~\cite{bluvstein2024logical}. Moreover, we prove the tomographic completeness of randomized $XZ$ measurements on the space of pure and real quantum states, allowing the method to be extendable to nonstoquastic real Hamiltonians' eigenstates~\cite{bravyiComplexityStoquasticFrustrationfree2008}. Thus, our method provides a pathway for characterizing two-dimensional states on near-term quantum devices, making use of only easy-to-implement measurements. 
% \wwc{For digitally prepared states, such single-qubit rotations has an almost perfect gate fidelity~\cite{bluvstein2024logical}.}
\wwc{For comparison, we also provide a detailed discussion of related works in Appendix~\ref{app:related_works}.}
\wwc{While our main results assume perfect measurements, in \appref{app:noise}, we review the errors that are important when measuring states in Rydberg atom arrays and numerically investigate the effect of noise.}

In our numerical experiments, we use an MPS ansatz as our model prior. While MPS-based representations of two-dimensional states are inherently limited, our approach sets the stage for higher-dimensional tensor networks, such as \emph{projected entangled pair states} (PEPS), for future investigations. In addition, our method provides a scalable approach that aligns with the existing capabilities of modern simulators, thereby providing a benchmark for more advanced tomographic techniques. 
To supplement our numerical results, we leverage the statistical properties of MLE to provide insights into the sample complexity of our approach. More specifically, we provide a pedagogical overview of standard techniques for analyzing MLE~\cite{neweyChapter36Large1994, wildeScalablyLearningQuantum2022, scholten2018behavior}, and discuss where they fail in this context, due to the gauge freedom in MPS representations of quantum states. We then sketch how these difficulties can be circumvented by considering MLE on a manifold~\cite{hajriMaximumLikelihoodEstimators2017}, and apply these results to the variational manifold given by our MPS prior~\cite{haegeman_geometry_2014}.

The rest of this work is structured as follows. In Sec.~\ref{sec:method}, we present our protocol in conjunction with an in-depth discussion of our state-tomographic measurement schemes. Then, in Sec.~\ref{sec:results}, we demonstrate the effectiveness of our protocol when applied to two topological states: a ground state of the surface code Hamiltonian and a $\mathbb{Z}_2$ quantum spin liquid state of Rydberg atoms arrayed on the ruby lattice. Our results are immediately relevant to current experiments since we only require feasible $XZ$ measurements and consider system sizes achievable with state-of-the-art simulators. Finally, we conclude by discussing the prospects and potential directions for improvement of our methodology in Sec.~\ref{sec:outlook}. 

%%%%%%%%%%%%%%%%%%%%%%%%%%%%%%%%%%%%%%%%%%%%%%%%%%%%%%%%%%%%%%%%%%%%
\section{Variational tensor network tomography}\label{sec:method}
%%%%%%%%%%%%%%%%%%%%%%%%%%%%%%%%%%%%%%%%%%%%%%%%%%%%%%%%%%%%%%%%%%%%
In this section, we introduce our variational \emph{tensor network} (TN) tomography protocol. 
To begin, we briefly review the framework of randomized measurements, wherein randomly chosen unitaries are applied to a target state before measuring it in the computational basis. Then, we utilize this 
data to variationally find the optimal tensor-network approximation of the target state. We do so via maximum likelihood estimation, a standard technique for statistical inference; we also have the ability to regularize the MLE solution with estimations of physical quantities via classical shadows. This is done by adding a term to the MLE cost function. These steps are described below and represented by the schematic in Fig.~\ref{fig:1_schematics}.
%%%%%%%%%%%%%%%%%%%%%%%%%%%%%%%%%%%%%%
\paragraph{Dataset and measurements.}
%%%%%%%%%%%%%%%%%%%%%%%%%%%%%%%%%%%%%%%
We follow the notation commonly used in randomized measurements~\cite{elbenRandomizedMeasurementToolbox2022, huangPredictingManyProperties2020}. Let $\ket{\phi}$ be the unknown $n$-qubit pure state prepared by the quantum simulator, and let $\mathcal{U}$ be the unitary ensemble with a uniform distribution over a set $\{U\}$. We apply a unitary $U \sim \mathcal{U}$ before measuring in the computational basis. 
Performing this process repeatedly, we collect $N$ unitary-bitstring tuples as our dataset  
\begin{equation}\label{eq:data}
\{(U_i,\,b_i)\}_{i=1}^{N},
\end{equation}
where each $(U, b)$ is effectively sampled with probability 
\begin{equation}\label{eq:b_U_probability}
\frac{1}{\abs{\{U\}}} \times \expval{U\dyad{\phi} U^\dagger}{b}.
\end{equation}
% One example of such measurements would be the random Pauli classical shadow protocol, where each qubit is rotated such that it is randomly measured in the $X$, $Y$, or $Z$ basis. 
In this work, motivated by the available prior knowledge of the target state (which we discuss in the next subsection) as well as what measurements are experimentally feasible with high precision~\cite{bluvstein2024logical, bluvsteinQuantumProcessorBased2022}, we consider two measurement schemes: (1) $\mathcal{U}_{\textnormal{Global}}$: measuring every qubit in the $X$ basis or every qubit in the $Z$ basis, and (2) $\mathcal{U}_{XZ}$: randomly measuring each qubit in either the $X$ or the $Z$ basis.
For the ensemble $\mathcal{U}_{\textnormal{Global}}$, the unitaries would be uniformly sampled from
\begin{align}\label{eq:fixed_paulis}
    \left\lbrace
        \prod_i U_{i}^{X}, \, \prod_i U_{i}^{Z}
    \right\rbrace.
\end{align}
The unitary $U_{i}^{a}$ rotates the measurement basis of the $i$-th qubit to the $a = X, Z$ basis.
Note that the most general global controls have been considered in earlier work~\cite{vankirkHardwareefficientLearningQuantum2022}, in which the same SU$(2)$ unitary is applied to each qubit. Here, we only focus on a discrete subset, of Pauli $X^{\otimes n}$ and $Z^{\otimes n}$ measurements, also considered in previous work~\cite{gomezReconstructingQuantumStates2022}. 
For our second ensemble $\mathcal{U}_{XZ}$, the unitaries are uniformly sampled from
\begin{align}\label{eq:xz_paulis}
     \left\lbrace
        \prod_i U_{i}^{a}, \, a = X, Z
    \right\rbrace.
\end{align}

In general, one can exploit priors on the target state to reduce the set of random unitaries required~\cite{hadfieldMeasurementsQuantumHamiltonians2020, vankirkHardwareefficientLearningQuantum2022}. 
For instance, consider learning the eigenstates of 
a real Hamiltonian. These states can be chosen to be real, and random-$XZ$ measurements can learn any property of a state within such a space of real and pure states. This can be formally stated as
follows.

\begin{theorem}[\textbf{Tomographic completeness}]\label{thm:real_pure_states}
    Random-$XZ$ measurements $\mathcal{U}_{XZ}$ are tomographically complete on the space of real and pure states.
\end{theorem}

The intuition here is that when $\ket{\phi}$ is decomposed in terms of the computational basis states, all amplitudes can be learned by measuring every qubit in the computational ($Z$) basis. The (relative) signs can be learned by rotating some qubits to the $X$ basis. We refer interested readers to Appendices~\ref{app:xz_shadow} and \ref{app:reality_random_xz_measurements} for details of the proof. In Appendix~\ref{app:xz_shadow}, we formalize our random-$XZ$ measurement protocol and rigorously derive the associated measurement channel in the framework of classical shadow tomography. Furthermore, in Appendix~\ref{app:reality_random_xz_measurements}, we motivate our restricted random-$XZ$ measurement scheme, proving that it is tomographically complete in the space of real and pure states considered in this work. We note that this scenario is more general than that of a fully positive eigenstate of a stoquastic Hamiltonian, for which measurements in only the computational basis suffice~\cite{bravyiComplexityStoquasticFrustrationfree2008}. 

%%%%%%%%%%%%%%%%%%%%%%%%%%%%
\paragraph{Loss function.}
%%%%%%%%%%%%%%%%%%%%%%%%%%%%
To estimate the target state from the dataset $\{(U_i,\,b_i) \}_{i=1}^N$, we use a tensor-network ansatz $\ket{\psi({A})}$, where a sequence of local tensors is denoted by ${A} = \left({A}^{[1]}, \cdots, {A}^{[n]}\right)$. The advantage of this formulation is that one can efficiently compute the probability of a given measurement outcome in any basis by applying the corresponding unitary to the tensor-network model as shown in the middle panel of Fig.~\ref{fig:1_schematics}. This allows us to evaluate the loss function, whose minimum leads to our estimate of the target state. The loss function is
\begin{align}\label{eq:loss_function_mps}
   L({A}) &= -\frac{1}{N}\sum_{i=1}^{N} \log \abs{\matrixel{b_i}{U_i}{\psi({A})}}^2  + \beta R({A}),
\end{align}
where $R(A)$ denotes the regularizer and its coefficient $\beta$ describes how much we prioritize the regularization over MLE.
The first term in Eq.~\eqref{eq:loss_function_mps} above, the \emph{negative log-likelihood} (NLL), arises from the Kullback–Leibler divergence, which measures the difference between two probability 
distributions. In our case, these two probability distributions are over measurement outcomes in randomly selected bases, defined by the target and the model state, respectively. Minimizing the $N$-sample NLL leads to an MLE solution $\est{A}$. To improve our training, we add the regularization term $R(A)$: it penalizes the differences between our model's prediction of a physical observable $\expval{O}{\psi({A})}$ and the estimation $\ev{O}_{N}$ inferred directly from the dataset. For instance, we can choose $R(A) = \|\expval{O}{\psi({A})} - \ev{O}_{N}\|_2$ with a $k$-body Pauli string $O = X^{\otimes k}$, where the estimation $\ev{X^{\otimes k}}_{N}$ can be accurately determined from our global-$XZ$ measurements. 
Broadly speaking, if we are confident about certain properties of the target state, our model should reflect those properties. The choices of regularization observables, therefore, depend on the systems of interest, which we discuss in Sec.~\ref{sec:results} and Appendix~\ref{app:loss_function}, where we describe the loss function and physical  observables used for our numerical results.

 \begin{figure*}[t!]
    % \centering
    \includegraphics[scale=0.27]{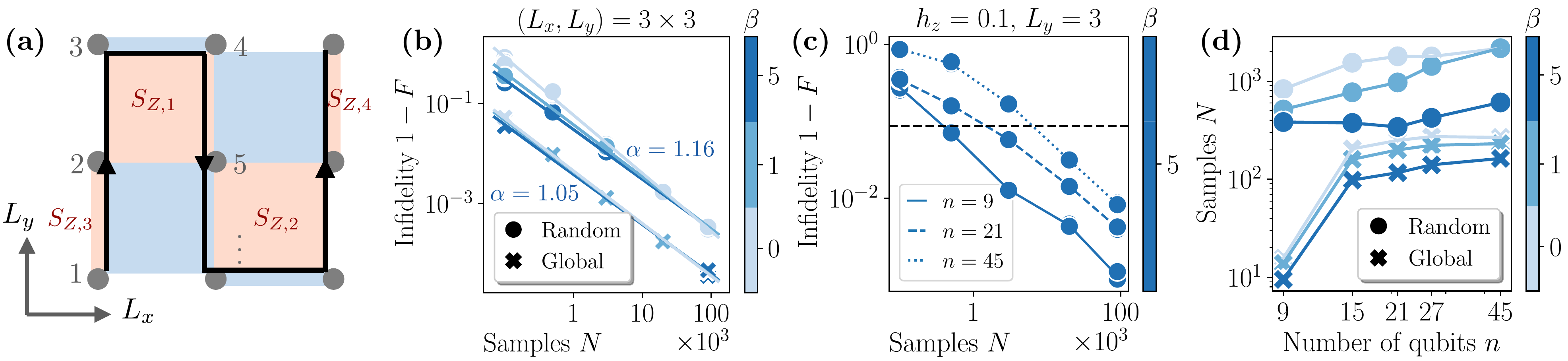}
    \caption{\textbf{Numerical results for the surface code.}  \textbf{(a)}~Schematic representation of the surface code of dimension $(L_x,\, L_y)$ with $Z$ stabilizers (red patches) and $X$ stabilizers (blue patches). The MPS ``snakes'' through the system along the path shown (black arrow). \textbf{(b)}~Scaling of the infidelity with the number of samples $N$ for a $3 \times 3$ surface code with global-$XZ$ (cross) or random-$XZ$
    (circle) measurements. Shades of blue indicate the regularization strength $\beta$.  \textbf{(c)}~To investigate larger system sizes, we focus on the surface code in a strip geometry ($L_y=3$) with a perturbation ($h_z = 0.1$). The three lines indicate the infidelity for different numbers of qubits, $n$, with random-$XZ$ measurements. Dark blue indicates a regularization $\beta=5$. The dashed line marks the fidelity threshold for $n=9$. \textbf{(d)}~Sample complexity \wwc{plotted on a $log$-$log$ scale} to reach a fixed local fidelity threshold $F_{\text{local}}=(0.99)^{n}$ extracted from \textbf{(c)}. \wwc{The number of samples required for achieving such a local fidelity is seen to scale at most polynomially with $n$.} Colors and markers are the same as in \textbf{(b)}.}
    \label{fig:2_surface_code}
\end{figure*}

% This was our paragraph 
\paragraph{MPS prior.} 
 Matrix product states, which are one of the most extensively studied families of tensor networks, have enjoyed remarkable success in accurately representing many-body quantum states. Thanks to an effective variational algorithm known as the \emph{density-matrix renormalization group} (DMRG)~\cite{schollwock2011density}, MPS ans\"atze can be efficiently employed to describe the ground states of various many-body Hamiltonians.  Strictly speaking, two-dimensional systems require the generalization of MPS, as the quasi-1D MPS description is limited to capturing short-range correlations and thus necessarily fails for sufficiently large and entangled systems. However, in practice, MPS ans\"atze are often used to represent ground states even in two dimensions, due to their relative advantage of efficient contractions. For the remainder of this work, we utilize an MPS as our model prior, and we variationally find a state that is most consistent with the observed measurements, as shown in Fig.~\ref{fig:1_schematics}. The optimization of the MPS is carried out by taking gradients simultaneously with respect to all the tensor components; we provide the details in Appendix~\ref{app:training}.

%%%%%%%%%%%%%%%%%%%%%%%%%%%%%%%%%%%%%%%%%%%%%%%%%%%%%%%%
\section{Results}\label{sec:results}
%%%%%%%%%%%%%%%%%%%%%%%%%%%%%%%%%%%%%%%%%%%%%%%%%%%%%%%%

In this section, we demonstrate the effectiveness of our variational tomography protocol for two topologically ordered states. While our eventual goal is to perform tomography on a state in the laboratory, for our results here, we perform numerical ``experiments'' using simulated target states obtained via DMRG, as detailed in Appendix~\ref{app:dmrg_state}. Our protocol assumes that the laboratory state is close to a pure state that could be approximated as an MPS (albeit with a possibly large bond dimension). After obtaining our simulated measurements by sampling the target state, we initiate our protocol with random initial parameters. These parameters are iteratively adjusted until they converge to the trained state. Then, we evaluate the efficacy of this training process by computing the fidelity between the trained and target states.  Here, we analyze $10$ such random initializations, which are shown in our numerical results below. 
Of course,
when using realistic experimental data, one would lack access to the ``ground truth'', i.e., a classical description of the target state, which is needed to compute the fidelity exactly.
In this case, we perform cross validation by comparing converged models on a test dataset [Appendix~\ref{app:training_schemes}], or alternatively, one could consider efficient certifications of high fidelity~\cite{ChoiEndres, huang2024certifying}.

\subsection{Learning the surface code}
In our first demonstration, we consider learning the ground state of the perturbed surface code Hamiltonian~\cite{kitaevFaulttolerantQuantumComputation2003}, $H_{\text{sc}} = - \sum_{l} S_{Z,\, l} - \sum_{l}S_{X,\, l} - h_z \sum_{i}Z_{i}$, defined on the square lattice. Here, $S_{P,\, l}=\prod_{i \in \square_l}P_i$ denotes the $l$-th stabilizer, which is a product of Pauli $P$ operators ($P = X$ or $P = Z$) on a plaquette or boundary [Fig.~\ref{fig:2_surface_code}(a)], and $h_z$ is a perturbative field. Such a perturbation drives the ground state away from the exactly solvable limit, but it can still be well approximated by a tensor network. Let $\ket{\phi}$ be the target MPS representation of the ground state along a ``snake'' path. 
Although our target state is possibly perturbed away from the exact stabilizer-state limit, the expectation values of the stabilizer operators should be smoothly connected to their fixed-point values ($\approx 1$) as a function of the perturbation strength. This knowledge informs our choice of the set of stabilizers as observables for regularization. We estimate these stabilizers from the finite dataset either  via direct measurement in the global-$XZ$ scheme or from classical shadows for random-$XZ$ measurements.

Our results for the surface code are presented in Fig.~\ref{fig:2_surface_code}. To begin, we consider the surface code without any external perturbation ($h_z=0$): for $n=9$ qubits arrayed  on a square geometry, our protocol outputs the state vector $\ket{\psi(\hat{A})}$, and we evaluate the error of our learning protocol using the infidelity, $1-F = 1 - \abs{\braket{\phi}{\psi(\hat{A})}}$. Figure~\ref{fig:2_surface_code}(b) shows that the infidelity improves with the number of samples $N$, for both the global-$XZ$ and random-$XZ$ measurements discussed in Sec.~\ref{sec:method}. For the latter, we find that regularization over the stabilizers, by increasing $\beta$, further reduces the infidelity in the sample-limited regime ($N < 1000$). We observe that both the global and random measurement schemes succeed in learning these states, and furthermore, the former outperforms the latter in terms of requiring fewer samples to achieve the same infidelity. However, while global-XZ measurements $\mathcal{U}_{\textnormal{Global}}$ are sufficient for certain highly structured states such as a surface code (or GHZ) state, they may not necessarily be suited for learning generic states, wherefore random-$XZ$ measurements $\mathcal{U}_{XZ}$ could be beneficial. Finally, we note that the infidelity scales polynomially with the number of samples $N$, similarly to  observations in earlier work on one-dimensional states~\cite{wangScalableQuantumTomography2020}. In particular, focusing on the exact surface-code ground state (for $h_z=0$) trained without regularization ($\beta=0$), we find an excellent fit of $1-F\propto 1/N^\alpha$, as shown in \figref{fig:2_surface_code}(b), with $\alpha = 1.16$ ($\alpha = 1.05$) for random-$XZ$ (global-$XZ$) measurements. 
To substantiate these numerical observations regarding the sample complexity needed to achieve a certain fidelity, we rigorously establish in Theorem~\ref{thm:mle-manifold-normality} that the infidelity can be probabilistically upper bounded such that
\begin{equation}\label{eq:infidelity-bound}
    \mathbb{P} \left[1 - \abs{\braket{\phi}{\psi(\est{A})}} \leq \epsilon(N) \right] \geq 1- \delta,
\end{equation}
where $\epsilon(N) = O\big( \big[\displaystyle n \chi_{\mathrm{max}}^2/ (N\delta) \big]^{1/2}\big)$ and $\chi_{\max}$ represents the maximum bond dimension.  
Interestingly, our numerical results, which exhibit $\alpha >0.5$, suggest better performance than our theoretical bound for arbitrary target states. Further details and substantial analytical and rigorous insights into the sample complexity, drawing from the statistical properties of MLE, are discussed in Appendix~\ref{app:mle_manifold}.

Next, we add a perturbation to the Hamiltonian, $h_z=0.1$, and analyze the performance of our protocol for various system sizes with up to $n=45$ qubits. Given the limitations of an MPS ansatz, which is inherently quasi-one-dimensional, here, we only increase the length of the system $L_x$ along the $x$-axis. In Fig.~\ref{fig:2_surface_code}(c), we illustrate how the infidelity scales for different system sizes (indicated by different lines) using random-$XZ$ measurements and find that the fidelity improves with $N$. Finally, we analyze how the sample complexity scales with the system size $n$. The requirement of achieving a constant fidelity is stringent as the overlap between two many-body states can vanish exponentially with $n$, a phenomenon known as the orthogonality catastrophe~\cite{anderson1967infrared}. For a finite system size, however, the overlap still serves as a useful diagnostic indicator~\cite{zanardi2006ground}. Therefore, we relax the requirement of a fixed global fidelity by considering instead the per-site local fidelity $F_{\text{local}} \equiv F^{1/n}$. Extrapolating from the data in (c), Fig.~\ref{fig:2_surface_code}(d) plots the number of samples required to achieve a local fidelity threshold, e.g., as indicated by the dashed line in (c) for $n=9$ qubits. For the largest system size, $n=45$, we find that only a few thousand measurements are necessary to attain a fidelity of $F_{\text{local}}= 0.99$.

\begin{figure}[tb]
    \centering
    \includegraphics[scale=0.33]{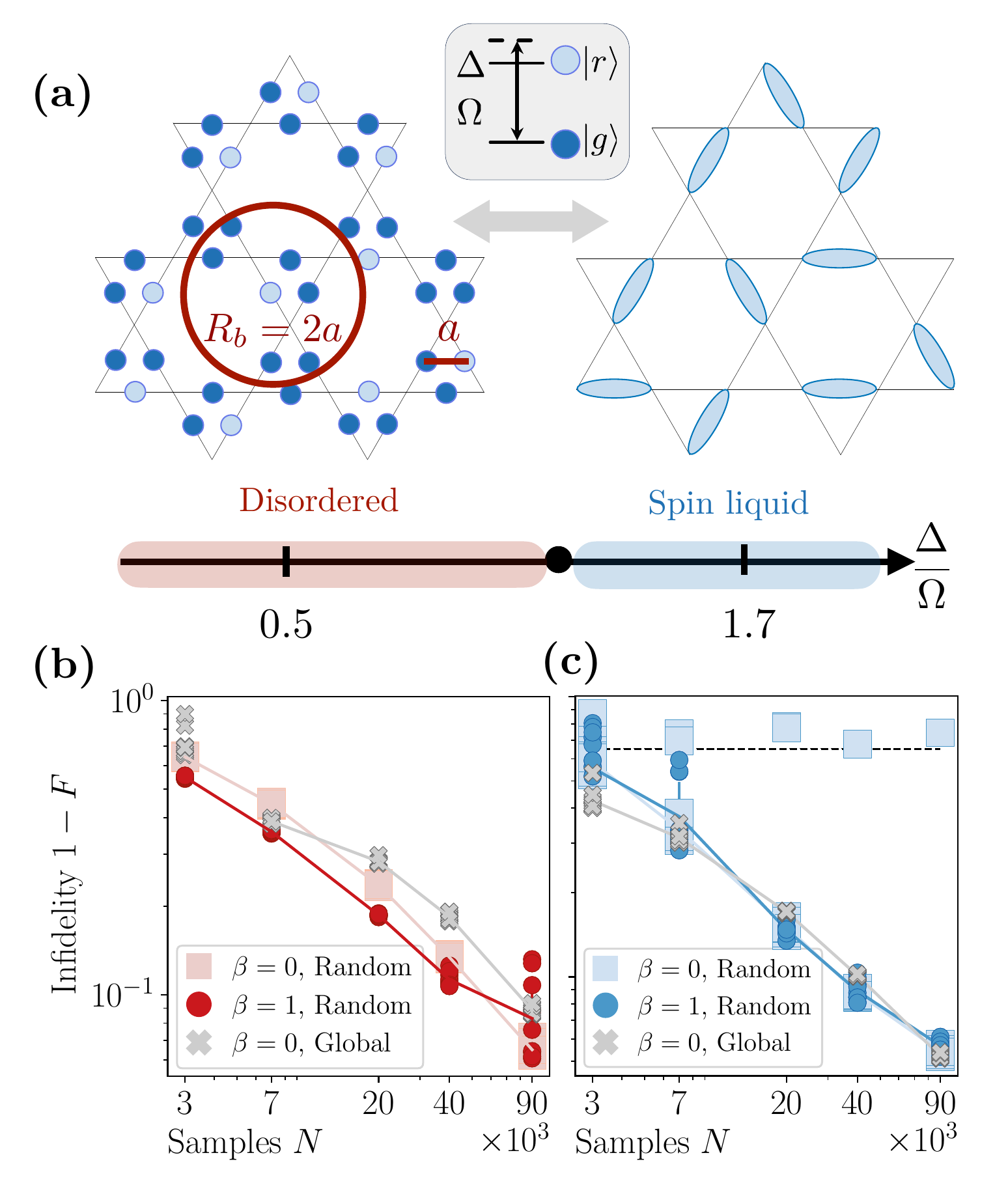}
    \caption{\textbf{Numerical results for learning quantum spin liquids.} 
    \textbf{(a)} The tuning parameters of the Rydberg Hamiltonian are the Rabi frequency $\Omega$, the detuning $\Delta$, and the blockade radius $R_b$. The gray arrow shows the mapping between a particular measurement bitstring on the ruby lattice and a dimer configuration on the kagome lattice: light (dark) blue circles correspond to an atom being in the excited $\ket{r}$ (ground $\ket{g}$) state and maps to the presence (absence) of a dimer.  A careful choice of the Rydberg blockade radius $R_b=2a$, where $a$ is the lattice spacing, ensures that the dimers do not overlap. 
    Bottom panel: The two phases as a function of the detuning ratio $\Delta/\Omega$. 
    \textbf{(b)} The ground state belongs to the trivial disordered phase at $\Delta / \Omega = 0.5$. The infidelity decreases with an increasing number of samples $N$. For random-$XZ$ measurements with $N \leq 40000$, a regularization of $\beta=1$ (red circles) further improves the fidelity compared to the case without regularization $\beta=0$ (red squares). Global-$XZ$ measurements (gray crosses) perform worse than their random-$XZ$ counterparts (red). \textbf{(c)} Same as in \textbf{(b)} but for the spin-liquid state at $\Delta / \Omega = 1.7$. Relative to the case with no regularization $\beta=0$ (light blue), regularization $\beta=1$ (dark blue) does not improve the fidelity but reduces the number of nonconvergent outliers (above the black dashed line). Global-$XZ$ (gray cross) measurements perform similarly to random-$XZ$ ones (blue). }
    \label{fig:3_rydberg_ruby}
\end{figure}

\subsection{Learning a Rydberg quantum spin liquid}
For a second application of our tomographic framework, we turn to learning ground states of strongly interacting arrays of neutral atoms. Over the last decade, these
Rydberg atom arrays have evolved into mature platforms for quantum simulation and have shown great potential for probing a variety of correlated quantum phases of matter~\cite{de2019observation,Samajdar_2020, ebadiQuantumPhasesMatter2021, scholl2021quantum,semeghiniProbingTopologicalSpin2021}. The effective Hamiltonian of this system can be written as 
\begin{alignat}{1}
\nonumber
H_{\rm Ryd} &= \frac{\Omega}{2} \sum_{\ell}X_{\ell} - \Delta \sum_{\ell} \frac{1}{2}(1 + Z_{\ell}) \\
 &+ \frac{1}{2}\sum_{\ell,\,\ell^\prime} \frac{V_{\ell,\,\ell^\prime}}{4} (1 + Z_{\ell})(1 + Z_{\ell^\prime}),
\end{alignat} 
where $\Omega$ represents the Rabi frequency, $\Delta$ is the detuning of the laser drive, and $V_{\ell,\,\ell^\prime}$ is the van der Waals interaction \wwc{potential between atoms at sites $\ell$ and $\ell'$}. The strong Rydberg-Rydberg interactions prevent neighboring atoms lying within a ``blockade radius'', $R_b$, from being simultaneously excited to the Rydberg state, thereby engendering strong quantum correlations. For an appropriate choice of the blockade radius, configurations of Rydberg atoms on the 
ruby lattice can be mapped to a set of ``dimers'' on the  kagome lattice [Fig.~\ref{fig:3_rydberg_ruby}(a)]; such quantum dimer models have long been known to host topological quantum spin liquids~\cite{moessner2008quantum,misguichQuantumDimerModel2002}. In the Rydberg system, depending on the detuning parameter $\Delta /\Omega$, the ground state (for this specifically chosen $R_b$) can be either a trivial disordered phase or a topologically ordered $\mathbb{Z}_2$ quantum spin liquid phase~\cite{samajdarQuantumPhasesRydberg2021,verresenPredictionToricCode2021,semeghiniProbingTopologicalSpin2021,PhysRevLett.130.043601}; see Fig.~\ref{fig:3_rydberg_ruby}(a). In the computational basis, these two phases are difficult to distinguish as they both lack any symmetry-breaking order. Here, we select two parameters in the phase diagram belonging to the spin-liquid ($\Delta /\Omega=1.7$) or the disordered ($\Delta /\Omega=0.5$) phase and perform DMRG simulations to find these ground states (see  Appendix~\ref{app:dmrg_state} for more details).

\begin{figure}[tb]
    \centering
    \includegraphics[width=0.48\textwidth]{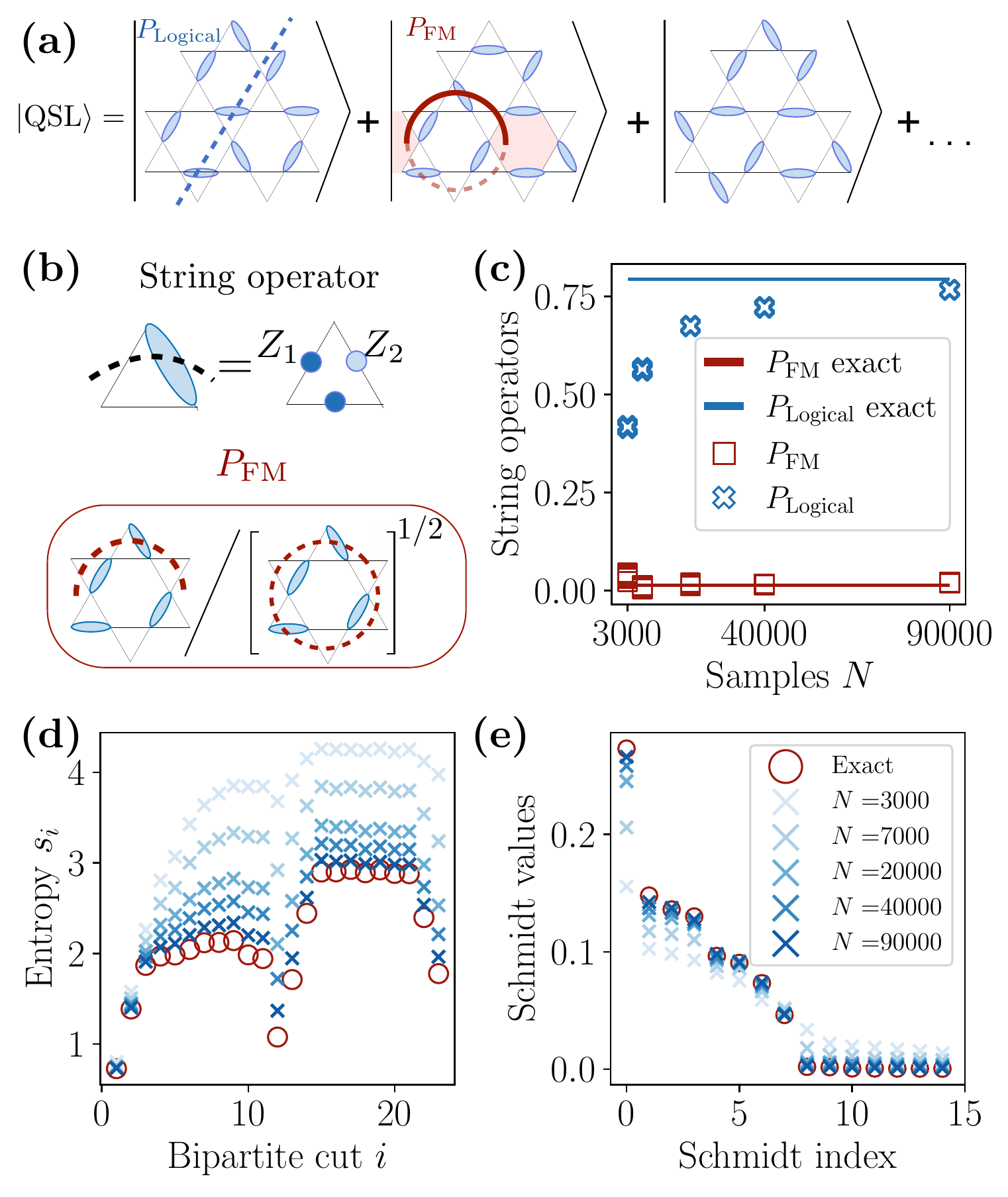}
    \caption{\textbf{Prediction of spin-liquid properties.}
    \textbf{(a)} A quantum spin liquid (QSL) state can be viewed as a  superpositions of dimers in the computational basis. \textbf{(b)} Definition of measurable diagnostic quantities such as the $Z$ string operators. Top panel: A dashed line across a bond is equivalent to applying the Pauli $Z$ operator to the qubit on the same bond. The $Z$ logical string $P_{\mathrm{Logical}}$ applies the $Z$ string operator across the system. Bottom panel: An open-string diagonal Fredenhagen-Marcu operator $P_{\mathrm{FM}}$, renormalized by the expectation value of a closed-loop operator; e.g., the semicircle corresponds to the product of $6$ Pauli $Z$ operators. \textbf{(c)} The expectation values of the $P_{\mathrm{FM}}$ (red) and $P_{\mathrm{Logical}}$ (blue) operators for different numbers of samples, $N$, converge to their exact values (lines). \textbf{(d)} Predictions (crosses) of the entanglement entropy for different sizes of the bipartition and \textbf{(e)} the Schmidt values (cut at site $i=21$) from MPS ans\"atze trained with $N$ samples converge to the exact values (open circles).}
    \label{fig:4_spin_liquid}
\end{figure}

Importantly, for this system, we do not have prior knowledge of which observables are important to regularize, unlike for the surface code studied above. Hence, analogously to MPS tomography, we consider estimations of all estimable Pauli strings that are supported on a local subsystem (taken to be 6 qubits here) and use them to regularize our MLE loss function. This subsystem should be chosen such that it captures the important correlations of the state and that the Paulis supported on this subsystem can be estimated accurately using randomized measurements~\footnote{Note that we could also consider regularization from global measurements, but then, we would only be able to estimate $X^{\otimes k}$ or $Z^{\otimes k}$.}. With regard to the task of learning the states of $n=48$ qubits, we now show that our protocol (see Appendix~\ref{app:loss_function}) is capable of learning both regimes in the phase diagram---i.e., the spin liquid as well as the disordered state---which cannot be distinguished by a local order parameter due to the topological nature of the former. 

Focusing first on the disordered state shown in Fig.~\ref{fig:3_rydberg_ruby}(b), we see that the fidelity improves with the number of measurements $N$ for all measurement schemes. Moreover, we observe a slight improvement with a regularization strength of $\beta=1$ in the sample-limited regime ($N \leq 40000$). 
Next, we turn to the spin liquid illustrated in \ref{fig:4_spin_liquid}(a). For this state as well, the infidelity generally reduces with $N$, as shown in \figref{fig:3_rydberg_ruby}(c), albeit with outliers that did not converge (indicated by squares above the dashed line). However, in contrast to the disordered phase, we observe that regularization, with strength $\beta=1$ (indicated by the dark blue circles), does not improve the fidelity beyond that obtained with no regularization $\beta=0$ (light blue squares), but the number of outliers (above the dashed line) is reduced. This lack of substantial improvement upon regularization suggests that the important operators capturing the correlations of the spin-liquid state are supported on a larger subsystem than used in our current numerical experiments.

Finally, we show that, using a successfully trained MPS, we can evaluate physical properties that may be difficult to measure directly given limited experimental controls. Certain diagnostic physical observables, such as the diagonal Fredenhagen-Marcu (FM) order parameter $P_{\mathrm{FM}}$  as well as the $Z$ logical loop operator $P_{\mathrm{logical}}$ and its $X$ counterpart, have been previously suggested for  the spin-liquid state~\cite{semeghiniProbingTopologicalSpin2021, verresenPredictionToricCode2021,cong2024enhancing}, as illustrated in Figs.~\ref{fig:4_spin_liquid}(a,b). These measurable observables can be used to verify the consistency of our learned states, and we see that the predictions from our trained MPS converge well to the target values [Fig.~\ref{fig:4_spin_liquid}(c)]. Moreover, our protocol allows for the extraction of additional quantities such as the von Neumann entanglement entropy, which is essential to determining the topological entanglement entropy---a key quantity that characterizes a gapped spin liquid~\cite{kitaevTopologicalEntanglementEntropy2006, levinDetectingTopologicalOrder2006}. Figure~\ref{fig:4_spin_liquid}(d) demonstrates that the bipartite entanglement entropy converges to the target's as the bipartition varies in subsystem size for $N\approx 90000$.  For a cut across roughly half of the system, the Schmidt values in Fig.~\ref{fig:4_spin_liquid}(e) also closely match the exact values of the target state.

\section{Discussion and outlook}\label{sec:outlook}

In this work, we combine randomized measurements with tensor networks to perform quantum state tomography via MLE, with optional regularization using classical shadows. While the randomized measurements toolbox is useful for estimating many physical \textit{properties}, tensor networks are known to efficiently represent physically relevant quantum \textit{states}. Here, we have considered a more restricted measurement setting, random-$XZ$ shadows, that is well motivated by both its tomographic completeness on real and pure states and its feasibility in near-term experiments. Using these $XZ$ measurements, we can regularize our MLE loss function with observables estimated via the classical shadow framework. We showcase the performance of our protocol by learning the topological ground states of the surface code Hamiltonian as well as the Rydberg Hamiltonian on the ruby lattice. For both of these states, our protocol with random-$XZ$ measurements can accurately reconstruct the target state up to a fidelity of over $0.95$ with under $10^5$  samples. In addition, we test our shadow regularization technique on these two applications: we find that requiring the stabilizer expectations values to agree with estimations from classical shadows further improves the fidelity for the surface code, whereas for the ruby-lattice Rydberg spin liquid, this only reduces the number of outliers observed during the optimization process. \wwc{Additionally, although our current work  focuses on perfect measurements, it will be important for future studies to consider measurement errors and carefully analyze their effects in our protocol, as the errors will---in practice---limit the best possible fidelity. While we do not systematically analyze such effects in this work, in \appref{app:noise}, we summarize in more details the common sources of  measurement errors in Rydberg atom arrays.}

Using higher-dimensional TNs combined with randomized measurements would, in principle, allow us to learn any states of interest. In the present work, we have numerically benchmarked our protocol using MPS as our model prior, given its advantages of exact contractions 
and sampling. An immediate extension would be to generalize this model's architecture to represent $2$D area-law entangled states, 
such as string bond states~\cite{glasserExpressivePowerTensornetwork2019} or \emph{projected entangled pair states} (PEPS)~\cite{verstraeteRenormalizationAlgorithmsQuantumMany2004}. \wwc{Another important generalization of our work, beyond pure states, is to adopting matrix product operator (MPO) ans\"atze, specifically, the locally purified density operator (LPDO) approach developed for learning quantum processes~\cite{torlai2023quantum}}.
% , or isometric TNs~\cite{zaletelIsometricTensorNetwork2020}.
We leave the numerical demonstrations for such tensor-network states to future work.

It constitutes an interesting perspective to extend these methods to the study of fermionic 
 quantum states, which are the focus of attention in the study of correlated electron systems. Our analysis could be extended to quantum gas microscope measurements of optical lattice Fermi-Hubbard systems~\cite{Koepsell_2021} which yield snapshots with the charge and spin on each site. 
 While in this setting, quench dynamics generated by noninteracting Hamiltonians are presumably most conceivable~\cite{Denzler},
it is possible that the spin degrees of freedom can be measured using a similar unitary ensemble such as our random-$XZ$ measurements. The variational optimization can be carried out using fermionic neural networks~\cite{Robledo_Moreno_2022} or auxiliary wavefunctions~\cite{Shackleton24,Yasir24}, replacing the MPS we have employed in the present work. To enhance the optimization process using neural-network ans{\"a}tze, it would be especially interesting to extend recent results~\cite{huang2024certifying} on locally randomized single-qubit measurements to such fermionic systems.

While our method can be generalized to quantum states realized in a wide class of quantum processors, more specifically, for Rydberg quantum simulators, it can be easily extended to learning logically encoded quantum states. In particular, our random-$XZ$ measurements are well suited for near-term encodings, such as the color codes recently implemented with neutral atom arrays~\cite{bluvstein2024logical}. In such a Rydberg array, every logical qubit, encoded, for example, in a Steane code,
% or the [[15,1,3]] Reed-Muller code~\cite{abbe2020reed}, 
can be measured in either the $X$ or the $Z$ basis, allowing for both error detection via post-selection and the ability to perform high-fidelity random-$XZ$ measurements. Even if current experiments do not aim to perform full logical quantum state tomography, our characterization and derivation of the tools required for random-$XZ$ classical shadows can be immediately applied to these logical encodings.

In settings where qubits are not logically encoded, our techniques are also experimentally amenable for physical qubits. Even as the degree of control over large quantum systems improves, single-qubit rotations will likely always remain preferable compared to many-qubit gates due to their orders-of-magnitude higher fidelity and 
short implementation times~\cite{evered2023high}. Therefore, the ability to accurately estimate a laboratory state from random measurements will remain attractive beyond the NISQ era. Our protocol can also be straightforwardly extended to random-Pauli measurements~\cite{huangPredictingManyProperties2020}. Moreover, this technique could be made even more attractive by incorporating adaptivity into our measurement protocol. Judiciously altering the measurement ensemble as more data is obtained allows for more efficient learning and thus reduces the sample complexity. To implement an adaptive learning scheme, one may incorporate ideas from Ref.~\cite{lange2022adaptive} into our learning technique in order to use fewer measurements to achieve the same infidelity.\\
\vspace{-0.001mm}
Finally, our work also lends itself to learning about experimental imperfections in state preparation. In this sense, our scheme goes beyond state verification: it learns not only the state prepared but also the unwanted deviations from the desired state, which constitute feedback to the experiment on how to better achieve a higher-fidelity preparation over many iterations. 
This variational state preparation is practical from a sample-complexity point of view when the underlying state has some efficient representation; in other words, an exponential number of degrees of freedom need not be learned. This is precisely the case in many condensed-matter contexts, where typical states exhibit entanglement structures that are well captured by tensor networks~\cite{AreaReview}.

\section{Code and data availability}
The numerical code used for producing results and generating the figures in this work is available at: \href{https://github.com/teng10/tn-shadow-qst}{https://github.com/teng10/tn-shadow-qst}. 

\begin{acknowledgments}
We thank Jacob Barandes, Christian Bertoni, Jonathan Conrad, Zo\"{e} Holmes, Robert Huang, Dmitrii Kochkov, Nathan Leitao, Alexander Nietner, Jakob Unfried, Guifre Vidal, and Manuel Rudolph for useful discussions. We thank Hong-Ye Hu and Leandro Aolita for feedback on the manuscript.
Y.T. is grateful to Dmitrii Kochkov for his inputs on the initial prototype of the numerical code base and systematic code reviews.
Y.T. also thanks the groups of Zo\"{e} Holmes and Jens Eisert for their hospitality during her visits, when part of this work was completed.
\wwc{We thank the anonymous referee for pointing out the subtlety of measurement noise for states prepared via analog simulation.}
R.S. is supported by the Princeton Quantum Initiative Fellowship.
K.V.K. acknowledges support from the Fannie and John Hertz Foundation and the National Defense Science and Engineering Graduate (NDSEG) Fellowship. The Berlin team acknowledges support from the BMBF (MuniQCAtoms), the DFG (CRC 183, FOR 2724), the Munich Quantum Valley, 
and the ERC
(DebuQC). Y.T. and S.S. have been supported by the U.S. Department of Energy under Grant DE-SC0019030. 

\end{acknowledgments}

\bibliography{Refs.bib}

\onecolumngrid  % switch to one column
\appendix

\begin{appendix}
\clearpage 

\vspace{2.0em}
\begin{center}
\textbf{\Large Appendix}
\end{center}
\tableofcontents

\renewcommand{\appendixname}{APPENDIX}
\renewcommand{\thesubsection}{\MakeUppercase{\alph{section}}.\arabic{subsection}}
\renewcommand{\thesubsubsection}{\MakeUppercase{\alph{section}}.\arabic{subsection}.\alph{subsubsection}}
\makeatletter
\renewcommand{\p@subsection}{}
\renewcommand{\p@subsubsection}{}
\makeatother

\renewcommand{\figurename}{Supplementary Figure}
\setcounter{figure}{0}
\setcounter{secnumdepth}{3}

\bigskip

\newpage

\newpage
\section{Relation to previous works}\label{app:related_works} 
This appendix aims to provide a comparison of our method to existing quantum state tomography techniques. We highlight the unique challenges and opportunities in characterizing complex quantum systems, particularly two-dimensional topological states. In addition to the references summarized below, we direct interested readers to Ref.~\cite{gebhartLearningQuantumSystems2023} for a comprehensive review of quantum state tomography. Finally, having discussed the advantages of our method for learning quasi-$2$D quantum states with topological order, we summarize the limitations faced by our method discussed in the main text.

\begin{enumerate}
    \item \textit{Matrix product state (MPS) tomography}~\cite{cramerEfficientQuantumState2010,  baumgratzScalableReconstructionDensity2013, 
    Efficient,Wick_MPS}
    is an effective method for reconstructing injective one-dimensional quantum states with short-range correlations associated with a local parent Hamiltonian, by finding the best MPS consistent with local reduced density matrices from randomized measurements. In principle, these reduced density matrices can be estimated from Pauli measurements using the techniques of classical shadows~\cite{nietner_unifying_2023}. Its application to $2$D topological states, however, is challenging: first of all, the size of the reduced density matrices required for uniquely determining a quantum state in $2$D is unclear, and secondly,  due to nontrivial ground-state degeneracies of topological states, they cannot be uniquely determined by local operators and thus, cannot be efficiently determined with this method.
    
    \item \textit{Variational MPS tomography}~\cite{wangScalableQuantumTomography2020, gomezReconstructingQuantumStates2022, kurmapuReconstructingComplexStates2022} extends MPS tomography to one-dimensional states without the injectivity constraint. It uses a variational approach to optimize the MPS representation through maximum likelihood estimation (MLE). Prior work also considered using matrix product operators (MPOs) for quantum process tomography~\cite{torlai2023quantum}.  Our work extends this protocol to $2$D states using a ``snake path'', commonly adopted in cylindrical DMRG calculations for $2$D systems. We further enhance our method by incorporating regularization techniques inspired by classical shadow tomography, combining the benefits of both approaches to better characterize two-dimensional quantum systems. We refer the reader to \secref{sec:method} for the details of our protocol.
    \item \textit{Neural network quantum state tomography} (NNQST)~\cite{torlaiNeuralnetworkQuantumState2018} 
    utilizes neural networks as variational ans\"atze within the MLE framework to model quantum states. Although particularly successful for states with fully positive amplitudes (for example, ground states of frustration-free stoquastic Hamiltonians~\cite{bravyiComplexityStoquasticFrustrationfree2008}), this method requires careful consideration when the target states have nontrivial signs and phases. Using measurement outcomes from different bases is therefore essential---a requirement directly addressed in MPS-based approaches, which can inherently relate measurements post local unitary transformations to learn relative phases. In this context, it may be useful to explore a recently proposed architecture that can accurately represent quantum spin liquids by incorporating physical symmetries~\cite{kufel2024approximately}.
    \item \textit{Generative modeling for density matrix reconstruction}~\cite{carrasquillaReconstructingQuantumStates2019}
    extends the use of neural networks to mixed-state tomography. This approach uses single-qubit positive operator-valued measure (POVM) measurements for information-complete reconstructions. While generative models can theoretically model measurement probability distributions across an exponentially large outcome space, practical challenges arise in ensuring convergence to the accurate density matrix, particularly given the sparse nature of such distributions in large systems. It is unclear how effective this approach is~\cite{huangPredictingManyProperties2020} for learning \wwc{highly} entangled quantum states such as the topological ordered states considered in this work.
    \item \textit{Classical shadow tomography}~\cite{ huangPredictingManyProperties2020, huangProvablyEfficientMachine2021} 
    efficiently estimates properties of quantum states by measuring them post application of random unitary rotations, followed by post-processing. This technique has demonstrated that a number of samples independent of the system size is sufficient for estimating the fidelity of specific projectors, such as ground states of the toric code Hamiltonian~\cite{huangPredictingManyProperties2020}, given prior knowledge of the target projector and the utilization of global Clifford measurements. A significant extension of this method~\cite{huangProvablyEfficientMachine2021, vankirkHardwareefficientLearningQuantum2022} employed a dimensional reduction technique along with a nonlinear kernel for post-processing to distinguish toric code states from trivial states using only random Pauli measurements. This unsupervised learning task, while innovative, requires the assembly of numerous distinct quantum states for shadow tomography, presenting a considerable experimental demand. Our method, utilizing random-$XZ$ measurements, sidesteps these requirements, offering a more experimentally feasible approach for a broader range of quantum systems.
    
    \item \textit{Shadow-NNQST}~\cite{weiNeuralShadowQuantumState2023} 
    combines a neural-network ansatz used in NNQST with ``regularization'' of fidelity estimations from classical shadow tomography. This regularization step necessitates global Clifford measurements for efficient fidelity estimation~\cite{huangPredictingManyProperties2020}, and an arbitrary Clifford unitary on $n$ qubits has nontrivial depth. In contrast, our approach leverages random-$XZ$ measurements, avoiding the need for such specific, experimentally costly measurement protocols. More recently, a quantity termed the ``shadow overlap'' has demonstrated faithful certification of fidelity requiring only polynomial sample complexity for a large set of quantum states~\cite{huang2024certifying}, which could significantly improve the training process of NNQST or our variational tensor network tomography. Notably, this technique also only uses single-qubit measurements that are experimentally feasible.
    \item \textit{Learning an unknown stabilizer state} can be carried out efficiently using single-copy measurements~\cite{ chialEfficientLearningDoped2023}. 
    Efficient learning of stabilizer states through single-copy measurements offers a direct method for those topological states that are also stabilizer states. Our protocol expands its applicability to nonstabilizer states amenable to representation as an MPS, which allows for broader quantum state characterization beyond the stabilizer formalism.
    A recent work has considered combining MPS with the stabilizer formalism~\cite{masot-llimaStabilizerTensorNetworks2024a}, which allows for the simulations of states beyond the scope of either method individually---this framework could be also explored in the current context of quantum state tomography.

    \item \textit{Predicting specific properties of quantum states}. If one knows the targeted properties that one wants to learn, there are a plethora of alternative methods, of which we list the most relevant approaches below. We will not try to compare our protocol to all of these directly, since such methods differ in their flexibility with regard to experimental measurements as well as the number of samples required to achieve accurate predictions. 
    \begin{enumerate}
        \item R\'enyi entropy:
        \begin{itemize}
            \item A multicopy approach~\cite{chenHierarchyReplicaQuantum2021}, based on measuring two copies of a quantum state, has been demonstrated experimentally~\cite{islamMeasuringEntanglementEntropy2015, bluvsteinQuantumProcessorBased2022}. However, this requires the preparation of multiple identical copies of a state while demanding advanced control over the quantum system to ensure the fidelity of the copies.
            \item Randomized measurements~\cite{brydgesProbingRenyiEntanglement2019, elbenEnyiEntropiesRandom2018, hu2024demonstration} leverage the statistical properties of the outcomes to estimate the R\'enyi entropy without needing multiple copies of a system. This is both similar to and encompasses classical shadow tomography.
        \end{itemize}
        \item von Neumann entropy:
        \begin{itemize}
            \item Entanglement Hamiltonian tomography~\cite{kokailEntanglementHamiltonianTomography2021} assumes that the reduced density matrices of a subsystem can be approximated by  an ansatz  involving local operators in the Hamiltonian. By learning the coefficients of this ansatz, the technique aims to reconstruct the entanglement Hamiltonian and thereby predict the von Neumann entropy.
        \end{itemize}
        \item Topological invariants:
        \begin{itemize}
            \item Direct tomography can be applied to small subsystem sizes, from which the topological entanglement entropy can be extracted~\cite{kitaevTopologicalEntanglementEntropy2006, levinDetectingTopologicalOrder2006}; this approach was demonstrated experimentally for a surface code~\cite{satzingerRealizingTopologicallyOrdered2021}. Nevertheless, the applicability of this technique for general systems remains unclear, particularly with regard to the required size of the subsystem for accurate measurement of the entanglement entropy.
            \item Randomized measurements can also be used to predict topological invariants of systems such as \emph{symmetry protected topological} (SPT) phases~\cite{elbenManybodyTopologicalInvariants2020} or fractional quantum Hall states~\cite{cianManyBodyChernNumber2021}. 
        \end{itemize}        
    \end{enumerate}
\end{enumerate}

\subsection{Potential advantages and limitations}
Our variational tensor network tomography, aimed at diagnostics and recovery of quantum states, offers potential for extracting various properties such as the von Neumann entanglement entropy, difficult to learn using other techniques. This method combines maximum likelihood estimation of the quantum state with randomized measurements and is detailed in previous works such as \cite{wangScalableQuantumTomography2020, kurmapuReconstructingComplexStates2022, torlai2023quantum}. We demonstrate this approach using simplified $XZ$ measurements, which proves sufficient for the two-dimensional topological quantum states considered, with complex long-range entanglement. Such quantum states are represented by noninjective MPS, and can not be learned via MPS tomography~\cite{cramerEfficientQuantumState2010,  baumgratzScalableReconstructionDensity2013}. 

One significant limitation is that the target state must admit an efficient tensor-network representation. In numerical experiments, we use an MPS ansatz, recognizing that while it effectively approximates one-dimensional states, it may not capture the full complexity of two-dimensional systems. This works sets the stage for future exploration of more sophisticated tensor networks. We also discuss the adaptation of statistical properties of MLE for quantum tomography, outlining both its potential and the challenges posed by the gauge freedom in MPS representations.

Another limitation is that real laboratory states can be mixed, which requires using generalizations of tensor networks such as locally purified matrix product operators~\cite{werner2016positive}. Additionally, one generally does not have access to the ground-truth fidelity with the lab state. This necessitates methods like cross-validation as discussed in \appref{app:training} or efficient fidelity certifications~\cite{huang2024certifying} to validate the trained model against empirical data.

%%%%%%%%%%%%%%%%%%%%%%%%%%%%%%%%
\newpage
\section{Asymptotic MLE error via MPS manifold}
\label{app:mle_manifold}
%%%%%%%%%%%%%%%%%%%%%%%%%%%%%%%%

Our variational algorithm in the framework of \emph{maximum likelihood estimation} (MLE) aims to find the MPS representation $\ket{\psi({A})}$ that is most consistent with measurements of the target state vector $\ket{\phi}$. Intuitively, we expect  that the accuracy of the optimal MLE solution improves with more samples. The goal of this appendix is to analyze the accuracy of such a solution as a function of the number of samples $N$.
More precisely, we provide a probabilistic upper bound of the infidelity $1-F$ by deriving a concentration inequality using the statistical properties of MLE. First, we summarize our main finding of the probabilistic bound in the theorem below:
\begin{theorem}[\textbf{Probabilistic bound}]\label{thm:mle-manifold-normality}
Let $\ket{\psi(\est{A})}$ be the MLE solution with a maximum bond dimension $\chi_{\mathrm{max}}$ for an $n$-qubit target state vector $\ket{\phi}$. With probability greater than $1-\delta$, the infidelity is upper bounded asymptotically in the limit of large number of samples $N$ as
\begin{equation}\label{eq:mle-manifold-informal}
    \mathbb{P} \left[1 - \abs{\braket{\phi}{\psi(\est{A})}} \leq \epsilon(N) \right] \geq 1- \delta,
\end{equation}
where $\epsilon(N) = O\left( \sqrt{\displaystyle\frac{n \chi_{\mathrm{max}}^2}{N\delta} }\right)$.
\end{theorem}
We provide a detailed and rigorous proof of Theorem~\ref{thm:mle-manifold-normality} concerning the \emph{optimal} MLE solution, focusing on conditions under which it represents a global minimum. Achieving such a global minimum may not always be practical, and our theorem does not address the complexity of reaching such a solution. Moreover, the theorem is specifically applicable in the \emph{asymptotic limit} as $N \rightarrow \infty$. For the validity of our theorem, we assume that the target state is within the variational space used for MLE, i.e., it is an MPS of known bond dimension.
In practice this condition might not be satisfied.

Ideally, one would like to prove a theorem that applies to our variational MPS tomography \emph{in practice}. There are a wide variety of technical obstacles to achieving this; thus, our goal in this work is to take the first steps toward this direction and provide a thorough introduction of the required mathematical tools. In this process,  we (a) present an initial theorem under a specific set of assumptions, and (b) offer pedagogical insights into the technical tools and conceptual challenges involved in proving our theorem and potentially more comprehensive statements.
The success of our numerical results, combined with these theoretical explorations, provides both motivation and foundational insights for developing more specialized analyses.

To build some intuition, we first review the maximum likelihood estimation (MLE) within traditional settings. It turns out that these results cannot be directly applied to our context due to the gauge freedom of MPS.
To generalize the MLE analysis to our MPS-estimation setting, we first elaborate on the gaps that we are filling in and how our statement differs from conventional results.

Let us begin with a standard MLE problem, where the goal is to estimate a single parameter. Considering a sequence of $N$ independent and identically distributed (i.i.d.) samples $(x_i)_{i=1}^{N}$ drawn from a target Gaussian distribution $\mathcal{P}_{\theta_0}$ with mean $\theta_0$ and a known variance $\sigma^2$,
\begin{equation}\label{eq:gaussian}
    \mathcal{P}_{\theta_0}(x)=\frac{1}{\sqrt{2\pi}\sigma} \exp\left[-\frac{(x - \theta_0)^2}{2 \sigma^2}\right].
\end{equation}
Since the variance is known, our goal is to find the estimator $\est{\theta}$ of the mean such that the likelihood of the observed samples is maximized. The \emph{negative log-likelihood} (NLL) 
of observing the $N$ samples is
\begin{equation}
    \mathcal{L}(\theta) = - \frac{1}{N}\sum_{i=1}^{N} \log\left( \exp\left[-\frac{(x_i - \theta)^2}{2 \sigma^2}\right] \right) = \frac{1}{N}\sum_{i=1}^{N} \frac{(x_i - \theta)^2}{2 \sigma^2}.
\end{equation}
For this simple case of using a Gaussian distribution as our model prior, the minimization of the NLL $\mathcal{L}$ is the same as minimizing the least-squared error. Differentiating $\mathcal{L}$ with respect to $\theta$ we readily find the extremum $\est{\theta} = \frac{1}{N}\sum_{i} x_i$, which is the well-known empirical mean.

Note that the estimator $\est{\theta}$ is a random variable, as it depends on the random samples $(x_i)_{i=1}^{N}$.
Hence, we can describe the probability distribution of $\est{\theta}$.
Under certain conditions this distribution tends to the normal distribution as the number of samples goes to infinity.
This property is known as \emph{asymptotic normality} and has been studied rigorously~\cite{neweyChapter36Large1994}.
We will state the precise technical conditions for asymptotic normality later (in the context of MLE on manifolds) in Lemma~\ref{lem:asymptotic-normality}, but emphasize here already that, in general, they entail two crucial requirements: 1) the MLE solution is consistent such that $\theta_0$ is the unique solution---the unique minimizer of $\mathcal{L}$ in the infinite-sample limit---describing $\mathcal{P}_{\theta_0}$, i.e., there exists no other $\theta$ for which $\mathcal{P}_\theta = \mathcal{P}_{\theta_0}$; 2) the landscape of the loss function $\mathcal{L}$ is well behaved for all $N$, and in particular, in the limit $N \rightarrow \infty$, the Hessian needs to be well behaved.
To be precise, asymptotic normality means that in the limit of a large number of samples $N$, the estimator $\est{\theta}$ converges to a normal distribution in the sense
\begin{equation}\label{eq:mle-normality}
    \sqrt{N}(\est{\theta} - \theta_0) \overset{\mathrm{d.}}{\longrightarrow} \mathcal{N}(0, \sigma^2).
\end{equation}
Here, $\overset{\mathrm{d.}}{\rightarrow}$ denotes \emph{convergence in distribution}, i.e., as $N$ increases, the sequence of distributions of the random variables $\sqrt{N}(\est{\theta} - \theta_0)$ for growing $N$ approximates the normal distribution increasingly well.
With the normality condition, we can write down a concentration inequality to probabilistically upper bound the error $\abs{\est{\theta} - \theta_0}$.
As $N$ increases, the error reduces with a rate of $1/\sqrt{N}$ so that in the infinite-$N$ limit, the convergence $\est{\theta} \overset{\mathrm{p.}}{\longrightarrow} \theta_0$ is guaranteed by the consistency condition.
To illustrate how this condition could fail, let us consider a modified Gaussian distribution periodic in the parameter
\begin{equation}\label{eq:gaussian}
    \mathcal{P}_{\theta_0} = \frac{1}{\sqrt{2\pi}\sigma} \exp\left[-\frac{(x - \sin(\theta_0))^2}{2 \sigma^2}\right].
\end{equation}
Now, a parameter shifted by a multiples of $2\pi$ still represents the same distribution $\mathcal{P}_{\theta_0} = \mathcal{P}_{\theta_0 + 2 \pi}$. Therefore, $\est{\theta}$ need not converge to $\theta_0$. In physics, the phenomenon where multiple parameters describe the same function is known as \emph{gauge redundancy}. In this simple scenario, we can explicitly account for the periodicity in the optimal estimator, or equivalently impose ``gauge fixing'' conditions to select a unique representative, by restricting the parameter domain to $[0, 2\pi)$. However, for distributions parametrized by MPS, it is challenging in practice to fix the gauge degrees of freedom consistently.
Specifically, transforming MPS to canonical forms using singular value decomposition only partially fixes the gauge. Additional unitary-gauge degrees of freedom remain in such a decomposition. Moreover, degeneracies in the singular values introduce redundancy within the subspace spanned by degenerate eigenvectors. 

Consequently, the MPS tensors that minimize our loss function do not satisfy the consistency condition. Nevertheless, if our measurement scheme is tomographically complete [see \appref{app:reality_random_xz_measurements}], the physical state itself meets the consistency condition. This implies that even though multiple MPS tensor configurations may solve the MLE problem, they all correspond to the same physical state. Thus, our goal is to investigate the asymptotic properties of MLE estimators across the set of physical states represented by MPS. These states, represented by a well-defined class of MPS described further in this section, form a manifold. Therefore, our analysis aims to understand the asymptotic properties of MLE estimators on this manifold.
Such a generalized asymptotic normality on a manifold has been established in previous work~\cite{hajriMaximumLikelihoodEstimators2017}. What remains is to relate the normality condition to infidelity, which we discuss by lifting the distance on a manifold to the state distance in the embedding Hilbert space via Lemma~\ref{lem:delta-bound}. 

In the rest of this self-contained appendix, we provide the necessary background on MLE and differential geometry for readers interested in technical details and a proof of the theorem, provided in \appref{sec:infidelity-bound}.
The first five sections review existing results,  which we then use in our proof.  First, we briefly review notations for MPS commonly used for the rest of the section in \appref{sec:mps}. We provide a minimal list of technical definitions in differential geometry in \appref{sec:diff-geo}, which can be skipped for readers already familiar with the topic. Equipped with this knowledge of differential geometry, we review how the projective Hilbert space can be viewed as a manifold in \appref{sec:projective-hilbert-space}. Then, given a subspace of the projective Hilbert space defined by an MPS, we review properties of the projective MPS manifold in \appref{sec:manifold}. 
Finally, our main result relies on a theorem establishing asymptotic normality of the MLE estimator on a manifold, which we restate in \appref{sec:mle-manifold}. We conclude this section with a summary of our main result as well as our interpretation of its relations to our numerical results in \appref{app:mle-discussions}.

\vspace{2mm}
\newpage
\subsection{Matrix product states}\label{sec:mps}
In this subsection, we introduce the notations and definitions used for \emph{matrix product states} (MPS) throughout the rest of the subsection, following \refcite{haegeman_geometry_2014}. 
To represent a quantum state across $n$ sites with physical dimension $d$, we define the MPS tensors over complex numbers as
\begin{equation}\label{eq:complex-tensors}
    \Abb_\chi = \bigoplus_{i=1}^n \cc^{\chi_i \times d \times \chi_{i+1}} \cong \cc^l.
\end{equation}
Here, $\chi = (\chi_1, \ldots, \chi_{n+1})$ represents the bond dimensions and sums to a total tensor dimension of $l = d \sum_{i=1}^n \chi_i\chi_{i+1}$. Specifically, the tensor at the $i$-th site is denoted by $A^{[i]}$. The tensor elements are $A_{\alpha, \beta}^{[i] \delta}$, with $1 \leq \alpha \leq \chi_i$, $1 \leq \beta \leq \chi_{i+1}$, and $1 \leq \delta \leq d$. For this work, we limit the physical dimension to $d=2$ (representing qubits) and consider open boundary conditions such that $\chi_1 = \chi_{n+1} = 1$.
We define an MPS via a map from the space of MPS tensors $\Abb_\chi$ to the Hilbert space $\mathcal{H} \cong \cc^{d^n}$ of $n$ sites.
\begin{align}\label{eq:mps_map}
    \psi_\chi: \Abb_\chi &\longrightarrow \mathcal{H}, \\
    A &\longmapsto \sum_{b \in {\lbrace 0, 1\rbrace}^n} A^{[1]b_1} \cdots A^{[n]b_n} \ket{b}.
\end{align}
The set \(\{\ket{b}\}\) represents the computational basis, where the bitstring \(b \in \{0, 1\}^n\), unless specified otherwise. Note that in the main text, we also use the notation $\ket{\psi(A)}$ to refer to the physical state, which is inferred from the context.

Physically distinct quantum states are identified up to a global phase and normalization factor, and are unique elements in the \emph{projective Hilbert space} $\projH$. The elements in this space are equivalence classes, often called ``rays'' of vectors in $\mathcal{H}$. 

\begin{definition}[\textbf{Projective Hilbert space \cite{haegeman_geometry_2014}}] 
    Associated with a complex Hilbert space $\mathcal{H}$, the projective Hilbert space is defined as
    \begin{equation}\label{eq:projective_H}
        \projH:= \mathcal{H}/\GL(1, \cc),
    \end{equation}
    which contains all rays represented by the equivalence classes $[\ket{\psi}]$.
    Two nonzero vectors $\ket{\psi}, \ket{\phi} \in \mathcal{H}$ are equivalent if and only if $\ket{\psi} = \lambda \ket{\phi}$ for some $\lambda \in \GL(1, \cc)$, the group of nonzero complex numbers.  For simplicity, we use the notation $[\psi]$ to denote $[\ket{\psi}]$.
\end{definition}

We can restrict MPS to the projective Hilbert space $\projH$ by defining a new mapping $\tilde{\psi}_{\chi}$. This projection maps an MPS tensor onto a ray, effectively associating each MPS with a unique physical state representation. While it is possible to consider the manifold of MPS without this projection~\cite{haegeman_geometry_2014}, using the projective MPS simplifies our MLE analysis by removing ambiguious normalization or phase factors that are not physically observable.

\begin{definition}[\textbf{Projective MPS \cite{haegeman_geometry_2014} (Definition 13)}] 
The map $\tilde{\psi}_{\chi}$ from the space of MPS tensors $\Abb_\chi$
  to the projective Hilbert space $\projH$ is
\begin{align}
        \tilde{\psi}_\chi: \Abb_\chi &\longrightarrow \projH, \\
        A &\longmapsto [\psi_\chi(A)], 
\end{align}
where $[\psi_\chi(A)]$ denotes the ray representing the equivalence class of the state vector $\psi_\chi(A)$ under rescaling by nonzero complex numbers. 
\end{definition}    
Throughout the rest of this section, what we refer to by $\ket{\psi}$ or $[\psi]$ should be understood from the context as denoting either a normalized state vector or a ray in the projective Hilbert space, in distinction to the projective mapping $\tilde{\psi}_{\chi}$. Note that the image $\tilde{\psi}_\chi(\Abb_\chi)$ represents only a subspace of $\projH$ and does not possess the structure of a vector space~\footnote{For instance, consider the sum of two MPS of bond dimension $\chi$ as $({\ket{\psi_{\chi}} + \ket{\phi_{\chi}}})/{\sqrt{2}}$. This new state has bond dimension up to $2\chi$, which is beyond the original space defined for bond dimension $\chi$.}. In the following sections, we will provide constraints on the domain of MPS tensors such that their image under $\tilde{\psi}_\chi$ (and $\psi_\chi$) forms a differentiable manifold within $\projH$. 
Equipping the set of MPS with such an additional structure, we will be able to rigorously analyze the accuracy of the MLE solution on the space of a manifold, as detailed in \appref{sec:mle-manifold}. 

\vspace{2mm}
\subsection{Differential geometry}\label{sec:diff-geo}

This subsection introduces the essential concepts of differential geometry required to navigate through the subsequent discussions in Appendix~\ref{sec:projective-hilbert-space}, \ref{sec:manifold}, and \ref{sec:mle-manifold}. Assuming no prior expertise in differential geometry, we aim to outline here the key definitions and provide intuitive explanations for these concepts that can be found in standard textbooks~\cite{kuehnel_diffgeo_2008}. While this section is in principle self-consistent, we highly encourage readers who would like to fully understand this material to consult Refs.~\cite{kuehnel_diffgeo_2008, nakahara_geometry_2003} for more context.
Readers who are well familiar with differential geometry of Riemannian manifolds may choose to skip most of this section and only briefly review Def.~\ref{def:geodesic} and Fact~\ref{fact:constant-velocity}, which introduces terminologies used specifically in the context of our theorem.

Let us begin by briefly recalling the concept of an  $r$-dimensional differentiable manifold $M$. We use $\theta \in M$ to denote elements on the manifold, where, in the subsequent sections, $\theta$ parametrizes some underlying probability distribution. Formally, any local neighborhood $M_i\subset M$ can be related to Euclidean space $\rr^r$. This relationship is established through an injective map $\varphi_i:M_i \rightarrow \rr^r$, known as a \emph{chart}. A chart provides a local coordinate system near each point in its neighborhood. Intuitively, the manifold is covered by a collection of neighborhoods 
$M_i$ such that the union $\bigcup_{i\in I} M_i = M$ represents the entire manifold.

The collection of charts, known as an \emph{atlas}, allows us to study functions defined on $M$ using the usual machinery from calculus defined on $\rr^r$.
Moreover, for $M$ to be called differentiable, any composition $\varphi_i^{-1} \circ \varphi_j$ defined on any two overlapping patches $M_i \cap M_j \neq \emptyset$ must be differentiable.
This allows us to define the tangent vectors $X$ at any point $\theta \in M$.
In alignment with Ref.~\cite{haegeman_geometry_2014}, we adopt a geometric perspective of tangent vectors by representing them as column vectors $(X_1, \cdots, X_r)^\top$. 
The components of 
the vector $X$ depend on the local coordinates in $\rr^r$, defined by the chosen 
chart $\varphi$.\footnote{Tangent vectors to $M$ at points $\theta$ can be naturally seen as directional derivatives. For a smooth function
$f:M \rightarrow \rr$ defined on the manifold, if $X$ is a tangent vector to the manifold $M$ at $\theta$,
then the directional derivative (defined in Def.~\ref{def:directional-derivative}) is given by $D_X(f) = X(f)$. For what follows, it is convenient to directly refer to tangent vectors.
}
The space of all tangent vectors at $\theta$ is called the tangent space $\rmT_{\theta}M$ and it is isomorphic to $\rr^r$. Consequently, each tangent vector can be expressed as $X = \sum_{a=1}^{r} X_a e_a$ where $\lbrace e_a\rbrace$ is a basis of $\rmT_\theta M$.

To compute the lengths and angles of the tangent vectors, it is essential to define an inner product for each tangent space, called \emph{metric}. 
In general, a metric provides an inner product for each tangent space $\rmT_\theta M$ which explicitly depends on the point $\theta$, i.e., as one moves across the manifold, the inner product \emph{changes}.
When the inner product is real-valued and its change is gradual, such that it is differentiable (as precisely defined below), the metric is called a \emph{Riemannian} metric.
This allows us to track the change of lengths and angles between tangent vectors as one moves across the manifold so that we can define distances between points on the manifold.

\begin{definition}[\textbf{Riemannian metric~\cite{kuehnel_diffgeo_2008}}]\label{def:rie-metric}
    A Riemannian metric $g$ is a smooth family of inner products $\lbrace g_\theta \vert \theta \in M\rbrace$ on the tangent spaces of $M$. The metric $g_\theta$ at each point $\theta$, defined by
    \begin{equation}
        g_\theta: \rmT_\theta M \times \rmT_\theta M \longrightarrow \rr,
    \end{equation}
    is bilinear and positive definite. Represented in the basis $\lbrace e_a\rbrace$, its matrix elements $(g_\theta)_{a, b} = g_\theta(e_a, e_b)$ vary smoothly with respect to $\theta$.
\end{definition}

\vspace{2mm}
In the discussion that follows, we will use $g$ to denote the metric $g_{\theta}$, as it is understood that the metric is defined for the tangent space at each point $\theta \in M$. Additionally, we adopt the notation $\braket{\bullet}{\bullet}_g$ to represent the inner product of any two tangent vectors, which is computed according to the metric $g$. 

\vspace{1mm}
\begin{definition}[\textbf{Riemannian manifold}]\label{def:rie-manifold}
    A differentiable manifold $M$ together with a Riemannian metric $g$ forms a Riemannian manifold $(M, g)$.
\end{definition}

\vspace{2mm}
To be able to discuss curves connecting two points on a manifold (Def.~\ref{def:geodesic}), we introduce the concept of a Riemannian connection, which generalizes the notion of the covariant derivative.

\vspace{1mm}
\begin{definition}[\textbf{Riemannian connection~\cite{kuehnel_diffgeo_2008}}]\label{def:rie-connection}
    Let $\Gamma(\rmT M)$ be the set of differentiable vector fields $\theta \mapsto X \in \rmT_\theta M$ on a Riemannian manifold $(M, g)$.
    A map $\nabla: \Gamma(\rmT M) \times \Gamma(\rmT M) \rightarrow \Gamma(\rmT M)$ is called a Riemannian connection, typically expressed as $\nabla_X Y$ representing $\nabla(X, Y)$. It holds that
    \begin{enumerate}
        \item $\nabla_{fX + hY}Z = f\nabla_X Z + h \nabla_Y Z$,
        \item $\nabla_X(Y + Z) = \nabla_X Y + \nabla_X Z$,
        \item $\nabla_X fZ = f\nabla_X Z + (\rmD_X f) Z$,
        \item $\rmD_X g(Y, Z) = g(\nabla_XY, Z) + g(Y, \nabla_XZ)$,
        \item $\nabla_XY - \nabla_YX - [X, Y] = 0$,
    \end{enumerate}
    for all $X, Y, Z \in \Gamma(\rmT M)$ and all smooth functions $f, h: M \rightarrow \rr$ defined on the manifold.
    $[\bullet, \bullet]$ is the Lie bracket defined by $[X, Y](f) := \rmD_X\rmD_Y f - \rmD_Y\rmD_X f$.
\end{definition}
\vspace{2mm}
Here we use $\rmD_X f$ to denote the directional derivative of a scalar-valued function $f$, also commonly denoted by $X(f)$ in the literature.
The directional derivative has a direct correspondence with the tangent vectors, and can be formally defined as follows:
\begin{definition}[\textbf{Directional derivative}]
    \label{def:directional-derivative}
    Let $M$ be a differentiable manifold, $\theta \in M$ and $X \in \rmT_\theta M$.
    Given a differentiable function $f:M \rightarrow \rr$, the directional derivative $\rmD_X f$ (also denoted $X(f)$) at $\theta$ is defined as
    \begin{equation}
        \rmD_X f \big\vert_\theta = \frac{\dd}{\dd t} f(\gamma(t))\Big\vert_{t=0},
    \end{equation}
    where $\gamma: \rr \rightarrow M$ is a differentiable curve such that $\gamma(0) = \theta$ and the velocity $\dot{\gamma}(0) = X$.
\end{definition}
\vspace{2mm}
The Riemannian connection is sometimes also referred to as the \emph{Levi-Civita connection} and for each Riemannian manifold, there exists exactly one Riemannian connection~\cite{kuehnel_diffgeo_2008}. As mentioned above and as the notation suggests, the Riemannian connection provides us with a notion of directional (or covariant, to be more precise) derivatives on a manifold.
This allows us to describe how vectors change as they are moved along a curve on the manifold.
Specifically, we can now define a notion of \emph{parallelism} along curves, a concept that is very intuitive in Euclidean space but much less so in curved spaces.
\vspace{2mm}
\begin{definition}[\textbf{Parallel transport}]\label{def:transport}
    Let $\alpha: [0, 1] \rightarrow M$ be a smooth curve on the Riemannian manifold $(M, g)$ and $\dot{\alpha}(t)$ be its velocity vector in the tangent space $\rmT_{\alpha(t)}M$. The parallel transport of a tangent vector $X_0 \in \rmT_{\alpha(0)}M$ along the curve $\alpha$ under a Riemannian connection $\nabla$ is defined by the equation
    \begin{equation}
        \nabla_{\dot{\alpha}(t)} X(t) = 0,
    \end{equation}
    where $X(t) \in \rmT_{\alpha(t)} M$ for all $t \in [0, 1]$ and $X(0) = X_0$.
    The vector field $X$ is said to be parallel along $\alpha$ if it satisfies the above condition.
\end{definition}
\vspace{2mm}
A special case of parallel vector fields is given by \emph{geodesic curves}.
The velocity vector $\dot{\gamma}$ of a geodesic curve $\gamma: I \rightarrow M$ (with $I \subseteq \rr$ some parameter interval) is parallel along $\gamma$ itself. For simplicity, we will take the unit interval $I = [0, 1]$ because a curve with any other interval can always be reparametrized. 
The notion of a geodesic extends the concept of a straight line from Euclidean space to curved spaces like manifolds.
In general, there can be multiple geodesic curves connecting two points on a manifold\footnote{For example, the great circle on the sphere $S^2$ is a geodesic curve. Two points on the great circle separate the circle into two arcs. The shorter arc is the shortest connecting geodesic, and the longer arc is the longer connecting geodesic.}. For our purpose, we only consider the curves that represent the \emph{shortest} path 
between two points on the manifold, parametrized by arc length on the unit interval.
\vspace{2mm}
\begin{definition}[\textbf{Shortest connecting geodesic}]\label{def:geodesic}
    Let $(M, g)$ be a Riemannian manifold and  $\nabla: \Gamma(\rmT M) \times \Gamma(\rmT M) \rightarrow \Gamma(\rmT M)$ the associated Riemannian connection. Given two points $\theta_1, \theta_2 \in W$ where $W \subset M$ is a sufficiently small neighborhood, the curve
    \begin{align}
        \gamma: [0, 1] &\longrightarrow W ,\\
        t &\longmapsto \gamma(t), 
    \end{align}
    is the shortest connecting geodesic of $\theta_1$ and $\theta_2$ if $\gamma(0) = \theta_1$ and $\gamma(1) = \theta_2$, and for all $t \in [0, 1]$ the tangent vectors $\dot{\gamma}$ satisfy
    \begin{equation}
        \nabla_{\dot{\gamma}}\dot{\gamma} = 0.
    \end{equation}
    This implies that $\dot{\gamma}$ is parallel transported along the geodesic curve $\gamma$.
\end{definition}

Note, that in general, one can loosen the condition for geodesics to 
$\nabla_{\dot{\gamma}}\dot{\gamma} = \lambda(t) \dot{\gamma}$ where $\lambda$ is some scalar function and $\gamma$ is required to be regular ($\dot{\gamma}(t) \neq 0$ for all $t$).
However, by setting $\nabla_{\dot{\gamma}}\dot{\gamma} = 0$, we enforce that the velocity vector $\dot{\gamma}$ of the geodesic cannot change in the direction of the curve, so that the velocity vector has a constant length throughout the curve. This is formally stated in Fact~\ref{fact:constant-velocity} below.
\vspace{2mm}

\begin{fact}[\textbf{Velocity vector}]\label{fact:constant-velocity}
    Consider a geodesic $\gamma$ (defined in Def.~\ref{def:geodesic}) on a Riemannian manifold $(M, g)$ of dimension $r$. If $\{e_a(t)\}_{a=1}^{r}$ represents the set of basis vectors parallel transported along $\gamma$, the velocity vector $\dot{\gamma}$ can be expressed as
    \begin{equation}
        \dot{\gamma} = \sum_{a=1}^{r} \Delta_a e_a(t),
    \end{equation}
    where the decomposition coefficients $\{\Delta_a\}_{a=1}^{r}$ are constant due to the constraint $\nabla_{\dot{\gamma}}\dot{\gamma} = 0$.
    Moreover, the length of the velocity vector is given by
    \begin{align}
        \Vert \dot{\gamma}(t)\Vert &:=  \sqrt{ \braket{\dot{\gamma}}{\dot{\gamma}}_g} \\
        & = \sqrt{\sum_{a,b}\Delta_a\, g_{a,b}\, \Delta_b},
        \nonumber
    \end{align}
    which is a constant for all $t \in [0, 1]$, due to the aforementioned constraint $\nabla_{\dot{\gamma}}\dot{\gamma} = 0$ representing the constant speed at which $\gamma$ traverses the manifold.
    The metric tensor components $g_{a,b}(t) = \braket{e_a(t)}{e_b(t)}_g$ in the basis $\{e_a(t)\}_{a=1}^{r}$ along the curve are constant as well.
\end{fact}

\begin{proof}
We begin by expanding the condition for the geodesics for a given coordinate system
\begin{align}
    0 &= \nabla_{\dot{\gamma}}\dot{\gamma} \\
    \nonumber
    &= \sum_a \nabla_{\dot{\gamma}} \Delta_a e_a \\
      \nonumber
    &= \sum_a \Delta_a \nabla_{\dot{\gamma}} e_a
    + \sum_a (\rmD_{\dot{\gamma}}\Delta_a) e_a.
      \nonumber
\end{align}
Using the definition of the parallel transport of the basis vectors, $\nabla_{\dot{\gamma}} e_a(t) = 0$, we conclude that $\sum_a (\rmD_{\dot{\gamma}}\Delta_i) e_a = 0$.
Since the basis vectors $\lbrace e_a \rbrace$ are linearly independent, the coefficients $\rmD_{\dot{\gamma}}\Delta_a$ must all be zero individually, hence $\Delta_a$ is a constant in $t$ for all $a$.
Moreover, any Riemannian connection preserves the metric (Property 4 in Def.~\ref{def:rie-connection}) and thus, we have
\begin{equation}
    \rmD_{\dot{\gamma}} g_{a,b}(t) = \braket{\nabla_{\dot{\gamma}} e_a(t)}{e_b(t)} + \braket{e_a(t)}{\nabla_{\dot{\gamma}} e_b(t)},
\end{equation}
which is zero since $\nabla_{\dot{\gamma}} e_a(t) = 0$ for all $a = 1, \ldots, r$, by definition.
\end{proof}

Given a connecting geodesic curve on a Riemannian manifold, the length of the curve can serve as a distance measure between the two end points $\theta_1$ and $\theta_2$.
The length of the curve depends on the metric along the path.
Let $\gamma$ be a connecting geodesic as in Def.~\ref{def:geodesic}.
The arc length $L_\gamma(0, 1)$ of any segment of of the curve between $\gamma(0)$ and $\gamma(1)$ is given by
\begin{align}
    \label{eq:arc-length}
    L_\gamma(0, 1) &= \int_0^1 \dd t \sqrt{ \braket{\dot{\gamma}}{\dot{\gamma}}_g}\\
    \nonumber
    &= \sqrt{ \braket{\dot{\gamma}}{\dot{\gamma}}_g}, \label{eq:velocity-length}
\end{align}
which is the constant length of the velocity vector along $\gamma$.

\vspace{2cm}
\subsection{Projective Hilbert space as a manifold}
\label{sec:projective-hilbert-space}
In this subsection, we apply the previously discussed concepts of differential geometry to describe quantum systems. 
In particular, we review how the projective Hilbert space $\projH$ can be understood as a manifold~\cite{haegeman_geometry_2014}.
These descriptions allow us to relate the length of geodesics within this manifold to the quantum state overlap---the fidelity---that we aim to analyze.
\vspace{2mm}
\begin{definition}[\textbf{Projective Hilbert space as a manifold \cite{haegeman_geometry_2014}}]\label{def:hilbert-space-manifold}
    The projective Hilbert space $\projH \simeq \cc \pp^{2^n-1}$ is a differentiable manifold, whose elements correspond to vectors in the Hilbert space modulo a global phase and
        normalization.
\end{definition}
\vspace{2mm}

Note that $\projH$ can be naturally seen as a complex manifold of complex dimension $2^n-1$, although it could, at the same time, also be defined as a real manifold of dimension $2(2^n - 1)$.
Following Ref.~\cite{nakahara_geometry_2003} (Sec.~8.2) and Ref.~\cite{haegeman_geometry_2014} (Sec.~IIA), we can take the tangent space of the \emph{complex} manifold $\projH$ at some point $[\psi] \in \projH$ as the \emph{real} vector space $\rmT_{[\psi]}\projH$ of dimension $2(2^n - 1)$.
To avoid confusion, we remark that it can be useful to consider the \emph{complexification} $\rmT^\cc_{[\psi]}\projH$ (i.e., 
the \emph{complex} span of any basis of $\rmT_{[\psi]}\projH$) of this vector space to define a set of so-called holomorphic tangent vectors.
Note, however, that $\rmT^\cc_{[\psi]}\projH$ has complex dimension $2(2^n - 1)$ and is therefore over-parametrized.

For our result, we only rely on the manifold $\projH$ having the structure of a Riemannian manifold and therefore, we will only work with $\rmT_{[\psi]}\projH$.
However, as we will see, $\rmT_{[\psi]}\projH$ is closely related to the Hilbert space $\mathcal{H}$ itself.
Hence, it is more natural to view the $2(2^n - 1)$-dimensional real tangent vectors as $2^n - 1$-dimensional complex vectors.
Those correspond to vectors in $\mathcal{H}$ that are orthogonal to $[\psi]$.
We stress that these vectors are \emph{not} elements of the complexified tangent space (which would be complex vectors of dimension $2(2^n - 1)$).
\vspace{2mm}
\begin{definition}[\textbf{Tangent space of projective Hilbert space~\cite{haegeman_geometry_2014}}]\label{def:tangent-hilbert-manifold}
    The tangent space $\rmT_{[\psi]}\projH = \mathcal{H}/\!\sim$ at a base point denoted by $[\psi]$ is the space of all tangent vectors that satisfy the following equivalence relation. Let $\ket{\psi_1}, \ket{\psi_2} \in \mathcal{H}$ be two state vectors in the Hilbert space, 
    \begin{equation}
        \ket{\psi_1} \sim \ket{\psi_2} :\!\iff \ket{\psi_1} - \ket{\psi_2} =\alpha \ket{\psi}, 
    \end{equation}
    for some constant $\alpha \in \cc$.
\end{definition}
The equivalence relation removes one dimension from the Hilbert space so that the tangent space is isomorphic to the Euclidean complex space and the Eucliean real space of twice the dimension, i.e., $\rmT_{[\psi]}\projH \simeq \cc^{2^n-1} \simeq \rr^{2(2^n-1)}$.
For the purpose of our result, we only rely on the minimal structure of a Riemannian manifold.
$\projH$ can be 
% \yan{We can either provide intuitions for why this is natural or remove this word?} \frederik{Just because it originates from the standard inner product in Hilbert space. We can remove it or keep it. I am fine either way.} 
equipped with the Fubini-Study metric in the context of quantum geometry. Such a choice of the Fubini-Study metric is natural because it represents the inner product between two state vectors in the Hilbert space. 

\begin{definition}[\textbf{Fubini-Study metric~\cite{cheng2010quantum}}]\label{def:fs-metric}
    Let $\ket{\psi_1}$ and $\ket{\psi_2}$ be two tangent vectors in $\rmT_{[\psi]}\projH$ and $\ket{\psi}$ be the normalized representation of $[\psi]$. The Fubini-Study metric is defined as
    \begin{align}\label{eq:fs-definition}
        g :\rmT_{[\psi]}\projH \times \rmT_{[\psi]}\projH &\longrightarrow \rr, \\
        (\ket{{\psi_1}}, \ket{{\psi_2}}) & \longmapsto \re \big[\braket{\psi_1}{\psi_2} - \matrixel{\psi_1}{\psi \rangle \langle \psi}{\psi_2} \big ]. 
    \end{align}
\end{definition}
\vspace{2mm}
Note that here we have defined the Fubini-Study metric by only taking the real part of what is defined in Ref.~\cite{haegeman_geometry_2014}. This is 
a more common definition used in other literature~\cite{cheng2010quantum}, and suffices for our purpose of defining the length of a geodesic.  

In the following, we denote the Fubini-Study metric by $g$, since the choice of metric for $\projH$ is unambiguous.
Using the Fubini-Study metric, we can calculate the distance between two rays in projective Hilbert space.
\vspace{2mm}
\begin{definition}[\textbf{State distance defined 
by Fubini-Study metric in projective Hilbert space}]\label{def:fs-distance-geodesic-length}
     Let $\gamma: [0, 1] \rightarrow \projH$ be the shortest geodesic (see Def.~\ref{def:geodesic}) connecting two rays $\gamma(0)=[\phi]$ and $\gamma(1)=[\psi]$ on the Riemannian manifold $(\projH, g)$.
     The Fubini-Study distance between these two quantum states is defined as the geodesic curve length
     \begin{equation}\label{eq:fs-distance-metric}
         d_\mathrm{FS}([\phi], [\psi]) := \int_0^1 \dd t \sqrt{\braket{\dot{\gamma}}{\dot{\gamma}}_g}.
     \end{equation}
% \yan{There are two versions of the definitions that are related. THe first one is a formal definition of Fubini-Study distance. The second one is its infinitesimal defintion, which is consistent with the first.}     
    %  Version 1:
    % \begin{equation}
    %     \vert\braket{\psi}{\phi}\vert = \cos \int_0^1 \dd t \sqrt{\braket{\dot{\gamma}}{\dot{\gamma}}_g}
    % \end{equation}
    % Version 2 (up to some factors in \href{https://arxiv.org/pdf/1012.1337}{paper EQ 9} and \href{https://arxiv.org/pdf/2106.13660}{paper EQ 26})
    % \begin{equation}
    %     \vert\braket{\psi}{\phi}\vert^2 = 1 - \frac{1}{2} \int_0^1 \dd t {\braket{\dot{\gamma}}{\dot{\gamma}}_g}
    % \end{equation}
\end{definition}
\vspace{2mm}
Note that this Fubini-Study distance is known to be related to the angle between the two states $[\phi]$ and $[\psi]$.
This is often stated without any proof, but rather the intuitive motivation is that the shortest path in projective space corresponds to the arc length of the enclosed segment of the unit circle containing $[\phi]$ and $[\psi]$.
Since we define the state distance via the curve length of the \emph{shortest} connecting geodesic, we state an explicit expression of the Fubini-Study distance below. % along with a proof.
\vspace{2mm}
\begin{lemma}[\textbf{Fubini-Study distance}]
    \label{lem:fubini-study-distance-arc-length}
    The Fubini-Study distance can be calculated as\begin{equation}\label{eq:fs-distance-arc}
        d_\mathrm{FS}([\phi], [\psi]) = \arccos \left\vert\braket{\phi}{\psi}\right\vert.
    \end{equation}
\end{lemma}

Ref.~\cite{Geometry}~ (Sec.~4.5) contains a proof by defining the Fubini-Study distance as given by \equref{eq:fs-distance-arc} and showing that the Fubini-Study metric in in \equref{eq:fs-definition} is consistent.

\vspace{2mm}
\subsection{Manifold of matrix product states}\label{sec:manifold}
Building on the notations established for MPS in \appref{sec:mps} and the concepts from differential geometry outlined in Appendix~\ref{sec:diff-geo} and \ref{sec:projective-hilbert-space}, we now proceed to introduce the variational set of projective MPS, as described in~\refcite{haegeman_geometry_2014}.
% \ryan{We should maybe foreshadow how this is related to our variational set}.\yan{Sounds good, we can do this in the last discussion section}
It turns out that the variational set $\psi_\chi(\Abb_\chi)$ (with $\Abb_\chi$ defined in \equref{eq:complex-tensors}) does not have the structure of a manifold.
In essence, points in $\Abb_\chi$ where one of the Schmidt ranks is lower than specified by the bond dimensions $\chi_1, \ldots, \chi_{n+1}$ constitute self intersections\footnote{
Technically, a self intersection at a point $\theta \in M$ is a point which has two neighborhoods $M_1$ and $M_2$ such that the images under the charts $\varphi_1(M_1)$ and $\varphi_2(M_2)$ are open (required by the definition of a differentiable manifold), but the images $\varphi_1(M_1 \cap M_2)$ and $\varphi_2(M_1 \cap M_2)$ are not open.
A simple example is the figure eight, which is the set of points in $\rr^2$ satisfying $y^2 = x^2 - x^4$. This is a 1-dimensional manifold if one excludes the self-intersection point $(x, y) = (0, 0)$.

To see how this phenomenon manifests itself in the variational set of MPS $\psi_\chi(\Abb_\chi)$, we can consider the example of $n=2$ qudits of local dimension $d=3$ and bond dimension $\chi=2$.
The (complex) dimension of the corresponding manifold $M = \psi_{\chi=2}(\mathcal{A}_{\chi=2})$ is
\begin{equation}
    \dim(M) = 8,
\end{equation}
since the dimension of $\Abb_{\chi=2}$ (and therefore also $\mathcal{A}_{\chi=2}$) is 12 and the dimension of the structure group is 4 (see Theorem~14 in Ref.~\cite{haegeman_geometry_2014}).
Therefore, the tangent space of $M$ at any point has complex dimension 8.

However, if the state vector $\ket{0,0}$---which has Schmidt rank 1---was an element of $M$, there would be 9 linearly independent ``tangent vectors'' at $\ket{0,0}$.
This can be easily seen by considering the curves
\begin{equation}
    \gamma_{i,j}(t) = \ket{0,0} + t \ket{i,j},
\end{equation}
for any $i=0,1,2$ and any $j=0,1,2$.
For any $t \in \rr$, any $i=0,1,2$, and any $j=0,1,2$, the Schmidt rank of $\gamma_{i,j}(t)$ is at most 2 and indeed, all velocity vectors $\dot{\gamma}_{i,j}$ are linearly independent.
This would require the tangent space of $M$ to be 9-dimensional at $\ket{0,0}$, which contradicts the fact that the manifold is 8-dimensional, as discussed above.
}
in $\psi_\chi(\Abb_\chi)$.
We refer to Ref.~\cite{haegeman_geometry_2014} for more details on this issue and only mention here that one needs to introduce additional constraints in order to define a variational set which forms a valid Riemannian manifold.

Let us begin with the MPS tensors $A^{[1]}, \ldots, A^{[n]}$ as defined in \equref{eq:mps_map} and  denote the transfer matrices by $T^{[i]}_{(\alpha, \alpha'), (\beta, \beta')} = \sum_\delta A_{\alpha, \beta}^{[i] \delta} \bar{A}_{\alpha', \beta'}^{[i] \delta}$, where $\bar{A}$ represents the complex conjugation of $A$.
Furthermore, we denote the left and right virtual density matrices by
\begin{align}
    l^{[i]}(A) &= T^{[1]} \cdots T^{[i-1]}, \\
    r^{[i]}(A) &= T^{[i-1]} \cdots T^{[n]},
\end{align}
while we define $l^{[0]}(A) = r^{[n]}(A) = 1$.
As mentioned above, we need to exclude MPS tensors for which the Schmidt rank between sites $i$ and $i+1$ is not equal to $\chi_{i+1}$, i.e., where any of the left or right density matrices are not positive definite.
\begin{definition}[\textbf{Full-rank MPS tensors (\cite{haegeman_geometry_2014} Sec III.B Definition 10)}]
    % \yan{Do positive-definite virtual left density matrices imply the same for the right density matrices? Only one is required to have a ``trivial stabilizer group'' [Corrolary 9 in \cite{haegeman_geometry_2014}]. Both need to be positive definite to define a proper distance in EQ(68) in \cite{haegeman_geometry_2014}}
    The subset $\mathcal{A}_\chi \subset \mathbb{A}_\chi$ of full-rank MPS tensors contains all tensors characterized by having positive-definite virtual left and right density matrices $l^{[i]}(A)$ and $r^{[i]}(A)$, at every site $i$:
    \begin{equation}
    \mathcal{A}_\chi = \{A \in \mathbb{A}_\chi \,|\, l^{[i]}(A) > 0, \,\, r^{[i]}(A) > 0, \forall i\}.
    \end{equation}
\end{definition} 

\vspace{2mm}
We can now define the set of projective 
full-rank MPS
\begin{equation}
    \label{eq:projectiv-manifold}
    \tilde{M} := \tilde{\psi}_{\chi}(\mathcal{A}_\chi)
\end{equation}
to represent physical states within the projective Hilbert space that are described by full-rank MPS tensors.
Having removed the self intersections, this variational set has the structure of a Riemannian manifold.
\vspace{2mm}
\begin{lemma}[\textbf{Projective MPS manifold (\cite{haegeman_geometry_2014} Theorem 15)}]\label{lem:projective_manifold}
    The set $\tilde{M}$, defined as the image of the projective MPS mapping $\tilde{\psi}_\chi$ over the full-rank MPS tensors $\mathcal{A}_\chi$, forms a differentiable manifold within the projective Hilbert space $\projH$, known as the projective MPS manifold.
\end{lemma}
Note that in Ref.~\cite{haegeman_geometry_2014}, it is demonstrated that $\tilde{M}$ is actually a \emph{K\"ahler manifold} when one endows it with the \emph{complex-valued} Fubini-Study metric, which is shown to be a Hermitian metric.
This introduces additional structure---studied in complex geometry---to the manifold.
Since we only rely on the variational set being a well-defined Riemannian manifold, we ignore the additional structure here and only define the \emph{real-valued} Fubini-Study metric in alignment with Def.~\ref{def:fs-metric} by inducing the Fubini-Study metric from the embedding space $\projH$.
The projective MPS manifold is a submanifold $\tilde{M} \subset \projH$ of the projective Hilbert space, which is equipped with the Fubini-Study metric defined in Def.~\ref{def:fs-metric}.
For each point $[\psi] \in \tilde{M}$, the associated tangent space $\rmT_{[\psi]} \tilde{M}$ is a subspace of $\rmT_{[\psi]} \projH$.
As a result, one can define the \emph{induced} Fubini-Study metric $\tilde{g}$ on $\tilde{M}$ such that for any $[\psi] \in \tilde{M}$ and any $\ket{\phi_1}, \ket{\phi_2} \in \rmT_{[\psi]} \tilde{M}$, one has
\begin{equation}\label{eq:induced-metric}
    \braket{\phi_1}{\phi_2}_{\tilde{g}} = \braket{\phi_1}{\phi_2}_g,
\end{equation}
as done in Sec.~IIIE of Ref.~\cite{haegeman_geometry_2014}.

\vspace{2mm}
\subsection{Maximum likelihood estimation on manifold}\label{sec:mle-manifold}

In this subsection, we review and provide intuitions for existing results on the asymptotic normality of 
\emph{maximum likelihood estimation} (MLE) on a manifold~\cite{hajriMaximumLikelihoodEstimators2017} in Lemma~\ref{lem:asymptotic-normality}. More concretely, \refcite{hajriMaximumLikelihoodEstimators2017} defines an appropriate notion of asymptotic normality of the distance between the estimator and the target on a manifold. Our main result of Theorem~\ref{thm:mle-manifold-normality} follows directly by applying Lemma~\ref{lem:asymptotic-normality} to our analysis of MLE on the manifold of projective MPS, as defined in \appref{sec:manifold}.

Let $(M, g)$ be a Riemannian manifold of dimension $r:=\dim(M)$ with a Riemannian connection $\nabla$ and let $\mathcal{P}_{\theta}$ be a family of distributions parametrized by elements $\theta \in M$. Let $\mathcal{S}$ represent the sample space of the distributions, i.e., $\mathcal{P}_{\theta}: \mathcal{S} \rightarrow [0, 1]$.
Our goal is to identify the target point $\theta_0 \in M$ on the manifold.
To achieve this, we assume access to a \emph{dataset}, which is formally defined as a random variable  $\estn{D}$ consisting of $N$ samples drawn i.i.d.\ from the target distribution $\mathcal{P}_{\theta_0}$.
Given a realization of the dataset $D=(x_i \sim \mathcal{P}_{\theta_0})_{i=1}^{N}$, MLE aims to identify the element on the manifold $\theta^{D}$ for which the associated distribution $\mathcal{P}_{\theta^{D}}$ is \emph{most likely} to have generated the dataset $D$.
This is achieved by minimizing the \emph{negative log-likelihood} (NLL).
More specifically, we define the NLL for the dataset instance $D$, denoted by $\mathcal{L}^{D}_{\text{NLL}}$, as a map from the manifold $M$ to the real numbers,
\begin{align}\label{eq:loss_manifold}
    \mathcal{L}^{D}_{\text{NLL}}: M~ &\longrightarrow \rr, \\
    \theta &\longmapsto  -\frac{1}{N}\sum_{i=1}^{N} \log \mathcal{P}_{\theta}[x_i], 
\end{align}
which measures how likely the observed samples are under the distribution $\mathcal{P}_{\theta}$.
By minimizing the NLL for the dataset $D$ we find the estimator
\begin{equation}\label{eq:min_nll}
    \theta^{D} = \text{argmin}_{\theta\in M} \mathcal{L}^{D}_{\text{NLL}} (\theta).
\end{equation}
For clarity, we use $\mathcal{L}^{D}_{\text{NLL}}(\theta)$ and $\theta^{D}$ to denote the NLL and MLE estimator \emph{for the specific dataset $D$ consisting of $N$ samples}, whereas we use the notation $\est{\theta}$ for the random variable corresponding to the minimizer of NLL associated to the $N$-sample dataset $\estn{D}$ drawn randomly.
In Euclidean space, minimization can be achieved by finding points where the gradient of the loss function vanishes. On a manifold $M$, we can define an estimating form $\omega$, which plays the role of the gradient.

\vspace{2mm}
\begin{definition}[\textbf{Estimating form (\cite{hajriMaximumLikelihoodEstimators2017} Definition 1)}]
    \label{def:est-form}
    A map
    \begin{equation}
        \omega: \mathcal{S} \times M \rightarrow \mathrm{T}M,
    \end{equation}
    is an estimating form if $\mathbb{E}_{x \sim \mathcal{P}_{\theta}}[\omega(x, \theta)] = 0$ for all $\theta \in M$. $\mathrm{T}M$ is the tangent bundle, which is a collection of tangent spaces over all points in the manifold. 
\end{definition}
\vspace{2mm}
A more concrete expression of the estimating form follows from the covariant derivative of the NLL.
\vspace{2mm}
\begin{fact}[\textbf{Differentiable estimating form}]\label{fact:estimating-form}
    The gradient (in local coordinates) of the NLL is a differentiable estimating form
\begin{equation}
    \omega(x, \theta) = - \big(\rmD_{e_1}\log(\mathcal{P}_{\theta}[x]), \ldots, \rmD_{e_r}\log(\mathcal{P}_{\theta}[x])\big)^\top,
\end{equation}
where $\lbrace e_a \rbrace_{a=1}^r$ is a basis of the tangent space $\rmT_\theta M$, and $\rmD_X$ denotes the directional derivative with respect to  a tangent vector $X \in \rmT_\theta M$ as introduced in Def.~\ref{def:directional-derivative}.
\end{fact}
\begin{proof}
To show $\mathbb{E}_{x\sim \mathcal{P}_{\theta}}[\omega(x, \theta)] = 0$, we write down its definition
\begin{align}
    \mathbb{E}_{x\sim \mathcal{P}_{\theta}}[\omega(x, \theta)] & = - \frac{1}{\abs{\mathcal{S}}} \sum_{x \in  \mathcal{S}} \mathcal{P}_{\theta}[x] \Big(\rmD_{e_1} \log (\mathcal{P}_{\theta}[x]), \ldots, \rmD_{e_r} \log (\mathcal{P}_{\theta}[x]) \Big)^\top, \\ 
    & = - \frac{1}{\abs{\mathcal{S}}} \sum_{x \in  \mathcal{S}} \Big(\rmD_{e_1} \mathcal{P}_{\theta}[x], \ldots, \rmD_{e_r}\mathcal{P}_{\theta}[x] \Big)^\top,  \nonumber\\ 
    & = - \Big(\rmD_{e_1} \underbrace{ \frac{1}{\abs{\mathcal{S}}} \sum_{x \in  \mathcal{S}}  \mathcal{P}_{\theta}[x]}_{=1}, \ldots, \rmD_{e_r} \frac{1}{\abs{\mathcal{S}}} \sum_{x \in  \mathcal{S}} \mathcal{P}_{\theta}[x] \Big)^\top, \nonumber\\     
    & = 0. \nonumber
\end{align}
In this definition, the directional derivative is well defined as long as $\mathcal{P}_{\theta}$ is differentiable, while terms with probability $\mathcal{P}_{\theta}[x] = 0$ are to be omitted in the sum, by the limit $y\log y \rightarrow 0$ as $y \rightarrow 0$.
Consequently, the set of estimating forms $\{\omega(x, \theta), \theta \in M\}$ are differentiable vector fields.
\end{proof}
\vspace{2mm}
Additionally, the asymptotic normality of an MLE estimator in Euclidean space requires the Hessian of the loss function to be well behaved around the target parameters. In the manifold setting, we need a similar condition and define the Hessian matrix below.  
\vspace{2mm}
\begin{definition}[\textbf{Hessian matrix~\cite{hajriMaximumLikelihoodEstimators2017}}]
    \label{def:hessian}
    Let $\omega$ be a differentiable estimating form. Choose a basis set $\{e_a\}_{a=1}^{r} \in T_{\theta}M$ in the tangent space of the manifold at $\theta$. The Hessian matrix, in this basis, is defined through the Riemannian connection of $\omega$, and is given by its matrix elements:
    \begin{equation}
        \mathrm{H}_{a,b}(x, \theta) =  \braket{\nabla_{e_a}\omega(x, \theta)}{ e_b}_g.
    \end{equation}
    In subsequent discussions, we often consider the expectation of the Hessian matrix elements over a distribution $\mathcal{P}_{\theta}$ given by
    \begin{equation}
        \mathrm{H}_{a,b}(\theta) = \mathbb{E}_{x \sim \mathcal{P}_{\theta}}\big[ \braket{\nabla_{e_a}\omega(x, \theta)}{ e_b}_g\big].
    \end{equation}
\end{definition}
\vspace{2mm}
Having reviewed the relevant concepts, we are now prepared to state the theorem concerning asymptotic normality of MLE on a manifold following Ref.~\cite{hajriMaximumLikelihoodEstimators2017}.
Let us adopt the following notation:
given an estimate $\theta \in M$, we denote by $\gamma_\theta(t)$ the shortest connecting geodesic (see Def.~\ref{def:geodesic}) such that $\gamma(0) = \theta_0$ and  $\gamma(1) = \theta$,
i.e., $\gamma_\theta$ is the shortest path connecting the MLE estimate to the target.
Moreover, let $\{e_a\}_{a=1}^{r}$ be a basis of the target tangent space $\rmT_{\theta_0} M$ and given a shortest connecting curve $\gamma_\theta$, denote the parallel transport (see Def.~\ref{def:transport}) of the basis vectors by $\{e_a(t)\}_{a=1}^{r}$, i.e., $e_a(0) = e_a$ for all $a$.
This allows us to decompose tangent vectors with respect to  well-defined basis vectors along the \emph{entire} curve $\gamma_\theta$.
\vspace{2mm}
\begin{lemma}[\textbf{Asymptotic normality (\cite{hajriMaximumLikelihoodEstimators2017} Theorem 2)}]
    \label{lem:asymptotic-normality}
    Let $\omega$ be the differentiable estimating form as defined in Def.~\ref{def:est-form} and $\mathrm{H}$ be the Hessian matrix defined in Def.~\ref{def:hessian}. Let $(x_i)_{i=1}^{N}$ be a sequence of i.i.d.\ samples from $\mathcal{P}_{\theta_0}$, where $\theta_0 \in M$ is the target element.
    Denote by $\est{\theta}$ the MLE estimator, which is the minimizer of the NLL, and let $\gamma$ be the connecting geodesic between $\gamma(0)=\theta_0$ and $\gamma(1)=\est{\theta}$.
    Furthermore, consider the following conditions:
    \begin{enumerate}
        \item The estimator $\est{\theta}$ satisfies $\sum_{i=1}^N \omega(x_i, \est{\theta}) = 0$ (first-order condition) for all $N$ and converges in distribution to the target $\hat{\theta}^{(N)} \overset{\mathrm{p}.}{\longrightarrow} \theta_0$ as $N \rightarrow \infty$.
        \item The expected Hessian matrix $\mathrm{H}_{a,b}(\theta_0)$ at the target point, defined in a basis $\{e_a\}_{a=1}^{r}$, is finite and invertible.
        \item The Hessian $\mathrm{H}(x, \theta)$ in a neighborhood of $\theta_0$ converges uniformly to $\mathrm{H}(x, \theta_0)$ as $\theta \rightarrow \theta_0$ in the following sense:
        \begin{equation}
            \mathbb{E}_{x \sim \mathcal{P}_{\theta_0}} \left[
            \underset{\theta \in B_{\theta_0}(\delta)}{\sup} \Big\vert
                 \braket{\nabla_{e_a(1)} \omega(x, \theta)}{ e_b(1)}
                -  \braket{\nabla_{e_a}\omega(x, \theta_0)}{e_b}
            \Big\vert \right] \overset{\mathrm{p.}}{\longrightarrow} 0
        \end{equation}
        as $\delta \rightarrow 0$, where $B_{\theta_0}(\delta)$ is the set of all points with distance at most $\delta$ from $\theta_0$ (the ``$\delta$-Ball'' around $\theta_0$).
    \end{enumerate}
    When these conditions are satisfied the components $\Delta_1, \ldots, \Delta_r$ of the velocity vector $\dot{\gamma}(t) = \sum_{a}\Delta_a e_a(t)$ converge in distribution to a multivariate normal distribution
    \begin{equation}
        \sqrt{N}(\Delta_1, \ldots, \Delta_r)^\top \overset{\mathrm{d.}}{\longrightarrow} \mathcal{N}(0, \Sigma),
    \end{equation}
    where $\Sigma = (\mathrm{H}^\dagger)^{-1}\Gamma \mathrm{H}$ and the Gram matrix $\Gamma_{a,b} = \mathbb{E}_{x \sim \mathcal{P}_{\theta_0}}\big[\braket{\omega(x, \theta_0)}{ e_a}\braket{\omega(x, \theta_0)}{e_b}\big]$.
\end{lemma}

\vspace{2mm}
\subsection{Main result: A concentration inequality for infidelity}\label{sec:infidelity-bound}
In this subsection, we prove our main result that the infidelity between an MLE estimator and the target state can be probabilistically upper bounded when the number of samples $N$ is asymptotically large. Theorem~\ref{thm:mle-manifold-normality} states that such an upper bound depends on the properties of the target state and converges to zero with a rate depending on the number of samples as $1/\sqrt{N}$. 
A central component of our proof relies on the asymptotic normality of the MLE estimator on a manifold discussed in Lemma~\ref{lem:asymptotic-normality}.
Before rigorously proving our theorem, we first integrate the concepts from previous sections and restate the normality condition in the context of the projective MPS manifold $\tilde{M}$, as discussed in Appendix~\ref{sec:manifold} and \ref{sec:mle-manifold}.

Consider the target state $[{\phi}] \in \tilde{M}$, with a normalized representative $\ket{\phi} \in \mathcal{H}$. Given a unitary ensemble $\mathcal{U}$, the sample space $\mathcal{S}$ consists of all unitary-bitstring pairs $\mathcal{S} = \mathcal{U}\times \{0,1\}^n$. Let $\mathcal{D}$ denote the data probability distribution corresponding to the target state $[{\phi}]$ given by
\begin{equation}\label{eq:phi_distributionD}
    \mathcal{D}[(U, b), [\phi]] = \frac{1}{\abs{\mathcal{U}}} \times \expval{U\dyad{\phi} U^\dagger}{b}.
\end{equation}
For a dataset $D=\{(U_i, b_i)\}_{i=1}^{N}$ drawn i.i.d.\ from $\mathcal{D}$, its corresponding \emph{negative log-likelihood} (NLL) is defined as
\begin{align}\label{eq:loss_manifold}
    \mathcal{L}^{D}_{\text{NLL}}: \tilde{M}~ &\longrightarrow \rr ,
    \\
    [{\psi}] &\longmapsto  -\frac{1}{N}\sum_{i=1}^N \log \abs{\matrixel{b_i}{U_i}{\psi}}^2, 
\end{align}  
Notice that by introducing a scaling factor $c > 1$, the loss associated with any vector can be arbitrarily reduced through the transformation $\ket{\psi} 
\mapsto c \ket{\psi}$. However, such rescaling does not affect physical observables and is effectively managed by constraining the analysis to the projective Hilbert space. This highlights the practical value of defining the projective MPS and its associated projective manifold. The MLE estimate is given by the minimizer of the NLL
\begin{equation}
    [{\psi}^{D}] = \text{argmin}_{[\psi]\in\tilde{M}} \mathcal{L}^{D}_{\text{NLL}}\left([\psi]\right).
\end{equation}
Using the notation established in \appref{sec:mle-manifold}, we use  $[\est{\psi}]$ to denote the random variable of the MLE estimator, which is the minimizer of the $N$-sample NLL for a dataset random variable $\estn{D}$.
% When the NLL loss function reaches its minimum, the estimating form evaluates to $0$. 
Using Fact~\ref{fact:estimating-form} we find a concrete expression of the estimating form by taking covariant derivative of the NLL as follows.
\vspace{2mm}
\begin{corollary}[\textbf{Differentiable estimating form}]\label{cor:est_form_ll}
 $\omega((U, b), [\psi]) = -\big(\rmD_{e_1}\log \abs{\matrixel{b}{U}{\psi}}^2, \ldots, \rmD_{e_r}\log \abs{\matrixel{b}{U}{\psi}}^2\big)^\top$ is a differentiable estimating form. 
\end{corollary}
\vspace{2mm}
In order to establish the essential property of asymptotic normality of the projective MPS estimator, as required to prove our main theorem, we first adapt the necessary conditions in Lemma~\ref{lem:asymptotic-normality} to our MPS manifold setting as follows.
\vspace{1mm}
\begin{proposition}[\textbf{Consistent estimator}]\label{prop:manifold-consistency}
        Let $\estn{D} =\{ (U_i, b_i)\}_{i=1}^{N}$ be a sequence of i.i.d.\ samples from $\mathcal{D}$ defined by the target state $[\phi]$ and a unitary ensemble $\mathcal{U}$, tomographically complete in the space of variational states on $\tM$. There exists a sequence of $N$-sample NLL minimizers $([\est{\psi}])_{N=1}^\infty$ such that $\sum_{i=1}^N \omega((U_i, b_i), [\est{\psi}]) = 0$ for all $N$ and the estimator converges to the target state in probability $[\est{\psi}] \overset{\mathrm{p.}}{\longrightarrow} [\phi]$.
\end{proposition}    
\begin{proof}
In the infinite-sample limit, the NLL is given by
    \begin{equation}
        \lim_{N \rightarrow \infty} \mathcal{L}^{}_{\text{NLL}}([\psi]) = - \frac{1}{\abs{\mathcal{U}}}\sum_{(U, b) \in \mathcal{U} \times \{0, 1\}^n} \lvert \matrixel{b}{U}{\phi} \rvert^2 \log \Big(\lvert \matrixel{b}{U}{\psi} \rvert^2 / \abs{\mathcal{U}}\Big),
    \end{equation}
    where $\ket{\phi}$ and $\ket{\psi}$ are normalized representations of their rays.
    The NLL loss function achieves the minimum denoted by $[\hat{\psi}]$ \footnote{The minimum is given by the Shannon entropy of the target distribution $\mathcal{D}$, calculated as $-  \sum_{(U, b)}\mathcal{D}[(U, b), [\phi]] \log \mathcal{D}[(U, b), [\phi]]$. This follows from Gibbs' inequality, which, for two discrete probability distributions $\{p_1, \cdots, p_N \}$ and $\{q_1, \cdots, q_N \}$, states that $-\sum_{i=1}^Np_i \log p_i \leq -\sum_{i=1}^N p_i \log q_i$. The equality is achieved when two distributions are the same $p_i = q_i$.} if the probabilities agree for all measurement outcomes
    \begin{equation}\label{eq:matrix_ele_U}
         \lvert \matrixel{b}{U}{\hat{\psi}} \rvert^2 = \lvert \matrixel{b}{U}{\phi} \rvert^2, \quad \forall (b, U).
    \end{equation}
    Since the unitary ensemble $\mathcal{U}$ is tomographically complete, \equref{eq:matrix_ele_U} implies $\ket{\hat{\psi}} = c \ket{\phi}$ up to a global phase $c$, so that $[\hat{\psi}] = [\phi]$.   
\end{proof}

\vspace{2mm}
\begin{assumption}\label{assume:hessian-invertible}
    The expected Hessian at the target state $[\phi]$ is invertible and finite.
\end{assumption}
\begin{assumption}\label{assume:hessian-convergence}
        The Hessian at $[\est{\psi}]$ in a small neighborhood near the target converges uniformly to the Hessian at $[\phi]$ in expectation with respect to  $\mathcal{D}$.
\end{assumption}  

\vspace{2mm}
Having established or assumed the conditions outlined in Lemma~\ref{lem:asymptotic-normality}, we are now prepared to formally apply the normality result to the MPS-manifold setting. This application is detailed in Corollary~\ref{cor:asymptotic-normality-mps}, which extends the lemma to our specific context.

\vspace{2mm}
\begin{corollary}[\textbf{Asymptotic normality on the projective MPS manifold}]
    \label{cor:asymptotic-normality-mps}
    Let $\mathcal{D}$ be the distribution corresponding to the target state $[\phi]$ (defined in \equref{eq:phi_distributionD}), and $\estn{D} =\{ (U_i, b_i)\}_{i=1}^{N}$ be a sequence of i.i.d.\ samples from $\mathcal{D}$. Let $[\est{\psi}]$ be the minimizer of the negative log-likelihood function and $\omega$ its assoicated estimating form as defined in Corollary~\ref{cor:est_form_ll}. Define $\gamma$ as the shortest connecting geodesic such that $\gamma(0) = [\phi]$, a point $\gamma(1) = [\est{\psi}]$, and $\dot{\gamma}$ as the velocity vector along the curve. Let the velocity vector be decomposed in the parallel transport---along $\gamma$---basis vectors $\{e_a(t)\}$ as $\dot{\gamma} = \sum_{a}\Delta_a e_a(t)$.
    Suppose the conditions of Proposition~\ref{prop:manifold-consistency}, Assumption~\ref{assume:hessian-invertible}, and Assumption~\ref{assume:hessian-convergence} are satisfied. By Lemma~\ref{lem:asymptotic-normality},
    the velocity vector's components $\{\Delta_a\}_{a=1}^{r}$ converge in distribution to a multivariate normal distribution
    \begin{equation}
        \sqrt{N}(\Delta_1, \ldots, \Delta_r)^\top \overset{\mathrm{d.}}{\longrightarrow} \mathcal{N}(0, \Sigma),
    \end{equation}
    where $\Sigma = (\mathrm{H}^\dagger)^{-1}\Gamma \mathrm{H}$ and the Gram matrix $\Gamma_{a,b} = \mathbb{E}_{(U, b) \sim \mathcal{D}}\big[\braket{\omega((U, b), [\phi])}{ e_a}\braket{\omega((U, b), [\phi])}{e_b}\big]$.
\end{corollary}

\vspace{2mm}
Corollary~\ref{cor:asymptotic-normality-mps} establishes the asymptotic normality of the velocity vector's components along the shortest connecting geodesic $\gamma$ in the projective MPS manifold $\tM$, which is crucial for determining the geodesic length (see Fact~\ref{fact:constant-velocity} and \equref{eq:arc-length}). Since $\tM$ is a submanifold of the projective Hilbert space $\projH$ with an induced metric $\tilde{g}$, lifting $\gamma:[0, 1] \rightarrow \tM$ to this embedding space to $\gamma^\prime:[0,1] \rightarrow \projH$ reveals that its length is always longer than or equal to that of a shortest connecting geodesic in $\projH$; equality occurs only when the geodesic curve in $\projH$ is exactly $\gamma^\prime$. 
As the geodesic curve in $\projH$ defines the Fubini-Study distance, the length of the geodesic thus relates to the state overlap using Lemma~\ref{lem:fubini-study-distance-arc-length}.
This relation allows us to upper bound the infidelity between quantum states by the length of $\gamma$ within $\tM$. Using the properties and asymptotic behavior of these geodesic paths, we provide a bounded estimation of the state infidelity in Lemma~\ref{lem:delta-bound} below. 

\vspace{2mm}
\begin{lemma}[\textbf{Infidelity bounded by geodesic curve length}]
    \label{lem:delta-bound}
    Given two states $[\psi], [\phi]$ in $\tilde{M}$, the projective space of MPS (see Lemma~\ref{lem:projective_manifold}),
    denote the shortest connecting geodesic by $\gamma:[0, 1] \rightarrow \tilde{M}$, such that $\gamma(0) = [\psi]$ and $\gamma(1) = [\phi]$, as defined in Def.~\ref{def:geodesic}.
    The infidelity of $[\psi]$ and $[\phi]$ is bounded by the length of the geodesic:
    \begin{equation}
        \label{eq:infidelity-bound-by-geodesic-length}
        1 - \vert\braket{\psi}{\phi}\vert \leq \int_0^1 \dd t \sqrt{\braket{\dot{\gamma}}{\dot{\gamma}}_{\tilde{g}}},
    \end{equation}
    where $\braket{\bullet}{\bullet}_{\tilde{g}}$ denotes the induced Fubini-Study metric on $\tilde{M}$.
\end{lemma}

\begin{proof}
    By Lemma~\ref{lem:fubini-study-distance-arc-length}, the infidelity between two state vectors $\ket{\psi}$ and $\ket{\phi}$ is upper bounded by their Fubini-Study distance in the projective Hilbert space as
    \begin{align}
        1 - \vert\braket{\psi}{\phi}\vert &= 1 - \cos d_\mathrm{FS}([\phi], [\psi]), \\
        & \leq d_\mathrm{FS}([\phi], 
        [\psi]).
        \nonumber
    \end{align}
    Recall that the Fubini-Study distance is defined as the length of the shortest curve between two points in projective Hilbert space (see Def.~\ref{def:fs-distance-geodesic-length}).
    Hence, for any differentiable 
    curve $\alpha:[0, 1] \rightarrow \projH$, with $\alpha(0) = [\psi]$ and $\alpha(1) = [\phi]$, the length of $\alpha$ also upper bounds the infidelity as
    \begin{equation}
        \label{eq:infidelity-leq-curve-length}
        1 - \vert\braket{\psi}{\phi}\vert \leq \int_0^1 \dd t \sqrt{\braket{\dot{\alpha}}{\dot{\alpha}}_g},
    \end{equation}
    where $\braket{\bullet}{\bullet}_g$ denotes the Fubini-Study metric on $\projH$.
    Note that the choice of the parameter interval $[0,1]$ is without loss of generality because we can reparametrize any connecting curve to an arbitrary closed interval of our choice.
    Since $\tilde{M}$ is a submanifold of $\projH$, we can lift the shortest connecting geodesic $\gamma$ from the projective MPS manifold $\tilde{M}$ into projective Hilbert space $\projH$.
    To be precise, we can formally distinguish the curve $\gamma$ in the two abstract manifolds $\tilde{M}$ and $\projH$.
    For this, we formally identify each point $\gamma(t) \in \tilde{M}$ with its corresponding point in $\gamma^\prime(t) \in \projH$.
    Since $\tilde{g}$ is induced by $g$, the lengths of any tangent vectors are the same in both the submanifold and the embedding manifold $\braket{\dot{\gamma}}{\dot{\gamma}}_{\tilde{g}} = \braket{\dot{\gamma}^\prime}{\dot{\gamma}^\prime}_g$ (see Eqs.~(103) and (108) in Ref.~\cite{haegeman_geometry_2014}).
    Therefore, the 
    lengths of $\gamma$ and $\gamma^\prime$ are equal.
    Choosing the curve to be $\alpha = \gamma^{\prime}$ in Eq.~\eqref{eq:infidelity-leq-curve-length}, 
    we have proved the result stated in Eq.~\eqref{eq:infidelity-bound-by-geodesic-length}.
\end{proof}

\vspace{2mm}
Finally, using Corollary~\ref{cor:asymptotic-normality-mps} and Lemma~\ref{lem:delta-bound}, we provide an upper bound on the infidelity, which is our main result below. 
\setcounter{theorem}{1}
\vspace{2mm}
\begin{theorem}[\textbf{Probabilistic bound}]
Let $\tilde{M}$ be the manifold of full-rank MPS of a maximum bond dimensions $\chi_{\max}$, as per Eq.~\eqref{eq:projectiv-manifold}, and let $[\phi]\in\tilde{M}$ be some target state. Let $D^{(N)}=\{(U_i,\,b_i)\}_{i=1}^{N}$ be the random variable representing a dataset of $N$ samples drawn from the distribution $\mathcal{D}[(U, b), [\phi]]$ defined in Eq.~\eqref{eq:phi_distributionD}. Denote by $[\hat{\psi}^{(N)}]$ the random variable representing the minimizer, in $\tilde{M}$, of the $N$-sample negative log-likelihood defined by $D^{(N)}$, and let $\omega$ be its associated estimating form as defined in Corollary~\ref{cor:est_form_ll}. 
Assume the conditions of Proposition~\ref{prop:manifold-consistency}, Assumption~\ref{assume:hessian-invertible}, and Assumption~\ref{assume:hessian-convergence} are satisfied. Then, for any $\delta\in (0,1]$, there exists $N_0$ such that for all $N> N_0$, one has
\begin{equation}\label{eq:}
        \mathbb{P}\left[1 - \vert \langle \est{ \psi}\vert\phi \rangle \vert \leq \epsilon(N)\right] \geq 1-\delta,
    \end{equation}
where 
    \begin{equation}
        \epsilon(N) = O\left(\sqrt{\frac{n \chi_{\max}^2}{N\delta}}\right). 
        \end{equation}
\end{theorem}

\begin{proof}
% Let us now bound the state infidelity $1 - \vert \langle\est{ \psi}\vert\phi\rangle\vert$.
For simplicity, we assume that $\ket{\est{\psi}}$ and $\ket{\phi}$ are normalized elements of the equivalence classes $[\est{\psi}]$ and $[\phi]$, respectively.
Let $\tilde{g}$ be the Fubini-Study metric of the projective MPS manifold $\tilde{M}$ induced by the Fubini-Study metric of the projective Hilbert space via \equref{eq:induced-metric}. On this manifold, $\gamma$ is the shortest connecting geodesic between $\gamma(0) = [\phi]$ and $\gamma(1) = [\est{\psi}]$. 
By Proposition.~\ref{prop:manifold-consistency}, the estimator  converges to the target state in the infinite-sample limit
\begin{equation}\label{eq:consistent-psin}
    [\est{\psi}] \overset{\mathrm{p.}}{\longrightarrow} [\phi],
\end{equation}
as $N \rightarrow \infty$.
To quantify the error between the $N$-sample estimator and the target state, we use Lemma.~\ref{lem:delta-bound} to upper bound the $N$-sample infidelity by the length of the geodesic curve. More specifically, according to Fact~\ref{fact:constant-velocity}, this geodesic length is the same as the length of the velocity vector, which is constant throughout the curve, so that
\begin{equation}\label{eq:infidelity-curve-length}
    1 - \vert \langle \est{ \psi}\vert\phi \rangle \vert \leq \sqrt{\braket{\dot{\gamma}}{\dot{\gamma}}_{\tilde{g}}}.
\end{equation}
Therefore, to obtain an upper bound for the infidelity, it suffices to analyze the velocity vector defined at an arbitrary point along the curve, which we choose to be $\gamma(0) = [\phi]$. The velocity vector, in the basis $\{ e_a\}_{a=1}^{r}$ of the tangent space $T_{[\phi]}\tilde{M}$, can be written as
\begin{equation}
    \dot{\gamma} = \sum_{a=1}^r \estn{\Delta}_a e_a.
\end{equation}
Expressing the Fubini-Study metric in the same basis as $\tilde{g}_{a,b}$, the squared length of the velocity vector is explicitly given by 
    \begin{equation}
        {\braket{\dot{\gamma}}{\dot{\gamma}}_{\tilde{g}}} = {\sum_{a,b} \estn{\Delta}_a\,\tilde{g}_{a,b}\,\estn{\Delta}_b}.
    \end{equation}
For simplicity, we represent the velocity vector's components as a column vector $\estn{\Delta} = (\estn{\Delta}_1, \ldots, \estn{\Delta}_r)^\top$, so that its squared length is 
\begin{equation}\label{eq:sqaured-length}
    {\braket{\dot{\gamma}}{\dot{\gamma}}_{\tilde{g}}} = {{\estn{\Delta}}^\top\tilde{g}\,\estn{\Delta}},
\end{equation}
where, in slight abuse of notation, we write $\tilde{g} = (\tilde{g}_{a,b})_{a,b=1}^r$.
By virtue of Assumptions~\ref{assume:hessian-invertible} and \ref{assume:hessian-convergence}, together with the consistency of the estimator in \equref{eq:consistent-psin}, Corollary~\ref{cor:asymptotic-normality-mps} states that the $N$-sample estimator $\sqrt{N}\estn{\Delta}$ converges in distribution to the normal distribution as
\begin{equation}
    \sqrt{N}\Delta^{(N)} \overset{\mathrm{d}}{\rightarrow} \Delta^{(\infty)} \sim  \mathcal{N}(0, \Sigma).
\end{equation}

 Therefore, in the infinite-sample limit $N\rightarrow \infty$, we can upper bound the probability of
 % a $\Delta^{(\infty)}$-inner product 
 the norm of $\Delta^{(\infty)}$ (with respect to the metric $\tilde{g}$) greater than an error $\varepsilon$ using Markov's inequality as
 \begin{align}\label{eq:markov_inequality}
    \mathbb{P}\left[ {\Delta^{(\infty)}}^\top \tg\, \Delta^{(\infty)} \geq \varepsilon\right]
    \leq \frac{\mathbb{E}\left({\Delta^{(\infty)}}^\top \tg\, \Delta^{(\infty)}\right)}{\varepsilon}.
\end{align}
To compute the expectation of the $\Delta^{(\infty)}$-inner product on the r.h.s., we use the following identity
\begin{align}
    \mathbb{E}\left({{\Delta^{(\infty)}}^\top \tg \, \Delta^{(\infty)}}\right)
    &= \mathbb{E}\left({\tr(\tg \, \Delta^{(\infty)}{\Delta^{(\infty)}}^\top)}\right), \\
    \nonumber
    &= \tr\Big(\tg\,\mathbb{E}\left({\Delta^{(\infty)}{\Delta^{(\infty)}}^\top}\right)\Big),\\
    & = \tr\left(\tg\,\Sigma\right).
     \nonumber
\end{align}
Going from the second to the third line, first we have used the definition of the covariance matrix
\begin{equation}
    \Sigma = \mathbb{E}\left({\Delta^{(\infty)}}{\Delta^{(\infty)}}^\top\right) - \mathbb{E}\left({\Delta^{(\infty)}}\right)\mathbb{E}\left({\Delta^{(\infty)}}^\top\right).
\end{equation}
Second, using the zero-mean property of $\mathbb{E}\left({\Delta^{(\infty)}}\right)=0$, we arrive at an explicit expression for the expectation of the $\Delta^{(\infty)}$-inner product as
\begin{align}\label{eq:squared-length}
    \mathbb{E}\left({{\Delta^{(\infty)}}^\top \tg \, \Delta^{(\infty)}}\right)
    = \tr\left(\tg\,\Sigma\right).
\end{align}
Plugging the expectation values derived in \equref{eq:squared-length} into the r.h.s. of \equref{eq:markov_inequality}, we have the following bound
\begin{equation}
    \mathbb{P}\left[ {\Delta^{(\infty)}}^\top \tg\, \Delta^{(\infty)} \geq \varepsilon\right]
    \leq \frac{\tr\left(\tg\,\Sigma\right)}{\varepsilon}.
\end{equation}
To analyze the finite-$N$ $\estn{\Delta}$-inner product estimator, we first notice that the estimator converges in probability to its infinite-$N$ limit as
\begin{align}\label{eq:finite_infnite_p}
    \mathbb{P}\left[ N {\Delta^{(N)}}^\top \tg\, \Delta^{(N)} \geq \varepsilon\right] \overset{N \rightarrow \infty}{\longrightarrow}
    \mathbb{P}\left[ {\Delta^{(\infty)}}^\top \tg\, \Delta^{(\infty)} \geq \varepsilon\right].
\end{align}
Once we choose a convergence error tolerance $\eta$, there exists an $N_0$ such that the finite-sample estimator is close to the infinite-sample limit upper bounded by $\eta$ for all $N > N_0$:
\begin{align}\label{eq:}
    \mathbb{P}\left[ N {\Delta^{(N)}}^\top \tg\, \Delta^{(N)} \geq \varepsilon\right] -
    \mathbb{P}\left[ {\Delta^{(\infty)}}^\top \tg\, \Delta^{(\infty)} \geq \varepsilon\right]
     \leq \eta.
\end{align}
Therefore, we can bound the finite-sample estimator by
\begin{align}\label{eq:}
    \mathbb{P}\left[ N {\Delta^{(N)}}^\top \tg\, \Delta^{(N)} \geq \varepsilon\right]
    &\leq
    \frac{\tr(\tg\,\Sigma)}{\varepsilon} + \eta.
\end{align}
Notice that in the l.h.s., there is a factor of $N$; in order to bound the squared length, we can rescale  $\varepsilon' = \varepsilon/N$, so that the following inequality is satisfied for all $N \geq N_0$
\begin{align}\label{eq:}
    \mathbb{P}\left[{\Delta^{(N)}}^\top \tg\, \Delta^{(N)} \geq \varepsilon'\right] \leq \frac{\tr(g\Sigma)}{N\varepsilon'} + \eta.
\end{align}
By choosing a tolerance probability $\delta$ such that the lower bound for the squared length is satisfied,
\begin{align}\label{eq:}
    \frac{\tr\left(\tg\,\Sigma\right)}{N\varepsilon'(N)} + \eta = \delta,
\end{align}
we can solve for the $N$-sample error $\varepsilon^\prime(N)$ as
\begin{equation}
    \varepsilon^\prime(N) = \frac{\tr(\tg\,\Sigma)}{N (\delta - \eta)}.
\end{equation}
Finally, we can upper bound the squared length by taking the complement of the inequality. With probability $1 - \delta$, the squared-length estimator is upper bounded by $\varepsilon^\prime(N)$, i.e.
\begin{equation}\label{eq:}
    \mathbb{P}\left[{\Delta^{(N)}}^\top \tg\, \Delta^{(N)} \leq \varepsilon^\prime(N)\right] \geq 1-\delta.
\end{equation}
To achieve a bound for the infidelity, we take the square root of both sides of the inequality and arrive at an upper bound of the length of  the velocity vector, which bounds the infidelity by \equref{eq:infidelity-curve-length} and \equref{eq:sqaured-length}, so that
\begin{equation}\label{eq:}
    \mathbb{P}\left[1 - \vert \langle \est{ \psi}\vert\phi \rangle \vert \leq \sqrt{\varepsilon^\prime(N)}\right] \geq 1-\delta.
\end{equation}
Finally, we can upper bound the $N$-sample error threshold as
\begin{align}
    \varepsilon^\prime(N) &\leq
    \frac{\dim(\tilde{M})\Vert \tg \,\Sigma\Vert_\infty}{N(\delta - \eta)}, 
    % & = O\left(\frac{n \chi_{\max}^2}{N(\delta - \eta)}\right)
    % \nonumber
\end{align}
where $\Vert\bullet\Vert_\infty$ is the $\infty$-Schatten norm. The dimension of the projective MPS manifold can be upper bounded as $\dim(\tilde{M}) = O(n \chi_{\max}^2)$ where $\chi_{\max}$ is the maximum bond dimension $\chi_{\max} = \max(\chi)$.

The $\eta$-term in the denominator is a consequence of the fact that the estimator $[\est{\psi}]$ is only \emph{asymptotically} normal.
The smaller one wants to make $\eta$, the larger $N_0$ becomes so that we have to move ``closer" to the asymptotic regime.
Asymptotically, i.e., in the limit $N\rightarrow \infty$, we can choose $\eta\rightarrow 0$ to arrive at our main result in Theorem~\ref{thm:mle-manifold-normality}
\begin{equation}\label{eq:}
    \mathbb{P}\left[1 - \vert \langle \est{ \psi}\vert\phi \rangle \vert \leq \epsilon(N)\right] \geq 1-\delta,
\end{equation}
where the infidelity error threshold is given by
\begin{equation}
    \epsilon(N) = O\left(\sqrt{\frac{n \chi_{\max}^2}{N \delta}}\right).
\end{equation}

\end{proof}

\begin{figure}[H]
    \centering
    \includegraphics[width=0.5\textwidth]{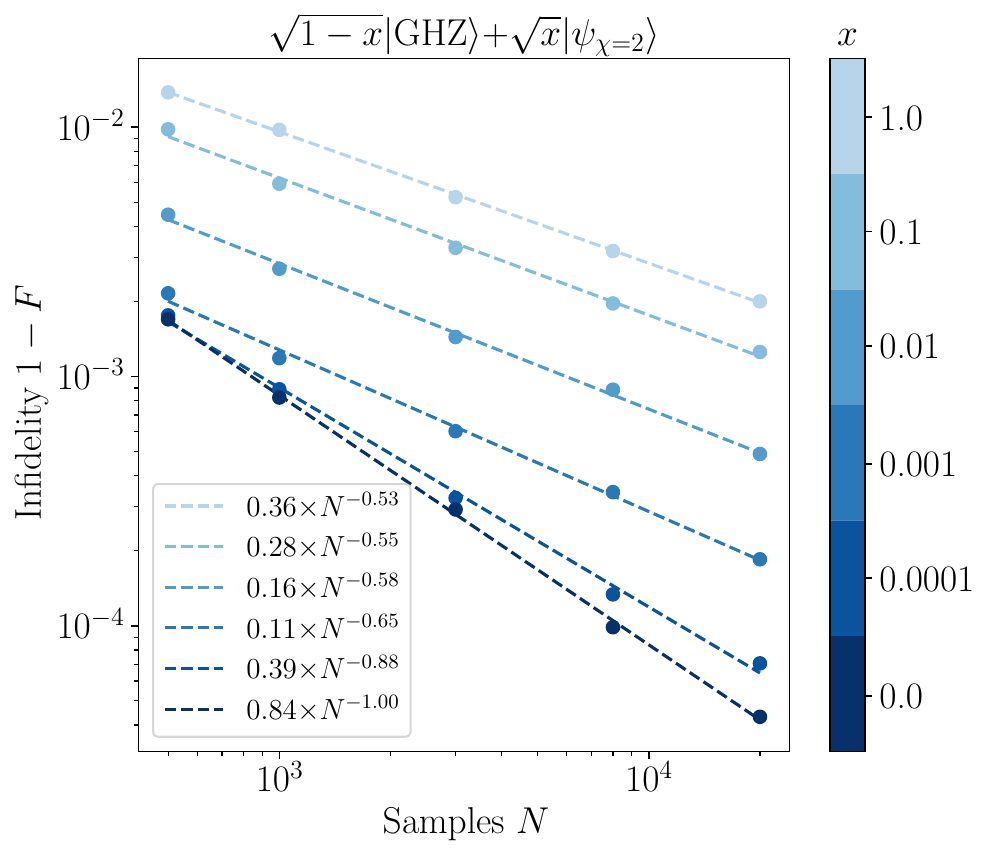}
    \caption{Infidelity scaling for interpolation between a random and GHZ $3$-qubit state. The numerical results illustrate the polynomial scaling of the infidelity with the number of samples $N$ across a family of states, parametrized by $x$. Here, $x$ ranges from 0 (darkest color, representing the GHZ state) to 1 (lightest color, representing the random state). Dashed lines indicate polynomial fits of the form $1-F = c \times N^{-\alpha}$, with values for $c$ and $\alpha$ provided in the legend.}
    \label{fig:appB-infidelity-random-states}
\end{figure}

\newpage
\subsection{Discussions and interpretations}\label{app:mle-discussions}
In this subsection, we discuss Theorem~\ref{thm:mle-manifold-normality} and its connection to the numerical results presented in the main text. This theorem demonstrates that the MPS solution of \emph{maximum likelihood estimate} (MLE) approaches the target state when the number of samples $N$ is sufficiently large, with the infidelity $1-F$ converging at a rate proportional to $1/\sqrt{N}$. This provides a theoretical foundation for understanding the scaling behavior observed in our numerical experiments.

First, we comment on how this theoretical result relates to our numerics in practice. 
In the derivation of our theorem, we have assumed that our variational space forms a projective MPS manifold. However, in our numerical simulations, we fix the bond dimension across sites, including variational states that are not always full rank, although they do include a projective MPS manifold as a subspace. 
In addition, several assumptions underpinning our result could be compromised: 1) the Hessian of the \emph{negative log-likelihood} (NLL) loss function may not behave as expected, or 2) the unitary ensemble is not tomographically complete in our variational space. Notably, while Theorem~\ref{thm:TomCom_realpurestatesWithXZ} confirms that our random-$XZ$ measurements are tomographically complete for real and pure states, our variational approach involves complex numbers for enhanced stability in optimization, leading to a space of complex pure states. To overcome this discrepancy, incorporating random-Pauli measurements, which are known to be tomographically complete, could provide a solution. Alternatively, our analysis of the complex projective MPS manifold could be adapted to the setting of a real projective MPS, and one can verify that the necessary conditions on the manifold structure are satisfied in the subspace of real states.

Focusing on our numerical observations, we empirically find that the infidelity decreases polynomially with the number of samples $N$ as  $1-F \propto \frac{1}{N^\alpha}$. Theorem~\ref{thm:mle-manifold-normality} generally predicts that $\alpha = 0.5$. Specifically, a strong alignment with the polynomial functional form indicates that $N$ is sufficiently large to approach the asymptotic limit. In practice, we see that some of our numerical results for $n=9$ qubits exhibit such a polynomial scaling of the infidelity with $\alpha \approx 1$ in \figref{fig:2_surface_code}(b), while others do not (see \figref{fig:3_rydberg_ruby}(b,c)). The violation of our prediction could be explained by the invalidity of one or more of our assumptions. Intriguingly, $\alpha=1$ is actually better than the upper bound provided in our theorem, pointing to additional complexities.

To address discrepancies between our numerical observations and our theoretical framework, we have performed additional numerical experiments using random states to validate our theorem's predictions. In \supfigref{fig:appB-infidelity-random-states}, we illustrate how the infidelity scales with the number of samples $N$ for states that interpolate between a GHZ state $\ket{\mathrm{GHZ}}$ and a random state vector $\ket{\psi_{\chi=2}}$ with bond dimension $2$. The level of randomness is controlled by the parameter $x$, as defined by
\begin{equation}\label{eq:random-GHZ}
    \ket{\phi} = \sqrt{1-x} \ket{\mathrm{GHZ}} + \sqrt{x} \ket{\psi_{\chi=2}}.
\end{equation}
Our findings reveal that for a purely random state ($x=1$), the observed scaling exponent $\alpha = 0.53$ aligns perfectly with our theoretical prediction. Conversely, at $x=0$, corresponding to the GHZ state, an unexpected deviation with $\alpha = 1$ mirrors the anomaly previously observed in \figref{fig:2_surface_code}(b). This suggests that the inherent structures of stabilizer states, like the GHZ and the surface code state, might facilitate a better scaling than our current upper bound.

Finally, we note that in our practical numerics, we perform MLE in the space of MPS tensors, and not directly on the manifold. Theoretically, the trajectory towards the minimum would differ between these two approaches. Nonetheless, our theorem focuses on the properties of the \textit{optimal} MLE estimator, acknowledging that practically achieving this solution may not always be efficient.
Note, however, that once we have found MPS tensors corresponding to the minimizer of the loss function, they correspond to a \emph{unique} physical state in Hilbert space.
Thus, the theorem still makes a meaningful statement about the practical optimizer with respect to the MPS tensors, under the condition that the optimizer finds a global minimizer of the loss function.
We leave potential improvement of the optimization by incorporating insights from the manifold structure~\cite{hauru2021riemannian, luchnikov2021riemannian} for future work.

\newpage
%%%%%%%%%%%%%%%%%%%%%%%%%%%%%%%%
\section{Random-$XZ$ classical shadow tomography}\label{app:xz_shadow}
%%%%%%%%%%%%%%%%%%%%%%%%%%%%%%%%

This self-contained appendix fully characterizes random-$XZ$ classical shadow tomography and derives the necessary tools to implement it. We use the results outlined in this section when estimating observables in the regularization step of  \equref{eq:loss_function_mps}; see also \appref{app:regularization} for the details of regularization using random-$XZ$ classical shadows. We are also motivated to characterize random-$XZ$ classical shadows because measurements of this form allow one to fully characterize real, pure states, which are the target states in this work (refer to Appendix \ref{app:reality_random_xz_measurements} for a proof). 
Random-$XZ$ classical shadows may also be of interest more broadly due to the ease of their experimental implementation across different platforms. For instance, with superconducting qubits, measurement in different bases only requires single-qubit rotation which are cheap to implement. In the context of Rydberg atom arrays, this holds for both physical and logical qubits. Single-qubit rotations on physical qubits have high fidelities~\cite{evered2023high}.  For logical qubits, transversal Hadamard rotates the logical to the $X$ basis by using single-qubit physical rotations, and even without these transversal Hadamard gates, logical qubits can still be effectively measured in the $X$ basis~\cite{bluvstein2024logical}. In this latter case, performing $XZ$ measurements is not only doable but also provides the ability to perform error detection (e.g., consider the [[15,1,3]] Reed-Muller 3D color code).  
\vspace{1mm}

Random-$XZ$ classical shadows involve preparing some (unknown) state, randomly measuring some $n$-qubit Pauli string composed of only $X$ and $Z$ single-site Paulis, and then post-processing to predict properties of interest. We assume that the chosen Pauli is uniformly sampled from the set of all $n$-qubit strings composed of only $X$s and $Z$s.
As this set of measurements is not informationally complete, random-$XZ$ classical shadows do not allow us to predict any arbitrary observable. However, we can predict some observables of interest, namely, those that contribute to our regularization step. 
In this appendix, we will discuss which observables can be estimated with this technique, discuss how to post-process the data to estimate observables, and how to calculate the sample complexity required for learning Pauli strings. We specifically discuss the sample complexity for learning Pauli strings because these are the observables that we estimate in our regularization step. In the first section of this appendix, we introduce these ideas pedagogically, with the rigorous backings for all our claims provided in the second section. 

\vspace{2mm}

\subsection{Performing random-$XZ$ classical shadow tomography}

Here, we derive an expression for the classical shadow of the quantum state measured by the randomized $XZ$ measurements discussed in \appref{app:reality_random_xz_measurements}. 
Random-$XZ$ classical shadows involve taking many random measurements on some prepared state. Each measurement involves two steps. First, the state is rotated under some unitary $U$ which is sampled from $U \sim \mathcal{U}_{XZ}$, and then, in the second step, the state is measured in the computational basis states $\{\ketbra{b}{b}\}_{b \in \{0,1\}^n}$. This ensemble of unitaries $\mathcal{U}_{XZ}$ exhibits a uniform probability distribution over unitaries $U$ which take the form
\begin{equation}
    U = \otimes_{i=1}^n U^{a_i} 
    \hspace{1mm} \text{ where } a_i \in \{X, Z\}.
\end{equation}
Here, $U^{a_i}$ is the single-qubit unitary which rotates from the computational basis to the (different) basis of the Pauli operator $\sigma_i^{a}$ acting on site $i$. Therefore, acting on the prepared state, the ensemble $\mathcal{U}_{XZ}$ allows one to randomly measure any $n$-qubit Pauli string composed of only $X$ and $Z$ single-site Paulis. Again, this set of measurements is not informationally complete, and as such, we cannot predict any arbitrary observable with random-$XZ$ classical shadows.

In this appendix, we will use the nomenclature developed in 
Ref.~\cite{vankirkHardwareefficientLearningQuantum2022}. In particular, we define the visible space of our random-$XZ$ measurement scheme as the space of all operators $O$ for which we \textit{can} estimate the expectation value $\tr(\rho O)$; we denote this as $\textnormal{\textsf{VisibleSpace}}$($\mathcal{U}$).
Similarly, the invisible space is the space of operators that are invisible to us under the implementable unitaries $\mathcal{U}$, and we label this as $\textnormal{\textsf{InvisibleSpace}}$($\mathcal{U}$). Note that the full operator space $\mathcal{B}(\mathcal{H})$ associated to the Hilbert space $\mathcal{H}$ can be expressed with the following tensor sum:
\begin{equation}
    \mathcal{B}(\mathcal{H}) \simeq \textnormal{\textsf{VisibleSpace}}(\mathcal{U}) \oplus \textnormal{\textsf{InvisibleSpace}}(\mathcal{U})\,,
\end{equation}
where the symbol $\simeq$ stands for ``isomorphic to''. 
As a general rule, a more limited set of implementable unitaries $\mathcal{U}$ corresponds to a smaller $\textnormal{\textsf{VisibleSpace}}(\mathcal{U})$.
In Corollary~\ref{cor:XZvisiblespace}, we find that the random-$XZ$ visible space is spanned by Pauli strings constructed from single-site $\mathds{1}$, $X$, and $Z$ Paulis. We estimate these Paulis and then plug them into our regularization step. 

We can estimate observables $O \in \textnormal{\textsf{VisibleSpace}}(\mathcal{U}_{XZ})$ by post-processing the data to construct an estimate $\hat{\rho}$ of the true density matrix $\rho$. This estimate is called the ``classical shadow'' of the true density matrix, and we can use it to compute $\tr(\hat{\rho} O) \approx \tr(\rho O)$. 
We construct the classical shadow using the data from our measurements. Each measurement snapshot can be characterized by the unitary applied, $U = \otimes_{i}\hspace{1mm} U^{a_i}$, and the computational basis state measured, $\ket{b }= \otimes_i \ket{b_i}$. Here the $i$ indexes the qubits. We can construct a single-shot estimate of our density matrix, $\hat{\rho}_1$, from random-$XZ$ data as 
\begin{eqnarray}\label{eq:xz_shadow_m_inverse}
        \hat{\rho}^{}_1 &=& \mathcal{M}_{XZ}^{-1}(U^\dagger  \dyad{b} U) \\ 
         &=& \otimes_{i=1}^n 
         \left(\tr(U^{a_i \dagger}  \dyad{b_i} U^{a_i}X)X + \tr(U^{a_i \dagger}  \dyad{b_i} U^{a_i}Z)Z + \frac{1}{2} \mathbb{I}\right) 
         \nonumber
         \\
         \nonumber
         &=&  \otimes_{i=1}^n 
          \left(2 U^{a_i \dagger}  \dyad{b_i} U^{a_i} - \frac{\mathbb{I}}{2} \right).
\end{eqnarray}
The subscript $1$ above indicates that the classical shadow was constructed with data from $1$ measurement. 
The channel $\mathcal{M}^{-1}_{XZ}$ is the Moore-Penrose pseudoinverse of the measurement channel $\mathcal{M}_{XZ}$ derived in Proposition \ref{prop:XZmeasurementchannel} in the next subsection.
This channel models the measurement step of the randomized-measurement procedure, and for an infinite set of informationally complete measurements, it guarantees that we will recover our original density matrix: 
\begin{equation}
    \rho = \lim_{N \rightarrow \infty} \hat{\rho}^{}_N,
\end{equation}
where $\hat{\rho}_N$ is the classical shadow constructed from $N$ measurements. One example, of an informationally complete measurement scheme is that of random Pauli measurements, where each qubit is randomly measured in either the $X$, $Y$, or $Z$ basis. Intuitively, since the Pauli strings form a basis of the operator space $\mathcal{B}(\mathcal{H})$, this allows one to predict \textit{any} observable. However, in our particular setting, the random-$XZ$ measurements are not informationally complete, and as such, the infinite-measurement $XZ$ classical shadow will not exactly recover every $\rho$. Instead, with an infinite number of $XZ$ measurements, we will recover the exact density matrix projected onto the visible space. Therefore, we can only learn the observables in the visible space because we only learn this portion of $\rho$. 

The sample complexity of classical shadow protocols is defined in terms of the shadow norm $\|O\|^2_\text{shadow}$, an upper bound on the variance of predicting $\tr(\rho O)$ with this measurement-and-reconstruction technique. 
Of course, to get an actual number of samples needed to predict $\tr(\rho O)$ to some precision $\epsilon$, one would need to employ a concentration inequality. 
While the original classical shadows proposal  demonstrated a median-of-means estimation technique~\cite{huangPredictingManyProperties2020}, this concentration inequality has a large prefactor, which increases the number of samples by almost two orders of magnitude. Instead, we recommend that the interested reader employs Bernstein's inequality, as has been used in Ref.~\cite{vankirkHardwareefficientLearningQuantum2022}, which uses an empirical average over the data from $N$ shots
\begin{equation}\label{eq:xz_shadow_Nshot}
    \hat{\rho}^{}_N = \frac{1}{N} \sum_{j=1}^N \mathcal{M}_{XZ}^{-1}\left(U^{(j)\dagger}  \dyad{b^{(j)}} U^{(j)}\right).
\end{equation}
Here, we represent the unitary applied in the $j$th random measurement and the outcome bitstring of that measurement as $U^{(j)}$ and $b^{(j)}$, respectively.

Just as is the case with random Pauli classical shadows (measuring any Pauli string in a given shot), random-$XZ$ classical shadows are most effective at learning local observables, the intuition being that we more often measure in bases that tell us something about that observable. 
Consider, for example, the observables $X \otimes \mathds{1}^{\otimes n-1}$ and $X^{\otimes n}$. 
As long as we measure the first qubit in the $X$ basis, the measurement will provide information about $\tr(\rho \hspace{1mm} X \otimes \mathds{1}^{\otimes n-1})$. This happens with probability $1/2$: there is a $1/2$ chance of measuring the first qubit in the $X$ basis and a $1/2$ chance of measuring it in the $Z$ basis. 
Now, let us consider the second observable: $X^{\otimes n}$. In order to learn about the expectation value of $X^{\otimes n}$ with our state, every qubit must be measured in the $X$ basis, and this happens with probability $1/2^n$ since there are $n$ qubits and each has a $1/2$ chance of being measured in $X$. 
Therefore, it is much easier to learn $X \otimes \mathds{1}^{\otimes n-1}$ than $X^{\otimes n}$. 
In the next subsection of this appendix, we provide sample complexity upper bounds on Pauli strings in the form of the shadow norm of Pauli strings. This result  is consistent with our intuition that random-$XZ$ classical shadows can more easily learn local Pauli strings than nonlocal Pauli strings. In particular, a $k$-local Pauli string $P_k$ ($k$ non-identity single-site Paulis) exhibits a shadow norm of $\|P_k\|^2_\text{shadow} = 2^k$. 

\subsection{Derivation of tools for random-$XZ$ classical shadows}

In the remainder of this appendix, we derive the tools and claims discussed above for performing random-$XZ$ classical shadow tomography. 
First, we will derive the classical shadow measurement channel $\mathcal{M}_{XZ}$ for the learning procedure described in the previous subsection. 
Once a set of implementable unitaries has been specified, the measurement scheme can be defined in the language of probability density functions. For the random-$XZ$ measurement scheme considered above, we define our ensemble $\mathcal{U}_{XZ}$ to have a uniform probability density function in the space of all possible unitaries rotating into $X$ and $Z$ single-qubit bases. 
This allows us to use Observation 4 from Ref.~\cite{vankirkHardwareefficientLearningQuantum2022} in order to define the $XZ$ visible space. 
Specifically, the image of the $XZ$ measurement channel defines the visible space. 
While the visible space is technically  independent of the underlying probability density function (being an indicator of what can be learned \textit{with the measurement unitaries}), the uniform distribution allows us to characterize the visible space since each unitary has a nonzero probability of being sampled~\cite{vankirkHardwareefficientLearningQuantum2022}. 
In addition to characterizing the visible space, we will also  use the measurement channel to define the shadow norm of Pauli strings.  As the unitaries in the ensemble are Cliffords, the eigenoperators of the measurement channel are Pauli strings, and the eigenvalues are the inverse shadow norms of the corresponding eigenoperators~\cite{ippoliti2023operator, hu2023classical}.

\vspace{1mm}

\begin{proposition} [\textbf{Random-$XZ$ classical shadow measurement channel}]\label{prop:XZmeasurementchannel}
% Consider a measurement primitive $\mathcal{U}$ \cite{huang2020predicting} such that ...
Assume a measurement primitive $\mathcal{U}_{XZ}$ such that the associated probability density function is uniform over random-$XZ$  unitaries.  The random-$XZ$ measurement channel for $n$ qubits can be written as
    \begin{equation}
        \label{eq:XZmeasurmentchannel}
        \mathcal{M}_{\text{\rm XZ}}(A) =   \sum_{P \in \text{\rm Pauli}_{\mathds{1},X,Z}(n)}  \frac{1}{2^{k(P)}}  \text{\rm tr}(A P) P\,,
    \end{equation}
where $\text{\rm Pauli}_{\mathds{1},X,Z}(n)$ is the set of $n$-qubit Pauli strings constructed from only $\mathds{1}$s, $X$s, and $Z$s, and $k(P)$ is the locality of Pauli string $P$. Note that we assume all Pauli strings $P$ above are normalized with respect to the Hilbert-Schmidt inner product: $\text{\rm tr}(P P^\dagger) = 1$.
\end{proposition}
\begin{proof}
For the measurement primitive~\cite{huangPredictingManyProperties2020} assumed above, the random unitaries are uniformly sampled from $\mathcal{U}_{XZ}$, which is a subset of the Clifford group. As such, by Lemma 1 of 
Ref.~\cite{bertoni2022shallow}, the eigenoperators of the corresponding measurement channel $\mathcal{M}_{XZ}$ are Pauli strings.
Therefore, we can express the action of the measurement channel as  $M_{XZ}(P) = \lambda_P P$, where $\lambda_P$ denotes the eigenvalue corresponding to the Pauli string eigenoperator $P$.
It suffices to show that a $k(P)$-local Pauli string $P \in \text{\rm Pauli}_{\mathds{1},X,Z}(n)$ has eigenvalue $\lambda_{P} = 1/2^{k(P)}$, and all other Pauli strings $P \notin \text{\rm Pauli}_{\mathds{1},X,Z}(n)$ have eigenvalue 0. 

To show this, we will derive an expression for $\lambda_P$. Let us start by rearranging the measurement channel. We can start with its original definition from Ref.~\cite{huangPredictingManyProperties2020}, to 
get
\begin{eqnarray}
    \mathcal{M}(P) &=&  \Expect_{U \sim \mathcal{U}_{XZ}} \sum_{b \in \{0,1\}^n} U^\dagger \ketbra{b}{b}  U \hspace{1mm} \Tr\left( U^\dagger \ketbra{b}{b} U P  \right) \\
    \nonumber
    &=& \Expect_{U \sim \mathcal{U}_{XZ}} \Tr_{B}\left[\sum_{b \in \{0,1\}^n} U^\dagger \ketbra{b}{b}U \hspace{1mm} \otimes \hspace{1mm} (U^\dagger \ketbra{b}{b} U \hspace{1mm} P )_{B}   \right] \\
    \nonumber
    & = & \Expect_{U \sim \mathcal{U}_{XZ}} \Tr_{B}\left[U^{\dagger \otimes 2} \hspace{1mm}  \sum_{b} \ketbra{bb}{b,b} \hspace{1mm} U^{\otimes 2} \hspace{1mm} (\mathds{1} \otimes P)\right] \\
    &=& \Expect_{U \sim \mathcal{U}_{XZ}} \sum_{\Tilde{P} \in \text{\rm Pauli}_{\mathds{1},Z}(n)} U^\dagger \Tilde{P} U \hspace{1mm} \tr(P \hspace{1mm} U^\dagger \Tilde{P} U).
    \nonumber
\end{eqnarray}

In the second line, we take the trace over subsystem ``$B$'', which we define to be the second subsystem (i.e., the one with $P$ in the term $\mathds{1} \otimes P$). In the third line, $\text{\rm Pauli}_{\mathds{1},Z}(n)$ is the set of all normalized Pauli strings containing only $\mathds{1}$s and $Z$s.  One can begin to see the transition from line two to line three by considering the single-qubit measurement channel. 
For a single qubit, we have $\sum_{b \in \{0,1\}}\ketbra{b,b}{b,b} = \frac{1}{2} (\mathds{1}\mathds{1} + ZZ)$. This can be extended to $n$ qubits, where the $n$-qubit version $\sum_{b \in \{0,1\}^n} \ket{b,b}\bra{b,b}$ becomes a sum of all length-$n$ $\mathds{1}$,$Z$ strings tensored twice, normalized by a factor of $1/2^n$~\cite{vankirkHardwareefficientLearningQuantum2022}. This normalization does not appear in the final expression above because we absorb it into the two $\tilde{P}$s such that they are all normalized with respect to the Hilbert-Schmidt inner product.

We can further simplify this expression by noticing that $\tr(P \hspace{1mm} U^\dagger \Tilde{P} U)$ will be nonzero only when $P = U^\dagger \Tilde{P} U$. With this observation, the expression becomes
\begin{equation}
    \mathcal{M}(P) = \left(\Expect^{}_{U \sim \mathcal{U}_{XZ}} \sum_{\substack{\Tilde{P} \in \text{\rm Pauli}_{\mathds{1},Z}(n), \\ \text{s.t.}  \Tilde{P} = U P U^\dagger }} 1 \right) \hspace{1mm} P.
\end{equation}
The quantity in parenthesis is simply an expression for $\lambda_P$, the eigenvalue corresponding to Pauli string $P$. It is the probability that one of our random unitaries $U \in \mathcal{U}_{XZ}$ rotates $P$ into a Pauli string $UPU^\dagger $ composed of only $\mathds{1}$s and $Z$s. 
We will separate the $4^n$ Pauli strings into two sets and consider their eigenvalues separately. First, let's consider the Pauli strings $P$ in $\text{\rm Pauli}_{\mathds{1},X,Z}(n)$. All sites in $P$ which have a $\mathds{1}$ will be trivially rotated back to $\mathds{1}$ because $U$ in $U P U^\dagger$ separates into $U = \otimes_{i} U_i$. All that we must consider are the non-identity sites of $P$, and there are $k(P)$ of them. Regardless of whether each site has an $X$ or a $Z$, there is a $1/2$ probability that each site is measured in the corresponding basis. Therefore, the eigenvalue of a $k(P)$-local Pauli string $P \in \text{\rm Pauli}_{\mathds{1},X,Z}(n)$ is $\lambda_{P} = 1/2^{k(P)}$.

Next, let us consider Pauli strings that are \textit{not} in $\text{\rm Pauli}_{\mathds{1},X,Z}(n)$. In other words, such Pauli strings $P$ contain at least one single-site $Y$ Pauli by definition. Since the random-$XZ$ unitaries measure each qubit in either the $X$ or the $Z$ basis, there exists no $U$ in $\mathcal{U}_{XZ}$ that would rotate $P$ to a string in $\text{\rm Pauli}_{\mathds{1},Z}(n)$. Therefore, there is no $\Tilde{P} \in \text{\rm Pauli}_{\mathds{1},Z}(n)$ such that $\Tilde{P} = U P U^\dagger$, so the eigenvalue of any Pauli string $P$ containing at least one single-site $Y$ Pauli is $\lambda_P = 0$.
    
\end{proof}

\begin{corollary}  [\textbf{Random-$XZ$ shadow norm of Pauli strings}] \label{cor:XZshadownormPauliStrings}
Assume a measurement primitive $\mathcal{U}_{XZ}$ such that the probability density function $p(U)$ is uniform over random-$XZ$  unitaries. Then, the shadow norm of a Pauli string $P \in \text{\rm Pauli}_{\mathds{1},X,Z}(n)$ is 
    \begin{equation}
    \|P\|^2_\text{shadow} = 2^{k(P)},
    \end{equation}
where $P$ is normalized with respect to the Hilbert-Schmidt inner product and $k(P)$ is the locality of $P$. Otherwise, for all Pauli strings $P \notin \text{\rm Pauli}_{\mathds{1},X,Z}(n)$, the shadow norm is $\|P\|^2_\text{shadow} = \infty$, and, therefore, the observable cannot be learned.
\end{corollary}
\begin{proof}
The shadow norm of an eigenoperator $O$ of the measurement channel $\mathcal{M}$ is $\|O\|^2_\text{shadow} = 1/\lambda_O$, where $\lambda_O$ is the eigenvalue of $\mathcal{M}$ corresponding to $O$: $\mathcal{M}(O) = \lambda_O O$ (see Ref.~\cite{ippoliti2023operator} for a proof). The eigenoperators of our $XZ$ measurement channel $\mathcal{M}_{XZ}$ are the Pauli strings (by Lemma 1 of Ref.~\cite{bertoni2022shallow}). Therefore, the result follows from Proposition \ref{prop:XZmeasurementchannel}.
\end{proof}

The intuition behind this result is as follows. Our measurment channel $\mathcal{M}_{XZ}$ has a nontrivial null space, and as such, there are eigenoperator(s) of this channel whose eigenvalue is zero. By definition, these eigenoperators can \textit{never} be learned and are invisible under these measurement settings~\cite{vankirkHardwareefficientLearningQuantum2022}. In the language of probabilities, we do not have access to information about these eigenoperators because the probability of measuring in the eigenbases of these operators is zero. This is further elucidated by considering the shadow norm: the shadow norm of an operator in the null space of the measurement channel will be infinite. Therefore, no matter how many measurements are taken, it can never be learned.

\vspace{3mm}
\begin{corollary}  [\textbf{Visible space of random-$XZ$ measurements}] \label{cor:XZvisiblespace}
    For the set of unitaries in the ensemble $\mathcal{U}_{XZ}$ defined above, the $\textnormal{\textsf{VisibleSpace}}(\mathcal{U}_{XZ})$ on $n$ qubits is the span of $\text{\rm Pauli}_{\mathds{1},X,Z}(n)$, where $\text{\rm Pauli}_{\mathds{1},X,Z}(n)$ is the set of $n$-qubit Pauli strings constructed from only $\mathds{1}$s, $X$s, and $Z$s. 
\end{corollary}
\begin{proof}
    Observation 4 of Ref.~\cite{vankirkHardwareefficientLearningQuantum2022} states that if unitaries are sampled uniformly at random in the ensemble $\mathcal{U}$, then $\text{\rm \textsf{VisibleSpace}}(\mathcal{U})$ is the image of the associated measurement channel $\mathcal{M}$. 
    Following Proposition \ref{prop:XZmeasurementchannel}, the image of the random-$XZ$ measurement channel is the span of $\text{\rm Pauli}_{\mathds{1},X,Z}(n)$: all Pauli strings composed of single-site $X$, $Z$ and $\mathds{1}$s. 
    Therefore, the $\textnormal{\textsf{VisibleSpace}}(\mathcal{U}_{XZ})$ is the span of $\text{\rm Pauli}_{\mathds{1},X,Z}(n)$.
\end{proof}

This result also informs our intuition about the $XZ$ visible space, the space of observables that we can estimate using the unitaries $\mathcal{U}_{XZ}$.
Recall that each classical shadow takes the form $\hat{\rho} = \mathcal{M}_{XZ}^{-1}(U^{\dagger}\ketbra{b}{b}U)$. Definition 3 in Ref.~\cite{vankirkHardwareefficientLearningQuantum2022} states that $U^{\dagger}\ketbra{b}{b}U \in \textnormal{\textsf{VisibleSpace}}(\mathcal{U}_{XZ})$, and per the observations above, the $\mathcal{M}$ channel is effectively block diagonal with respect to the visible and invisible spaces. Therefore, $\hat{\rho} \in \textnormal{\textsf{VisibleSpace}}(\mathcal{U}_{XZ})$, and in the infinite-measurement limit, the classical shadow will exactly be the projection of $\rho$ onto  $\textnormal{\textsf{VisibleSpace}}(\mathcal{U}_{XZ})$.  As such, we can define the \textit{projected density matrix} $\rho_{XZ}$, as the portion of the true density matrix $\rho$ that lives in the $\textnormal{\textsf{VisibleSpace}}(\mathcal{U}_{XZ})$
\begin{equation}
    \rho^{}_{XZ} = \sum_{P \in \text{\rm Pauli}_{\mathds{1},X,Z}(n)}\tr(P \rho) P.
\end{equation}
This projected density matrix represents the part of $\rho$ that is learnable using random-$XZ$ measurements.
Consequently, when we use classical shadows generated by the $XZ$ unitaries $\mathcal{U}_{XZ}$ to estimate $\ev{O}$ used in our regularization step in \appref{app:loss_function}, we find that we can only estimate operators $O$ that live in the visible space.

\vspace{2mm}
% \yan{I moved the details about M channel inverse to here}
Finally, we conclude this appendix by using the results of the derivations above to find the expression for the inverse of the measurement channel $\mathcal{M}_{XZ}$. This inverse measurement channel is a crucial tool because, as discussed at the start of this appendix, it is used to compute the classical shadow. 
Here we will show that this inverse measurement channel $\mathcal{M}_{XZ}^{-1}$ takes the form stated in \equref{eq:xz_shadow_m_inverse} in the previous section. First, we can re-express the forward channel $\mathcal{M}_{XZ}$ in Proposition \ref{prop:XZmeasurementchannel} in terms of independent, single-site channels $\mathcal{M}_{XZ}^{(1)}$. This channel factorizes as  
\begin{equation}
    \mathcal{M}^{}_{XZ}(A) = 
    \left(\mathcal{M}_{XZ}^{(1)} \right)^{ \otimes n } (A) ,
    \hspace{5mm}\textnormal{ where } \hspace{1mm}
    \mathcal{M}_{XZ}^{(1)}(B) = \frac{1}{4}\left( \tr(BX)X + \tr(BZ)Z \right) + \frac{1}{2}\tr(B)\mathbb{I} \hspace{3mm} 
\end{equation}
for all $\hspace{1mm} A $ and $B$ in the $n$-qubit and single-qubit visible spaces, respectively.
Explicitly, an operator in the single-qubit visible subspace can be decomposed as follows: 
\begin{equation}\label{eq:B-decomposition}
    B = \frac{1}{2} \left( \tr(BX)X + \tr(BZ)Z + \tr(B) \mathbb{I}\right).
\end{equation}
Next, we notice that each single-site channel $\mathcal{M}_{XZ}^{(1)}$ is not only independent, but also the same on every site. Thus, all that is left is to take the Moore-Penrose pseudoinverse of this single qubit channel. This yields the following result:
\begin{equation}\label{eq:m-inverse}
    {\mathcal{M}^{-1}}_{XZ}^{(1)}(B) = \left(\tr(BX)X + \tr(BZ)Z + \frac{1}{2}\tr(B) \mathbb{I}\right),
\end{equation}
which matches our result in equation  \equref{eq:xz_shadow_m_inverse}.
% is well-defined as long as we restrict to the visble subspace and 
% For operators $B$ in the visible space (i.e. $B \in \textnormal{span}(\{\mathds{1}, X, Z\})$), we can express 
% \begin{equation}\label{eq:xz_shadow_m_inverse}
        % \hat{\rho}^{}_1 = \mathcal{M}_{XZ}^{-1}(U^\dagger  \dyad{b} U) = \otimes_{i=1}^n \left(2 U^{\alpha_i \dagger}  \dyad{b_i} U^{\alpha_i} - \frac{\mathbb{I}}{2} \right)
% \end{equation}

% Using this inverse, we can verify that ${\mathcal{M}^{-1}}_{XZ}^{(1)}(\mathcal{M}_{XZ}^{(1)}(B)) = B$ for all $B$ in the single qubit visible space (i.e. $B \in \textnormal{span}(\{\mathds{1}, X, Z\})$). For example, to compute the inverse channel of a classical shadow measured in Z basis with outcome $b$, we first note that $\dyad{b} = 1/2\left(\mathbb{I} - Z\right)$. Then using the definition in \equref{eq:m-inverse}, we see
% \begin{eqnarray}\label{eq:m-inverse-example}
% %     {\mathcal{M}^{-1}}_{XZ}^{(1)}(\dyad{b}) &=& \frac{\mathbb{I}}{2} - Z, \\
%     &=& 2 \dyad{b} - \frac{\mathbb{I}}{2}.
% \end{eqnarray}
% Performing such pseudoinverse for each site, we recover the channel expression used in \equref{eq:xz_shadow_m_inverse}.  

%%%%%%%%%%%%%%%%%%%%%%%%%%%%%%%%%%%%
\newpage
\newpage
\section{Tomographic completeness of randomized $XZ$ measurements}\label{app:reality_random_xz_measurements}
%%%%%%%%%%%%%%%%%%%%%%%%%%%%%%%%%%%%

The target states considered in this work are real pure states. We are interested in such states because they are eigenstates of a Hamiltonian whose matrix elements are purely real, e.g., a Hamiltonian that is equivalent to a composition of $X$ and $Z$ operators up to a global rotation. For example, if we have a (complex) Hamiltonian $H = \sum_{i,j} Z_i Z_j + \sum_i Y_i$, then we can perform a rotation identically at each site such that $Z_i \rightarrow Z_i$ and $Y_i \rightarrow X_i$. These Hamiltonians, which are interesting from the perspective of both condensed matter physics and quantum optimization \cite{Ebadi2022Quantum,PhysRevResearch.6.013271,PRXQuantum.5.020313}, can  suffer from a sign problem~\cite{troyer2005computational}, so finding their ground states is a highly nontrivial task~\cite{hangleiterEasingMonteCarlo2020}.

In this appendix, we will explore how and when we can fully characterize real pure states. These states can be expressed as a pure-state vector $\ket{\psi} = \sum_b c_b \ket{b}$ with real coefficients $c_b^* = c_b$. For simplicity, we will assume that the basis $\{\ket{b}\}_{b\in\{0, 1\}^n}$ is the computational basis, unless indicated otherwise. 
We will demonstrate that we can fully characterize these states by only measuring each qubit in either the $X$ or the $Z$ basis (Lemma~\ref{lem:realpurestatesWithXZ}). 
While our numerical demonstrations are not always limited to measurements only in the $X$ or $Z$ bases, our motivation for specifically pursuing measurement schemes of this type is that randomized-$XZ$ measurements are tomographically complete on the subspace of real and pure states (Theorem~\ref{thm:TomCom_realpurestatesWithXZ}).

Below, this appendix is separated into two parts. First, we will provide an overarching discussion of our formal statements in order to develop our intuition and contextualize them in our broader work. We will then provide rigorous proofs of our claims.

\subsection{Discussion of results}

% \noindent \wwc{Paragraph to introduce XZ measurements: }
As elaborated on in the previous appendix, we can parametrize the learning power of a measurement protocol by looking at the visible space. 
Recall that the visible space is the space of all operators $O$ for which we can estimate the expectation value $\tr(\rho O)$. 
In particular, in Corollary \ref{cor:XZvisiblespace}, we defined the visible space corresponding to the unitaries in the random-$XZ$ ensemble $\mathcal{U}_{XZ}$. 
We found that the visible space is \textit{not} the entire operator space $\mathcal{B}(\mathcal{H})$. In other words, there are some observables that are invisible to us---they can never be learned.  With random-$XZ$ measurements, we can measure any Pauli string in $\text{\rm Pauli}_{\mathds{1},X,Z}(n)$, the set of  $n$-qubit Pauli strings constructed from only single-site $\mathds{1}$, $X$, and $Z$ Paulis. Therefore, the visible space is the span of all such Pauli strings. 
It turns out that this $XZ$ measurement scheme is sufficient for learning states that live in the visible space as well as those with some components outside the invisible space, as long as certain assumptions are satisfied. 
One such example of this is a major result of this work: randomized-$XZ$ measurements are tomographically complete on the subspace of real and pure states (Theorem \ref{thm:TomCom_realpurestatesWithXZ}).

Real, pure states do not live entirely in $\textnormal{\textsf{VisibleSpace}}(\mathcal{U}_{XZ})$. For example, the GHZ states have nontrivial support on Pauli strings with an even number of single-site $Y$ Paulis. Even though these states do not live entirely in $\textnormal{\textsf{VisibleSpace}}(\mathcal{U}_{XZ})$, we can still fully characterize them with access to the visible $XZ$ observables.  This is because we can rely on the assumption that the state is both real and pure. 
Therefore, even though real, pure states are not fully encapsulated in the visible space, the $XZ$ measurement scheme is sufficient for fully characterizing these states.
To gain some intuition, let us consider a real, pure two-qubit state vector $\ket{\psi}$. We can write this state as a density matrix and decompose it in terms of its support on all Pauli strings $P \in \text{\rm Pauli}_{\mathds{1},X,Y,Z}(n)$ as
\begin{eqnarray}
    \ketbra{\psi}{\psi} &=& \frac{1}{2^n} 
    \left( 
    \sum_{P \in \text{\rm Pauli}_{\mathds{1},X,Y,Z}(n)} c^{}_P P
    \right) \\
    &=& \frac{1}{2^n} 
    \left( 
    \sum_{P' \in \text{\rm Pauli}_{\mathds{1},X,Z}(n)} c^{}_{P'} P' 
    + c^{}_{YY} YY
    \right).
    \nonumber
\end{eqnarray}
Since the state is real, its support must be zero on strings with an odd number of $Y$s, and therefore, there is only one term remaining that contains single-site $Y$ Paulis: the two-qubit string $Y \otimes Y$. As a result, this two-qubit state does not live entirely in the visible space because $Y \otimes Y \notin \textnormal{\textsf{VisibleSpace}}(\mathcal{U}_{XZ})$. However, we can still estimate the coefficient $c_{YY}$ using the purity condition: $\Tr(\ketbra{\psi}{\psi}^2) = 1$. In other words, since our $XZ$ measurements allow us to learn all coefficients except $c_{YY}$, we can solve for $c_{YY}$ using this extra constraint. While, so far, we have been considering a real, pure $n=2$ qubit state, it turns out that this holds true for any $n$-qubit real, pure state. The information from the $XZ$ measurements allows us to extract the support of $\ketbra{\psi}{\psi}$ on Pauli strings containing an even number of $Y$s. This will be shown rigorously in the following subsection.

\subsection{Proof of Theorem 1}

% \noindent
%  \wwc{Paragraph: Important takeaways + transition into the proof}
In this subsection, we prove that $XZ$ measurements, where each qubit is measured in either the $X$ or the $Z$ basis, are tomographically complete on the space of real pure states. Tomographic completeness is similar to informational completeness but defined with respect to state characterization. A measurement primitive, which samples $U\sim \mathcal{U}$ and subsequently measures the rotated state in the computational basis $\{\ket{b} : b \in \{0,1\}^n\}$, is tomographically complete if and only if for each $\rho \neq \sigma$, there exist $U \in \mathcal{U}$ and $b$ such that $\bra{b} U \rho U^\dagger \ket{b} \neq \bra{b} U \sigma U^\dagger \ket{b}$ \cite{huangPredictingManyProperties2020}. 
% This is the general definition of tomographic completeness. 
For our purposes, we are only interested in establishing tomographic completeness on a subspace of the density matrices, namely, all states that are both real and pure.

Theorem~\ref{thm:TomCom_realpurestatesWithXZ} states that real, pure states are tomographically complete under $XZ$ measurements. While our theorem might at first be surprising, it demonstrates that there is always a unitary $U \sim \mathcal{U}_{XZ}$ such that its measurement will allow us to differentiate between real, pure states $\rho$ and $\sigma$. Furthermore, it elucidates that the visible space does not immediately render an understanding of tomographic completeness in the presence of extra assumptions.

Before providing the rigorous details, we first outline a sketch of the proof. Theorem \ref{thm:TomCom_realpurestatesWithXZ} follows immediately from Lemma \ref{lem:realpurestatesWithXZ}, which shows that we can fully characterize a real state vector $\ket{\psi}$ up to a global phase using only $XZ$ measurements. 
We prove Lemma \ref{lem:realpurestatesWithXZ} by induction on the number of qubits $k$. Our base case is a real, pure state on $k=1$ qubits. This can be immediately learned by measuring in the $X$ and $Z$ bases. The inductive step ($k$ qubits) involves first learning all the amplitudes $|c_b|$ of $\ket{\psi}$, which can again be expressed as 
\begin{equation}\label{eqn:compbasisrealpurerep}
    \ket{\psi} = \sum_{b \in \{0,1\}^n} c^{}_b \ket{b}.
\end{equation}
We learn the amplitudes by measuring every qubit in the $Z$ basis. Then, we use our inductive assumption to argue that we can always characterize all relative signs using $XZ$ measurements. In particular, our technique is more powerful than simply proving learnability, because we explicitly provide an algorithm by which one can use these measurements to fully characterize a real, pure state up to a global sign. Therefore, we show that $XZ$ measurements are tomographically complete on the space of real pure states. We leave the details of this algorithm to the interested reader (see below).

\begin{lemma}[\textbf{Random-$XZ$ measurements are sufficient for characterizing a real, pure state}]\label{lem:realpurestatesWithXZ}
Up to a global sign,
we can uniquely characterize any real state vector $\ket{\psi}$ using XZ measurements, formally defined as measuring any qubit in either the X or the Z basis. We represent our real pure state in the computational basis $\{\ket{b}\}_{b \in \{0,1\}^n}$ as
$\ket{\psi} = \sum_b c_b \ket{b}$, where $c_b^* = c^{}_b$ and $\sum_b |c_b|^2 = 1$.  We uniquely characterize the state by learning each $c_b$ up to a global sign. 
\end{lemma}
\begin{proof}
To fully characterize the state vector $\ket{\psi}$, we must learn all coefficients $\{c_b\}$ up to a global sign.
Each coefficient is characterized by its amplitude and sign: $c_b = \sgn(c_b) \times |c_b|$. We can extract the amplitudes by measuring all qubits in the $Z$ basis -- whereafter, it remains to learn the relative signs. The remainder of this proof will show how we can learn the relative signs with $XZ$ measurements.  We will proceed by induction on the number of qubits $k$.  

\vspace{1mm}
\textit{Base case ($k=1$ qubits).} We extract the relative sign between $\ket{0}$ and $\ket{1}$ by learning the sign of the expectation value $\expval{X}{\psi}$: $\sgn(\expval{X}{\psi})$.

\vspace{1mm}
\textit{Inductive step ($k$ qubits).}
Assume we can learn all the relative signs of a real, pure state on $k-1$ qubits with $XZ$ measurements. We have a $k$-qubit state. If we measure 
our $k$th qubit in the computational $Z$ 
basis, then, by the inductive assumption, we can learn all the relative signs in the $\{\ket{0} \otimes \ket{\tilde{b}}\}_{\tilde{b} \in \{0,1\}^{k-1}}$ subspace as well as in the $\{\ket{1} \otimes \ket{\tilde{b}}\}_{\tilde{b} \in \{0,1\}^{k-1}}$ subspace. Here, the tilde on $\tilde{b}$ indicates that the bitstring is on $k-1$ qubits instead of $k$. Also note that when we project the $k$th qubit of our state vector $\ket{\psi}$ onto $\ket{0}$ or $\ket{1}$, the state on the remaining $k-1$ qubits is projected onto a real pure state.  

What is, therefore, left to be shown is that we can learn the relative sign between these two subspaces. If the amplitudes of all states in either of the subspaces are zero, then the relative sign is inconsequential because we already know the state up to a global phase. Otherwise, there must exist at least one state in each subspace that has nonzero amplitude. 
Let us choose two computational basis states: $b^{(0)}$ from the $\{\ket{0} \otimes \ket{\tilde{b}}\}_{\tilde{b}}$ subspace and $b^{(1)}$ from the $\{\ket{1} \otimes \ket{\tilde{b}}\}_{\tilde{b}}$ subspace, respectively. We choose them such that they that have nonzero amplitude and, among all possible choices, they have the smallest Hamming weight $\mathbf{d}(b^{(0)},b^{(1)})$; this Hamming weight must be greater than or equal to 1 because the $k$th bit will always differ. We can extract the relative sign of the two spaces by estimating the observable $O \in \textnormal{\textsf{VisibleSpace}}(\mathcal{U}_{XZ})$ given by
        \begin{equation}
            O = \otimes_{i=1}^k O_i; \hspace{4mm}
            O_i = \begin{cases}
                \ketbra{+}{+} &\quad\text{if } i = k ,\\
                \left \rvert b^{(0)}_i\right\rangle\left\langle b^{(0)}_i \right \lvert &\quad\text{if } b^{(0)}_i = b^{(1)}_i, \\
                X &\quad\text{otherwise}, \\
            \end{cases}
        \end{equation}
where $b^{(0)}_i$ is the $i$th bit of bitstring $b^{(0)}$, and likewise $b^{(1)}_i$ is the $i$th bit of $b^{(1)}$.  

Consider first the case when the Hamming weight is $\mathbf{d}(b^{(0)},b^{(1)}) = 1$. Then, since the $k$th qubit must always differ between the two states, all of the remaining bits are the same: $\forall i \ne k,  b^{(0)}_i = b^{(1)}_i$. Therefore, $\expval{O}{\psi}$ takes the form $|c_{b^{(0)}} + c_{b^{(1)}}|^2$. Since we already know $|c_{b^{(0)}}|$ and $|c_{b^{(1)}}|$, and by assumption they are nonzero, we can uniquely identify the relative sign between the subspaces. 
Next, we consider the case when the Hamming weight is $\mathbf{d}(b^{(0)},b^{(1)}) > 1$ and learn $\expval{O}{\psi}$. Notice that $O$'s projection of $\ket{\psi}$ onto the matching bits in $b^{(0)}$ and $b^{(1)}$ leaves only two (nonzero amplitude) basis state components: 
\begin{equation}
    \left(\otimes_{i \textnormal{ s.t. } b^{(0)}_i = b^{(1)}_i} \left\langle b^{(0)}_i \right \lvert \right) \ket{\psi}
    = 
    c^{}_{b^{(0)}}  \left(\otimes_{i \textnormal{ s.t. } b^{(0)}_i = b^{(1)}_i} \left\langle b^{(0)}_i \right \lvert \right) \ket{b^{(0)}} + c^{}_{b^{(1)}} \left(\otimes_{i \textnormal{ s.t. } b^{(0)}_i = b^{(1)}_i} \left\langle b^{(0)}_i \right \lvert \right) \ket{b^{(1)}}.
\end{equation}
This is because we chose the pair $b^{(0)}$ and $b^{(1)}$ such that they have the smallest Hamming weight among all possible pairings. Then, the $X$ operator flips any remaining bits at $i \ne k$, and we obtain a state with precisely swapped amplitudes
\begin{equation}
    \left(\otimes_{i \textnormal{ s.t. } b^{(0)}_i = b^{(1)}_i} \left\langle b^{(0)}_i \right \lvert \right) \left(c^{}_{b^{(0)}} \ket{b^{(1)}} + c^{}_{b^{(1)}} \ket{b^{(0)}}\right).
\end{equation}
Consequently, evaluating the expectation value yields $\expval{O}{\psi} = c_{b^{(0)}} \times c_{b^{(1)}}$. Since $|c_{b^{(0)}}|$ and $|c_{b^{(1)}}|$ are known and (by assumption) nonzero, we can uniquely identify the relative sign between the subspaces.

\end{proof}

\setcounter{theorem}{0}
\noindent
\begin{theorem}[\textbf{Tomographic completeness of random-$XZ$ measurements on the space of real, pure states}]\label{thm:TomCom_realpurestatesWithXZ}
A measurement primitive which samples $U\sim \mathcal{U}_{XZ}$ and subsequently measures the rotated state in the computational basis $\{\ket{b} : b \in \{0,1\}^n\}$  is tomographically complete \cite{huangPredictingManyProperties2020} on the space of real, pure states. Namely, for all states $\ketbra{\psi}{\psi} \neq \ketbra{\phi}{\phi}$, there exist $U \in \mathcal{U}_{XZ}$ and $b$ such that $\bra{b} U \rho U^\dagger \ket{b} \neq \bra{b} U \sigma U^\dagger \ket{b}$.
\end{theorem}

\begin{proof}
    We proceed by contradiction. Assume that this measurement primitive is not tomographically complete on the space of real, pure states. 
    Then, there exist real, pure states $\ketbra{\psi}{\psi} \neq \ketbra{\phi}{\phi}$ such that
    \begin{equation}
        \forall\, U, b \hspace{2mm} 
        \bra{b} U \ketbra{\psi}{\psi} U^\dagger \ket{b} = \bra{b} U \ketbra{\phi}{\phi} U^\dagger \ket{b}.
    \end{equation} 
    In other words, after the application of any $U$, the resulting computational basis state probabilities will always be the same: we cannot differentiate between $\ketbra{\psi}{\psi} $ and $ \ketbra{\phi}{\phi}$ with this measurement primitive.
    However, Lemma \ref{lem:realpurestatesWithXZ} shows that this same measurement primitive can uniquely characterize, up to a global phase, both $\ket{\psi}$ and $\ket{\phi}$. 
    Therefore, since these states differ by more than just some global phase ($\ketbra{\psi}{\psi} \neq \ketbra{\phi}{\phi}$ by definition), we arrive at a contradiction.
\end{proof}

%%%%%%%%%%%%%%%%%%%%%%%%%%%%%%%%%%
\newpage
\section{Loss function}\label{app:loss_function}
%%%%%%%%%%%%%%%%%%%%%%%%%%%%%%%%%%

This appendix presents the details of the loss function in \equref{eq:loss_function_mps}.
In the first subsection, we provide some background for maximum likelihood estimation and derive its contribution as the first term in our loss function. In the next subsection, we explicitly state the physical observables that we use for regularization in our numerical results. While keeping most of the section self-contained, some of the notations and terminologies used have been set up in \appref{app:xz_shadow}.

\subsection{Maximum likelihood estimation: Randomized Kullback–Leibler divergence}\label{app:random_KL}
First, we review the standard framework of \emph{maximum likelihood estimation} (MLE), before applying such a loss function to our specific setting of quantum measurements. In general, we observe data from an underlying probability distribution $\mathcal{D}: X \rightarrow \mathbb{R}$ on the dataset domain $X$, and our goal is to find the most likely model with parameters $A$ denoted as $\mathcal{Q}_{A}: X \rightarrow \mathbb{R}$ that is consistent with our observations. To find such a model, we would want to variationally optimize for the model's parameters $A$ by minimizing the Kullback–Leibler (KL) divergence
\begin{align}\label{eq:kl_general}
    K(\mathcal{D}_{}| \mathcal{Q}_{A}) &= -\sum_{x\in X} \mathcal{D}_{}(x) \log {\frac{\mathcal{Q}_{A}(x)}{\mathcal{D}_{}(x)}}.
\end{align}
Since the data distribution $\mathcal{D}$ is a constant, we just need to minimize the cross-entropy contribution to the KL divergence
\begin{align}\label{eq:cross_entropy_general}
    \mathds{H}(\mathcal{D}_{}| \mathcal{Q}_{A}) &= -\sum_{x \in X} \mathcal{D}_{}(x) \log {\mathcal{Q}_{A}(x)}. 
\end{align}

More specifically, for our purposes, the dataset consists of bitstring measurements of a target state $\dyad{\phi}$ in the computational basis after applying a unitary $U$. In a randomized measurement scheme, such as the $XZ$-shadow measurements discussed in \appref{app:xz_shadow}, we first draw a unitary from an ensemble $\mathcal{U}$ with a uniform probability over a set $\{U\}$. Then, we measure the bitstring $b$ in the computational basis. Taken together, we sample a data point, which is a tuple of the sampled unitary and the measured bitstring $x_i=(U_i, b_i)$, with probability 
\begin{equation}\label{eq:data_distribution_bitstring_unitary}
\mathcal{D}[(U, b)] = \frac{1}{\abs{\{U\}}} \times \expval{U \dyad{\phi} U^\dagger}{b}.
\end{equation}
Similarly, we define the probability of observing this data point from a model parametrized by $A$ as $\mathcal{Q}_{A}$.
Then, the \textit{randomized} KL divergence is 
\begin{subequations}
    \begin{align}\label{eq:kl_randomized_bitstring}
    K(\mathcal{D}| \mathcal{Q}_{A}) &= - \Expect^{}_{U \sim \mathcal{U}} \sum_{b \in \{0, 1\}^n} \mathcal{D}(U, b) \log {\frac{\mathcal{Q}_{A}(U, b)}{\mathcal{D}(U, b)}}.
    \end{align}
\end{subequations}
As in \equref{eq:cross_entropy_general}, we can drop the constant factor from the data distribution and obtain the randomized cross entropy or the infinite-sample loss function
\begin{equation}\label{eq:cross_entropy_bitstring}
    \mathds{H}(\mathcal{D}| \mathcal{Q}_{A}) = - \Expect^{}_{U \sim \mathcal{U}} \sum_{b \in \{0, 1\}^n} \mathcal{D}(U, b) \log \left[\mathcal{Q}_{A}(U, b)  \right].
\end{equation}
 For $N$-samples i.i.d.\ drawn from $\mathcal{D}$ as our dataset denoted by $D=\{(U_i, b_i)\}_{i=1}^{N}$, this leads to the negative log-likelihood (NLL) contribution to our loss function which is
\begin{equation}\label{eq:loss_finite_bitstring}
L^{D}_{\textrm{NLL}}(A) = -\frac{1}{N}\sum_{i=1}^{N} \log \mathcal{Q}_{A}(b_i, U_i).
\end{equation}
In this work, we focus on a model distribution given by a MPS $\ket{\psi_{\chi}(A)}$ with tensors $A$ and bond dimension $\chi$ such that
\begin{equation}\label{eq:model_distribution_bitstring_unitary}
\mathcal{Q}_{A}(U, b) = \frac{1}{\abs{\{U\}}} \times \abs{\bra{b}U \ket{\psi_{\chi}(A)}}^2.
\end{equation}
Dropping the constant normalization factor from the uniform distribution over the unitary ensemble gives us the MLE contribution described in \equref{eq:loss_function_mps} of the main text.

\subsection{Regularization}\label{app:regularization}

In this section, we discuss the second contribution to our loss function, which additionally regularizes the optimization landscape of our training process. Note that these estimations from finite-$N$ samples have deviations from the true physical expectation values; in general, we do not want our regularization strength to be dominant. 

In the main text, we numerically explore different regularization strengths. Moreover, we have focused on two measurement schemes: global-$XZ$ measurements, and random-$XZ$ measurements. For the former, we measure bitstrings in two fixed bases $X^{\otimes n}$ and $Z^{\otimes n}$, and we can estimate certain physical observables using empirical means. For the latter, using the toolbox of randomized measurements, we apply random unitaries from an ensemble and measure bitstrings in the computational basis. Depending on the choice of the unitary architectures, different physical properties can be estimated. Performing classical shadow tomography using random-$XZ$ measurements---as detailed in \appref{app:xz_shadow}---allows us to estimate any observables in the visible space of our measurement. 
We now discuss our choices of regularizers for the two systems considered in the main text: 

\textit{Perturbed surface code.} Here, we use the stabilizers $\mathcal{S}$ of the exact surface code defined in \appref{app:state_surface_code}. The intuition behind such a choice is that we know that the stabilizer expectations will change continuously with the perturbation and that these operators would have large weight in the density matrix. The regularization contribution to the loss function is
\begin{equation}\label{eq:reg_surface_code}
    R(A) = \sum_{\substack{S \in  \mathcal{S}}} \abs{\ev{S}_{N} - \expval{S}{\psi(A)} }^2,
\end{equation}
where $\ev{S}_{N}$ are estimates from the measurements and the second term is the model's prediction. The regularizer $R(A)$ would be modulated by $\beta$, which is a hyperparameter controlling its strength. The knowledge of these specific observables allows us to estimate their expectation values in two measurement schemes.

With global-$XZ$ measurements in the bases $X^{\otimes n}$ and $Z^{\otimes n}$, we estimate the stabilizer expectations using empirical means. For instance, to estimate $S_{Z, 3}=Z_1 Z_2$ from fixed $Z^{\otimes n}$ basis measurements $\{Z_1^{(i)} Z_2^{(i)} \cdots Z_n^{(i)}\}_{i=1}^N$, we average over the measurement outcomes
\begin{equation}
    \ev{Z_1 Z_2}_{N,\,\text{Global}} = - \frac{1}{N} \sum_{i=1}^{N} Z_1^{(i)}Z_2^{(i)}.
\end{equation}
In the scheme of random-$XZ$ measurements, we use classical shadow tomography for such estimation. In particular, we can estimate any $k$-local observable $S_{Z,\, l}$ using the classical shadows $\hat{\rho}_N$ defined in \equref{eq:xz_shadow_Nshot} as
\begin{equation}
    \ev{S_{Z,\, l}}_{N,\,\text{Random}} = \tr(S_{Z,\, l} \hat{\rho}_N).
\end{equation}
As discussed in \appref{app:xz_shadow}, the number of samples required to estimate such observables accurately scales exponentially as $\mathcal{O}(2^k)$. Here, a key surprise of our protocol is that even though we only regularize with physically local observables $k \ll n$, we are able to see an improvement in global quantities as the infidelity is reduced.

\textit{Rydberg atom arrays.} More generally, if we do not know \textit{a priori} which operators are important---as is indeed the case for non-fixed-point states realizable in quantum simulators such as Rydberg atom arrays---we would need to consider all the operators that one can estimate using a finite number of samples. The intuition behind the regularization contribution is that capturing more physical observables accurately is consistent with improvement of the state fidelity. In fact, describing a subsystem with a size comparable to the correlation length guarantees an efficient learning of a one-dimensional injective quantum state using MPS tomography~\cite{cramerEfficientQuantumState2010, baumgratzScalableReconstructionDensity2013}. Here, using random-$XZ$ measurements, the space of the operators we can estimate is spanned by all combinations of $\mathds{1},X,Z$ strings normalized with respect to the Hilbert-Schmidt inner product $\text{\rm Pauli}_{\mathds{1},X,Z}(n)$.  This motivates us to define the \textit{projected density matrix} 
\begin{equation}\label{eq:projected_rdm}
    \rho^{}_{XZ} = \sum_{P \in \text{\rm Pauli}_{\mathds{1},X,Z}(n)}\tr(P \rho) P.
\end{equation}
Now, the classical shadows learnable using random-$XZ$ measurements 
[cf.\ \equref{eq:xz_shadow_m_inverse}] converge to this projection in the infinite-data limit, i.e.,
\begin{equation}
    \rho_{XZ} = \lim_{N \rightarrow \infty}\hat{\rho}^{}_m.
\end{equation}
In the finite-sample regime, we regularize on observables that are supported on subsystems $B$ of $6$ qubits, described by a classical shadow
\begin{equation}
    \hat{\rho}_{N}^B = \otimes^{}_{i\in B} \left(2 U^{a_i \dagger}  \dyad{b_i} U^{a_i} - \frac{\mathbb{I}}{2} \right).
\end{equation}
Such a projection can also be computed from our MPS model $\ket{\psi(A)}$. First, we trace out the complement $\Bar{B}$ 
\begin{equation}
    \tilde{\rho}^B(A) = \tr_{\Bar{B}}(\dyad{\psi(A)}).
\end{equation}
Next, we project the reduced density matrix onto  $\textnormal{\textsf{VisibleSpace}}(\mathcal{U}_{XZ})$ as
\begin{equation}
    \tilde{\rho}_{XZ}^B(A) = \sum_{P \in \text{\rm Pauli}_{\mathds{1},X,Z}(n)} \tr(P \tilde{\rho}^B(A)) P.
\end{equation}
Putting all the pieces together, we regularize over these projected density matrices---supported on all the unit cells that cover the lattice---using the second term in our loss function
\begin{align}\label{reg:projected_rdm_shcatten2norm}
            R(A) = \norm{\hat{\rho}_{N}^B -\tilde{\rho}_{XZ}^B(A)}_{2}.
\end{align}
Here, the matrix norm is given by the Schatten-$2$ norm (or Frobenius norm)
\begin{equation}
    \norm{O}_2 = \left[ \sum_{i,j} O_{i,j}^2 \right]^{\frac{1}{2}}.
\end{equation}

\newpage
\section{Dataset and state generation with DMRG}\label{app:dmrg_state}
In this appendix, we discuss the generation of our numerical dataset using Gibbs sampling from the target MPS. The target MPS is variationally optimized using the \emph{density-matrix renormalization group} (DMRG) algorithm. Here, we outline a snippet of  the numerical code written in the Python language.  The calculations are performed using an open source library quimb~\cite{grayQuimbPythonPackage2018}.
% \wwc{Explain Gibbs sampling}
All of the code used to generate our numerical results is available at: \href{https://github.com/teng10/tn-shadow-qst}{https://github.com/teng10/tn-shadow-qst}.

\subsection{Perturbed surface code}\label{app:state_surface_code}
In this section, we discuss the details of our surface code numerics and provide the relevant code snippets used to generate our dataset. The perturbed surface code Hamiltonian for a system of $n$ qubits is defined as
\begin{equation}
    H_{\text{sc}} = - \sum_{l}S_{Z,\, l} - \sum_{l}S_{X,\, l} - h_z \sum_{i}Z_{i}, \quad S_{Z,\, l}=\prod_{i\in\square_l} Z_i, \quad S_{X,\, l}=\prod_{i\in\square_l} X_i,
\end{equation}
where $\mathcal{S} = \{S_{Z,\, l}, S_{X,\, l}, \, l = 1, \cdots, \frac{n-1}{2} \}$ is the set of $n-1$ stabilizers around a square, as illustrated in \figref{fig:2_surface_code}(a). To construct the Hamiltonian for an $(L_x, L_y) = (3, 3)$ system, the following lines create an object called \textit{PhysicalSystem} and generate its matrix product operator representation.
\begin{lstlisting}[language=Python]
surface_code = physical_systems.SurfaceCode(3, 3)
surface_code_mpo = surface_code.get_ham()
\end{lstlisting}
To run the DMRG algorithm, first, a random MPS is selected as an initial state. Then, the MPO and this initial state are used to solve for an MPS approximation of the ground state with a given bond dimension. 
\begin{lstlisting}[language=Python]
import quimb.tensor as qtn
size = 9
bond_dim = 10
mps = qtn.MPS_rand_state(size, bond_dim)
dmrg = qtn.DMRG2(surface_code_mpo, bond_dims=bond_dim, p0=mps)
dmrg.solve()
result_mps = dmrg.state
\end{lstlisting}
Finally, to generate the dataset for training, we first select a measurement scheme (e.g., random-$XZ$ [see \appref{app:reality_random_xz_measurements}]), and specify the number of samples required.
\begin{lstlisting}[language=Python]
init_seed = 42 # for reproducibility
sampling_method = `x_or_z_basis_sampler` # random XZ measurements
num_samples = 1000
ds = run_data_generation._run_data_generation(init_seed, num_samples, sampling_method, mps)
\end{lstlisting}

\subsection{Ruby-lattice quantum spin liquid}
In this section, we include the details of our DMRG calculation for preparing ground states of the Rydberg Hamiltonian on a ruby lattice~\cite{semeghiniProbingTopologicalSpin2021}. Specifically, we consider a system of atoms arrayed on a ruby lattice with an aspect ratio of $\rho=\sqrt{3}$ [\supfigref{supfig:rydberg_cylinder}] and cylindrical boundary conditions. The Hamiltonian is given by
\begin{align}\label{eq:rydberg_ham}
    H_{\rm Ryd} = \frac{\Omega}{2} \sum_{\ell}X_{\ell} - \Delta \sum_{\ell} \frac{1 + Z_{\ell}}{2}  + \frac{1}{2}\sum_{\ell, \ell^\prime} \frac{V_{\ell, \ell^\prime}}{4} (1 + Z_{\ell})(1 + Z_{\ell^\prime}).
\end{align}
Denoting by $\ket{b} = \ket{\pm}$ the eigenstate of the Pauli $Z$ operator in the computational basis (i.e., $Z\ket{b} = b\ket{b}$), we work with a sign convention such that $b=+1$ ($-1$) corresponds to an atom being in the Rydberg (ground) state.

Depending on the ratio of the detuning to the Rabi frequency, $\Delta/\Omega$, different quantum phases can be prepared~\cite{samajdarQuantumPhasesRydberg2021,verresenPredictionToricCode2021}.
The native van der Waals interaction potential decays rapidly with the distance as $V_{\ell, \ell^\prime} = {V_0}/{\abs{r_{\ell} - r_{\ell^\prime}}^6}$. Choosing the interaction strength such that $\Omega \ll V_{\ell, \ell^\prime}$, the strong repulsion between neighboring atoms implements the Rydberg blockade effect, such that no two atoms are allowed to be simultaneously excited within a blockade radius of $R_b = ({V_0}/{\Omega})^{\frac{1}{6}}$. While, in principle, the interactions act between all pairs of atoms, in practice, an approximated potential obtained by truncating the long-range tails of the interaction is known to capture the essential physics of the spin-liquid phase~\cite{verresenPredictionToricCode2021}. In our DMRG simulations, we only retain interactions up to third-nearest neighbors ($r = 2a$, where $a$ is the lattice spacing), as illustrated in \figref{fig:3_rydberg_ruby}(a). Setting $\Omega=1$, we approximate the potential as a step function
\begin{equation}
\label{eq:int}
    V_{\ell, \ell^\prime} = \begin{cases}
        47 \Omega, & \quad \abs{r_{\ell} - r_{\ell^\prime}} \leq 2a, \\
        0, & \quad \text{otherwise}.
    \end{cases}
\end{equation}
Note that unlike the so-called PXP model~\cite{verresenPredictionToricCode2021}, in which the Hilbert space is constrained to forbid simultaneous excitations of neighboring atoms, here, the blockade is only imposed as an energetic penalty. The  interaction scale in Eq.~\eqref{eq:int} follows from the choice of the blockade radius, $R_b=3.8a$. Moreover, the Rabi frequency $\Omega$ is positive, so the present Hamiltonian is nonstoquastic~\cite{bravyiComplexityStoquasticFrustrationfree2008}, without guarantees of positive wavefunction amplitudes in the computational basis. 

In the present work, we focus on 
the ruby-lattice Rydberg Hamiltonian with periodic boundary conditions along the $y$-direction and open boundary conditions along the $x$-axis. We perform numerical simulations for a lattice of size $(L_x, L_y) = (4, 2)$; the unit cell has $6$ basis atoms, so the total number of qubits is $n=48$. To mitigate finite-size effects, we compensate for the fewer neighboring atoms of any site along the boundary by adding a local field $h_{\text{bd}}=-0.6 \Omega$ to the regions shaded by the blue boxes in \supfigref{supfig:rydberg_cylinder}:
\begin{equation}\label{eq:rydbgerm_ham_boundary}
    H_{\text{Ryd, bd}} = H_{\text{Ryd}} - h_{\text{bd}} \sum_{\ell \in {\partial}} \frac{1 + Z_{\ell}}{2}.
\end{equation}
\begin{lstlisting}[language=Python]
delta = 1.7 # detuning frequency for spin liquid state
boundary_z_field = -0.6 # adding an onsite field to mitigate boundary effects
rydberg = physical_systems.RubyRydbergPXP(4, 2)
rydberg_mpo = rydberg.get_ham(delta, boundary_z_field)
\end{lstlisting}

\begin{figure}[h!]
    \centering
    \includegraphics[width=0.6\textwidth]{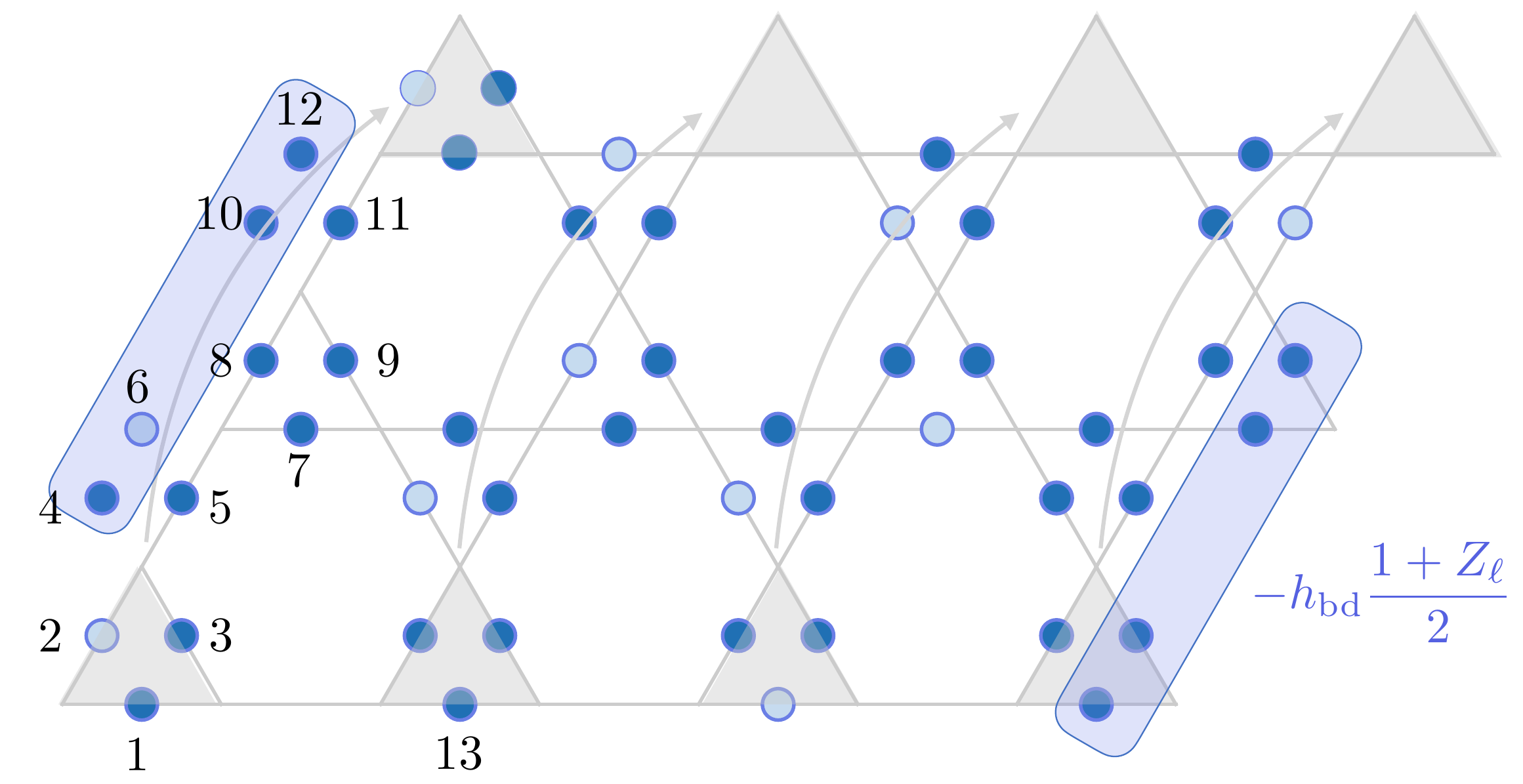}
    \caption{Rydberg atom array on a cylinder. The ruby lattice is placed on a cylinder, imposing periodic (open) boundary conditions along the $y$ ($x$)-axis. The numbers label the snake-like path for the tensors in the MPS representation of the ground state. An onsite field applied to the sites in the blue boxes compensates for the reduced coordination number of atoms situated along the open boundaries, thereby preventing them from getting pinned.}
    \label{supfig:rydberg_cylinder}
\end{figure}

\newpage
\section{Training details}\label{app:training}
In this appendix, we discuss our approach to optimizing \emph{matrix product state} (MPS) tensors through simultaneous updates of all tensor components by gradient descent. We then highlight the selection process for models likely to converge in the absence of direct access to the infidelity. 
% delve into employing classical shadow tomography for fidelity estimation between learned MPS models and laboratory states, alongside 
This can be done by partitioning datasets into training and test datasets and postselecting models that successfully generalize on the test dataset. This method's efficacy is illustrated through the analysis of training and test \emph{negative log-likelihood} (NLL) metrics. As evidenced in the surface code example, we evaluate the model performance's with sample size to ensure robust generalization.

\subsection{Optimization scheme via gradient descent}
We approach the optimization of \emph{matrix product state} (MPS) tensors by updating all tensors simultaneously. While the target states we are learning are real states, we use MPS tensors with complex variables because we have found that optimization over complex tensors is more robust, as also observed in Ref.\ ~\cite{gomezReconstructingQuantumStates2022}. The minimization of our loss function involves the logarithm of sums over expectation values of $U_i^\dagger \dyad{b_i} U_i$ defined by the measured unitary-bitstrings $D = \{(U_i, b_i)\}_{i=1}^N$,  
\begin{align}\label{eq:loss_function_mpo}
   L_{\text{NLL}}^D(A) &= -\frac{1}{N}\sum_{i=1}^{N} \log \left(\matrixel{\psi(A^*)}{U^{\dagger}_i}{b_i}\matrixel{b_i}{U_i}{\psi(A)} \right).
\end{align}   
For simplicity, we have considered the loss function without a regularization with $\beta=0$. While $U_i^\dagger \dyad{b_i} U_i$ can be represented as a \emph{matrix product operator} (MPO),  
this loss function is different from the traditional eigenvalue problem of finding the ground state of a local Hamiltonian typically considered in DMRG calculations.
Given the presence of the logarithmic term in our cost function, previous work has considered a ``DMRG-like'' sweeping algorithm~\cite{wangScalableQuantumTomography2020}. However, to simplify the training process, we have adopted global updates of MPS tensors over local sweeping updates and find that this strategy suffices for our purpose. Future improvements are possible by adopting optimization techniques on the MPS manifold~\cite{luchnikov2021riemannian, hauru2021riemannian} as discussed in \ref{app:mle-discussions}, particularly focusing on the projection of gauge degrees of freedom to refine tensor updates. 

Note that in order to minimize the loss function, we calculate gradients for all MPS tensors in the direction opposite to that dictated by the derivative of the complex conjugate \footnote{To see this, let $f(z, \bar{z}): \cc \rightarrow \rr$ be a real-valued function of variables $z=x + i y$ and $\bar{z}=x - i y$. Note that the loss function could be nonholomorphic. An update to minimize the loss function is given by $(x^{\prime} + i  y^{\prime}) = x - \delta \pdv{f}{x} + i (y - \delta \pdv{f}{x})$. This is the same as $z^{\prime} = z - \frac{\delta}{2} \pdv{f}{\bar{z}}$. Numerically, we compute $\pdv{L}{A}$ using JAX and take the direction given by its complex conjugate $(\pdv{L}{A})^*$.}
    \begin{align}
        A^{\prime} = A - \delta \frac{\dd L}{\dd \bar{A}} .
    \end{align}    
To integrate with machine-learning optimization techniques, our procedure employs a two-step strategy: initial descent towards an approximate global minimum using first-order \emph{stochastic gradient descent} (SGD), followed by refinement through a second-order quasi-Newton's method, the \emph{limited-memory Broyden–Fletcher–Goldfarb–Shanno} (L-BFGS) algorithm. This combination allows for efficient preliminary optimization followed by detailed adjustments.

%%%%%%%%%%%%%%%%%%%%%%%%%%%%%%%%%%%%

\subsection{Model selection}\label{app:training_schemes}
%%%%%%%%%%%%%%%%%%%%%%%%%%%%%%%%%%%%
In applying our protocol to a laboratory  state, in general, we would not have access to the infidelity. Here, we outline the steps to select models that are likely to converge without access to the fidelity. In the context of machine learning, a common approach is to divide the full dataset into the training and test datasets. First, we train the MPS model on only the training dataset partition and then evaluate the \emph{negative log-likelihood} (NLL) in \equref{eq:loss_finite_bitstring} using the test dataset. For the model to not be ``overfitting'' the training dataset, we would want both the training and test NLLs to be small. Specifically, we consider the surface code numerics as an example below. In \supfigref{fig:appedix_model_selection}, we plot both the training NLL and the test NLL for different numbers of samples. We observe that for a small number of samples, while we achieve a lower training NLL, the test NLL is higher, indicating overfitting of the MPS to the training measurements. On the other hand, an MPS trained with a large number of samples achieves a lower infidelity and hence, better generalization. 

A different approach is to estimate the fidelity of the learned MPS model with the laboratory state~\cite{wangScalableQuantumTomography2020} by using classical shadow tomography with global random unitary ensembles, which could be hard to implement in practice. Finally, we remark that a recent protocol would also allow us to efficiently certify the fidelity using single-qubit measurements given a query model~\cite{huang2024certifying}, in our case, an MPS.
\begin{figure}[h!]
    \centering
    \includegraphics[width=\textwidth]{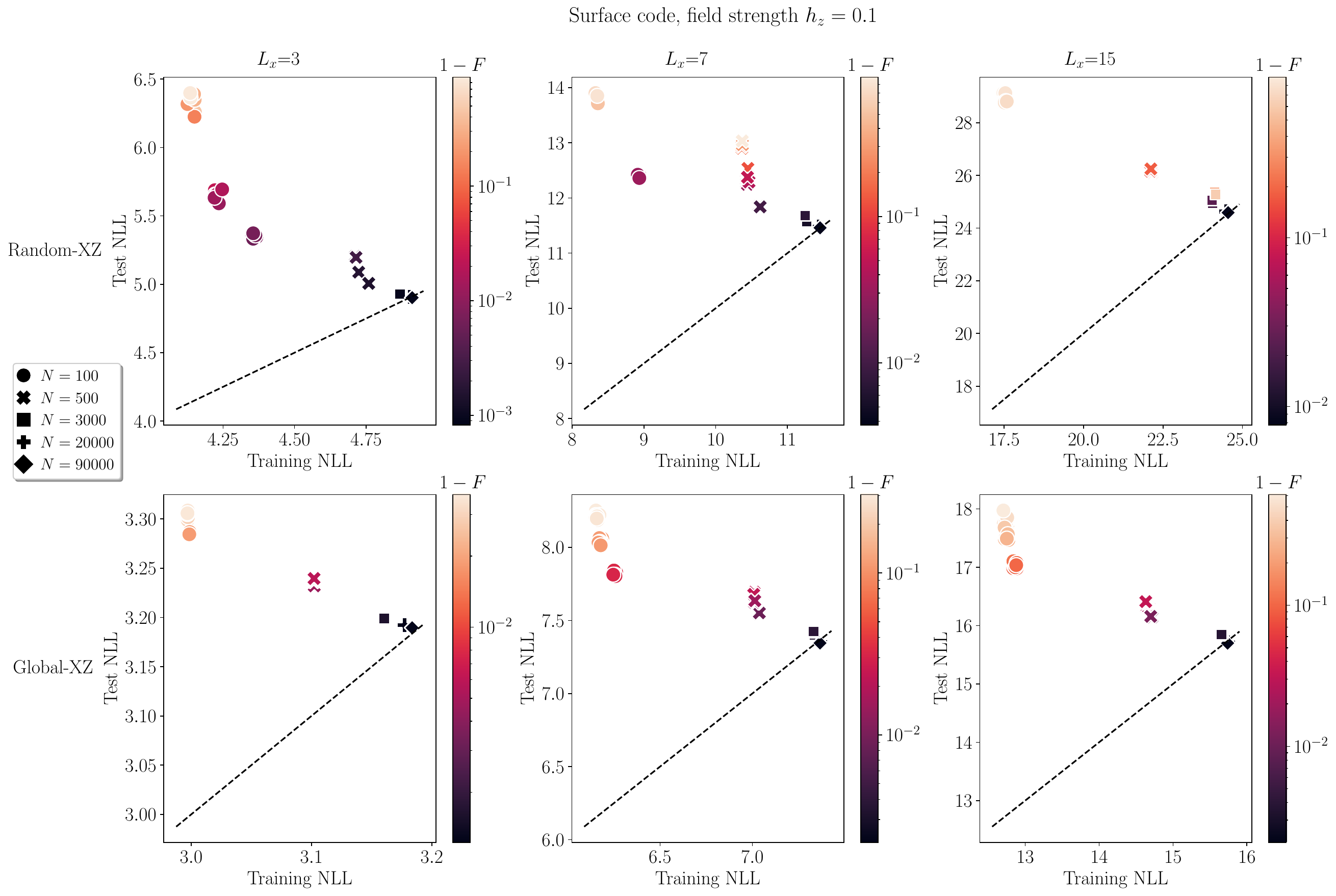}
    \caption{Model selection by comparing training and test losses for a surface code with different system sizes and measurement schemes. On the dashed lines, the training NLL $=$ test NLL; ideally, the model should be selected along the dashed line \textit{and} with a small training NLL. The models which have small NLLs and are close to the dashed line yield high fidelities (darker colors).}
    \label{fig:appedix_model_selection}
\end{figure}

\newpage

% \color{red}
\section{Measurement noise}\label{app:noise}
In this section, we discuss the effect of measurement noise in our tomographic protocol. As we have emphasized in the main text, the motivations for choosing our random-$XZ$ measurements are twofold: 1) such measurements are sufficient and tomographically complete for the real pure states of interests; 2) these measurements only require single-qubit rotations before measurements, suggesting their experimental feasibility. Nevertheless, there are imperfections in realistic measurements even in the most well-developed systems. A rigorous treatment of such errors is left for future work. Here, we numerically study how the fidelity of the reconstructed state is affected by measurement noise.

Rydberg atom arrays have seen significant recent progress, achieving impressive gate fidelities. In digital simulations, qubits are typically encoded in the hyperfine states $\{\ket{0}, \ket{1}\}$ of $^{87}$Rb atoms. In this setting, the most significant errors originate from many-qubit gates such as the controlled-phase entangling gate. 
Currently, the state-of-the-art two-qubit entangling gate fidelity is $99.5\%$ and single-qubit gate fidelity reaches $99.9\%$~\cite{evered2023high}.

In contrast, the states studied in this work are prepared through analog evolution, where interactions are mediated by the strong repulsive forces between atoms in the Rydberg state $\ket{r}$ in a two-level system of $\{\ket{1}, \ket{r}\}$. Since the Rydberg state is short-lived and hard to measure directly, a ``coherent mapping'' between the Rydberg and hyperfine manifolds is necessary. The mapping can be achieved by first mapping the hyperfine excited state to the ground state $\ket{1} \rightarrow \ket{0}$ via a Raman $\pi$ pulse, and then mapping the Rydberg state to the hyperfine excited state $\ket{r} \rightarrow \ket{1}$ via a Rydberg $\pi$ pulse. Beyond the imperfections of the pulses (which could be neglected), there are additional sources of errors from the coherent mapping. Such a protocol is analogous to a hybrid simulation presented in Fig.~4 of Ref.~\cite{bluvsteinQuantumProcessorBased2022}, wherein a state is prepared via analog simulation and thereafter manipulated by digital evolution. The main sources 
of errors during such a coherent mapping are~\cite{bluvsteinQuantumProcessorBased2022} the following: 
\begin{enumerate}
    \item \textit{Failure of the Rydberg $\pi$ pulse.} With a small probability, two neighboring $^{87}$Rb atoms both occupy the Rydberg state $\ket{r}$, and such a violation of the Rydberg blockade would dramatically shift the energy level. When this occurs, the pair of Rydberg states cannot be converted to the hyperfine states, leading to an incomplete sample.
    \item \textit{Off-resonance of Rydberg $\pi$ pulse.} Even without such a violation, the long-range interaction between the Rydberg states would slightly detune the Rydberg $\pi$ pulse, leading to a different phase factor in the mapped state.
    \item \textit{Dephasing during Raman $\pi$ pulse (Doppler shifts).} When the atoms are idling, their atomic motions lead to off-resonance of the Raman pulse, called the Doppler shift. As the Raman pulse has a small frequency, and is therefore slow compared to the Rydberg pulse, the atoms can accumulate a more substantial random phase. 
\end{enumerate}

In summary, because a single-qubit gate of an analog-prepared state requires coherent mappings between the Rydberg state and the hyperfine state, additional errors during such mappings contribute to the errors in these single-qubit gates. It is therefore important to carefully consider these measurement errors in future work. As a first step, we consider a simple bit-flip noise model where $\ket{0}$ and $\ket{1}$ are misclassified with each other with a probability of $\xi$. We compare the fidelity achieved from training with a noisy dataset with probability $\xi$ for the $3\times 3$ surface code state in \supfigref{supfig:noise}. As expected, higher noise levels lead to a reduced fidelity. The infidelity \wwc{scales continuously with $N$---even in the presence of noise---}at an initial stage of training, and exhibits a plateau with a different scaling rate after a certain threshold number of measurements. Such a scaling is worth exploring systematically in future work.

\begin{figure}[b]
    \centering
    \includegraphics[width=0.32\linewidth]{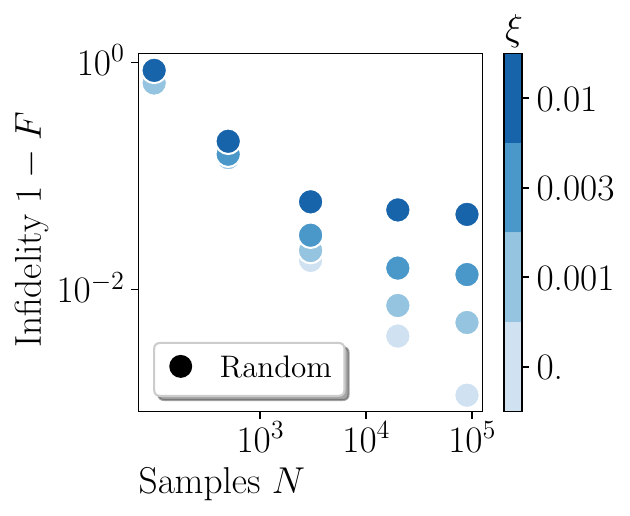}
    \caption{\wwc{Robustness of the tomographic protocol to measurement noise.
We study the effect of measurement noise modeled by a symmetric bit-flip channel, where measurement outcomes $\ket{0}$ and $\ket{1}$ are flipped with probability $\xi$. The plot shows the infidelity $1-F$ of the reconstructed $3 \times 3$ surface code state as a function of the number of samples $N$, for varying noise rate $\xi$. Results are shown for the random-$XZ$ measurements.}}
    \label{supfig:noise}
\end{figure}
\end{appendix}
\end{document}